\documentclass[10pt]{IEEEtran}	
\usepackage{mathrsfs}
\usepackage{subfigure}
\usepackage{amsfonts}
\usepackage{amsmath}
\usepackage{amsbsy}
\usepackage{graphicx}
\usepackage{subfigure}
\usepackage{amssymb}
\usepackage{hyperref}
\newcommand{\set}[1]{\mathscr{#1}}
\newcommand{\net}{\mathscr{N}}
\newcommand{\onet}{{\mathscr{\overline{N}}}}
\newcommand{\head}[1]{h_{({#1})}}
\newcommand{\tail}[1]{t_{({#1})}}
\newcommand{\comment}[1]{}

\newcommand{\mN}{\mathscr{N}}
\newcommand{\negspace}{\hspace*{-2mm}}
\begin{document}
\onecolumn
\title{On the solvability of 3-source 3-terminal sum-networks}

\author{
\authorblockN{Sagar Shenvi and Bikash Kumar Dey} \negspace\\
\authorblockA{Department of Electrical Engineering \negspace \\
Indian Institute of Technology Bombay\negspace \\
Mumbai, India, 400 076\\
\{sagars,bikash\}@ee.iitb.ac.in}
}
\maketitle               
\setlength{\baselineskip}{1.5\baselineskip}
\begin{abstract}
We consider a directed acyclic network with three sources and three terminals
such that each source independently generates
one symbol from a given field $F$ and each terminal 
wants to receive the sum (over $F$) of the source symbols.
Each link in the network is considered
to be error-free and
delay-free and can carry one symbol from the field in each use.
We call such a network a 3-source 3-terminal {\it $(3s/3t)$ sum-network}.
In this paper, we give a necessary and sufficient
condition for a $3s/3t$ sum-network to
allow all the terminals to receive the sum of the source symbols
over \textit{any} field. Some lemmas provide interesting simpler
sufficient conditions for the same.
We show that linear codes are sufficient for this problem
for $3s/3t$ though they are known to be insufficient for arbitrary number
of sources and terminals. We further show that in most cases, such networks
are solvable by simple XOR coding.
We also prove a recent conjecture that if fractional coding is allowed,
then the coding capacity of a $3s/3t$ sum-network is either $0,2/3$ or $\geq 1$.
\end{abstract}
\begin{keywords}
Network coding, function computation, multicast, multiple unicast
\end{keywords}

{\renewcommand{\thefootnote}{}
\footnotetext{A part of this work is accepted for presentation at IEEE International Symposium on Information Theory
2010, Austin, Texas.}}

\newtheorem{theorem}{Theorem}
\newtheorem{corollary}{Corollary}
\newtheorem{lemma}{Lemma}
\newtheorem{claim}{Claim}
\newtheorem{definition}{Definition}
\newtheorem{sub-claim}{Sub-claim}
\newtheorem{observation}{Observation}

\section{Introduction}
It was shown by Ahlswede et. al.~\cite{ahlswede1}
that mixing/coding of incoming information at the intermediate
nodes, called network coding, could result in throughput advantages.
In particular, it was shown that the
coding capacity of a directed multicast network is equal to the
minimum of the min-cuts
between the source and the individual terminals.
Further, linear network coding
was shown to be sufficient to achieve this capacity \cite{li1,koetter2}.
A polynomial time algorithm for linear multicast code construction
was given in \cite{jaggi1}, whereas distributed random network
codes were shown to achieve capacity for multicast networks
in \cite{ho1}. Network coding has since evolved into
a rich field of study with connection to many other areas~\cite{yeung1}.

In this paper, we consider the problem of communicating the sum of
messages at some sources to a set of terminals in a directed
acyclic network of unit-capacity edges. The problem is a subclass
of the problem of distributed computation over a network.
Due to the immense complexity of the problem in its full generality
with all its model-variations, the problem has been studied in
various simplified forms by researchers from diverse fields.
We list some known approaches to the problem below.

\begin{enumerate}
\item {\bf Simple and small networks:} Early work in the area of information
theory considered the distributed function computation problem 
as a generalization to the Slepian-Wolf problem. Here the network
has multiple sources with separate encoders connected to a
receiver which wants to compute a function of the symbols generated
at the sources, possibly with a limited 
allowed distortion~\cite{korner1,han1,krithivasan1,doshi1}. Another variation
is where the receiver has access to correlated side-information,
and it wants to compute a function of the source symbol encoded
by an encoder and the side-information~\cite{orlitsky1,feng1}.
There are two features in this approach which
make the problem complex. First, the sources are correlated with
a known arbitrary joint distribution. Second, the aim is
to compute the region of encoding rates which allow the recovery
of the function at the receiver under the allowed distortion.

\item {\bf Large networks:} Gallager~\cite{gallager2} first posed
the problem of computing the parity or modulo-2 sum of a large
number of binary sources in a broadcast network. Here all the
nodes are independent sources and terminals. The question addressed is how
the number of required communications scale with the number of nodes.
This line of work became more popular in the context of wireless networks
with scheduling constraints arising due to
interference~\cite{giridhar,kanoria,boyd}. A large body of work
now exists with many variations in various aspects.

\item {\bf Distributed Detection:} The problem of distributed
detection/estimation of some underlying parameter from the
signals generated at different nodes is a problem
which gained renewed impetus from the widespread interest in sensor
networks~\cite{varshney1,chair1,hu1}. The main aim is to find optimum
or near-optimum algorithms. This problem also has the essence of
a distributed function computation problem.

\item {\bf Network Coding:} Before network coding acquired
its recent level of maturity, the problem of distributed function
computation was addressed in a rigorous way by information theorists
only for small or simple networks. The techniques of network coding
were used in some recent efforts to get some results of elementary
nature for larger networks~\cite{ramamoorthy,RaiD:09a,RaiD:09b,langberg3,RaiDS:itw10,appuswamy1,appuswamy2,appuswamy3}. Our present work is along this line.
\end{enumerate}

In the first, second, and the fourth category of approaches listed
above, the particular function ``sum'' received special interest~\cite{korner1,
krithivasan1,gallager2,ramamoorthy,RaiD:09a,RaiD:09b,langberg3,RaiDS:itw10}
because (i) it is a simple illustrative example function which
is easier to work with, (ii) it
reveals many interesting intricacies of the general problem, and
(iii) it may reveal techniques for addressing the problem
for more general functions (for example, \cite{krithivasan1} makes direct use
of encoding for linear functions for other functions) or other
network coding problems (the equivalence with other network coding
problems is shown in \cite{RaiD:09c}). In particular, linear multicast coding
and linear coding for computing ``sum'' at one terminal are equivalent problems
~\cite{RaiD:09b,RaiD:09c, koetter3}.
Both ``modulo sum'' as in a finite field, and more generally a finite
abelian group; and ``arithmetic sum'' as in a characteristic-$0$ field
are of interest. We consider the function ``modulo-sum'' in this paper.
Arithmetic sum, though important for many practical applications, is
more difficult to analyze because of the unbounded alphabet size.
However, the techniques for modulo-sum has also been found useful
for getting bounds for the capacity of computing
arithmetic-sum~\cite{appuswamy3}.

We
consider networks where there are multiple sources which generate
independent i.i.d. random processes over an alphabet finite field $F$,
or more generally an abelian group $G$. The edges are assumed to carry
one symbol from the alphabet per use without delay or error, i.e.,
they are delay-free, error-free and unit-capacity.
There are multiple terminals which want to recover the sum of
the source symbols in each symbol-interval.
We specifically consider the case of $3$ sources and $3$ terminals.
This has been known to be the first, i.e. with the smallest number
of sources and terminals, nontrivial and highly intriguing
case for
sum-networks~\cite{ramamoorthy,RaiD:09a,RaiD:09b,langberg3,RaiDS:itw10}.
 
\subsection{Standard definitions and review of known results}

We first define some standard terms which will be used in this
paper. Our network is represented by a directed acyclic multigraph
$\mathscr{N}=(V,E)$. A network with source nodes
$\{s_1,s_2,s_3,\ldots ,s_l\} \subset V$ and terminal nodes
$\{t_1,t_2,t_3,\ldots ,t_j\}\subset V$, so that each
terminal wants to recover $\sum_{i=1}^{l}{x_{it}}$ for every $t$, where
$(x_{it})_t$ is the source process of the $i$-th source, is
called a sum-network with $l$ sources and $j$ terminals. A 3-source
3-terminal sum-network will be called a $3s/3t$ sum-network in short.

\begin{definition}
A sum-network where there is a path from every source to every terminal will be called
a {\it connected sum-network}.
\end{definition}

\begin{definition}
For a sum-network or a multiple-unicast network~\cite{dougherty3}, the {\it reverse
network}~\cite{dougherty3,RaiD:09b,RaiD:09c,riis1} is the network obtained by reversing the direction
of every edge, and interchanging the roles of the sources and the terminals.
\end{definition}

\begin{definition}
If one can satisfy the demands 
of all the terminals over a finite field $F$ using each edge of
$\set{N}$ once, we say that $\mathscr{N}$ is
{\it solvable over $F$}. In particular, for a sum-network, it means that
all the terminals can recover one sum (for one $t$) by using the network
once. If a linear network code over $F$ is sufficient for this purpose,
we say that $\mathscr{N}$ is {\it linearly solvable over $F$}.
We say that $\mathscr{N}$ is solvable if it is solvable
over at least one field. If $\mathscr{N}$ is not solvable over any field, we say that
$\mathscr{N}$ is non-solvable. In terms of another well-known term,
solvability here refers to {\it scalar solvability}, i.e., solvability using
a {\it scalar network code}~\cite{yeung1,ho2}.
\end{definition}
Clearly, for solvability of a sum-network,
it is necessary that every source-terminal pair is connected.
For a single source, the sum-network reduces to the well-investigated
multicast network, and the source-terminal connectivity is also
a sufficient condition for solvability if the edges are unit-capacity.

We now define a simple form of linear network code. 
\begin{definition}
A scalar linear network code is called an {\it XOR network code} if
all the nodes in the network, including the terminal nodes,
require to perform only addition and subtraction. In other words,
all the local coding coefficients~\cite{yeung1,ho2} are $\pm 1$. For the binary
alphabet, this means that the nodes only need to perform XOR operation.
A network
which is solvable by a XOR network code is said to be {\it XOR solvable}.
\end{definition}
Such a network code is computationally much simpler. Further note that,
if a sum-network is XOR solvable then only the group structure in
the alphabet is relevant, and the multiplicative structure in
the alphabet field is not relevant. Though for simplicity, we will
restrict to a finite field alphabet from now onward, it can be
checked that whenever a network is XOR solvable over all fields, it
is also XOR solvable over any abelian group.

\begin{definition}
A $(k,n)$ {\it fractional network code} is a network code where
the source processes are blocked into packets of length $k$,
and encoded into vectors/packets of length $n$. The edges carry
$n$-length vectors and nodes operate on incoming $n$-length vectors
to construct $n$-length message vectors on outgoing edges. The terminals
recover their demanded function (specifically their demanded
source symbols for a traditional communication problem)
for $k$ consecutive symbols of the sources. Thus the rate
of computing/communication achieved by using such a code is $k/n$ per use
of the network. Such a code can be linear or non-linear.
\end{definition}

For example, a $(k,n)$ fractional network code for a sum-network
will enable the terminals to recover $\sum_{i=1}^l x_{it}$
for $t=k\tau,k\tau +1,\ldots, k\tau +k-1$ using the links of the network
$n$ times. Here $\tau$ denotes the block index.

\begin{definition}
The rate $r$ is
said to be {\it achievable} if there exists a $(k,n)$
fractional (possibly non-linear) network code such that $k/n \geq r$.
\end{definition}

\begin{definition}
The supremum
of all achievable rates is
called the {\it capacity} of the network. Clearly the capacity
of a solvable sum-network is $\geq1$.
\end{definition}
It can be easily argued that
the minimum of the min-cuts for all source-terminal pairs is an upper bound
on the capacity of a sum-network~\cite{RaiDS:itw10}.

For the most part of the paper, we will consider the question of 
solvability of a sum-network, and so will consider a single
symbol interval and a single usage of the network. So,
we will omit the index $t$ in $x_{it}$ and use $x_i$ to mean
the symbol generated by the $i$th source in one representative symbol-interval.
\begin{figure}[h]
\centering
\subfigure[]{
\includegraphics[height=1.5in]{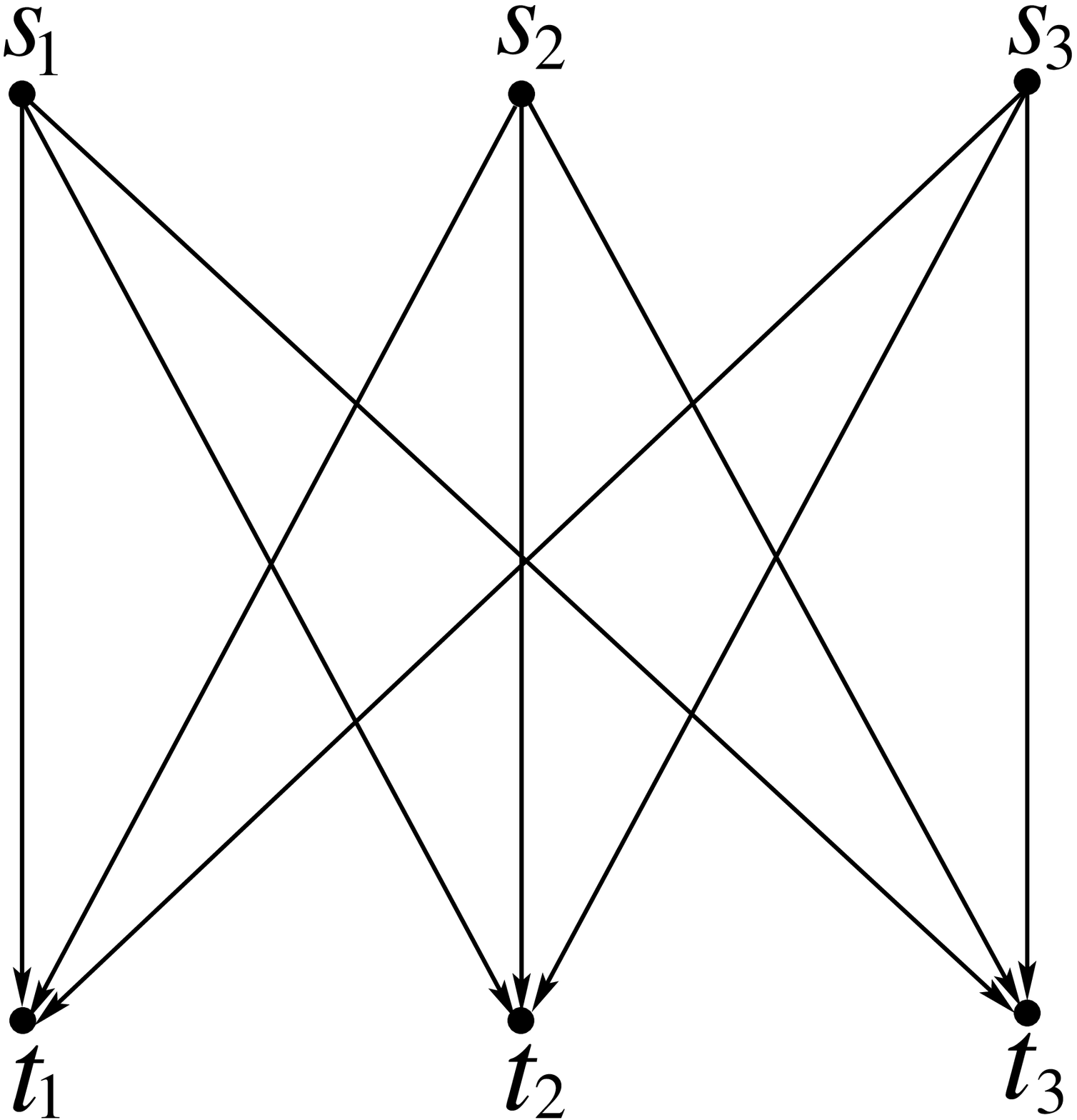}
\label{fig:eg1}
}
\subfigure[]{
\includegraphics[height=1.5in]{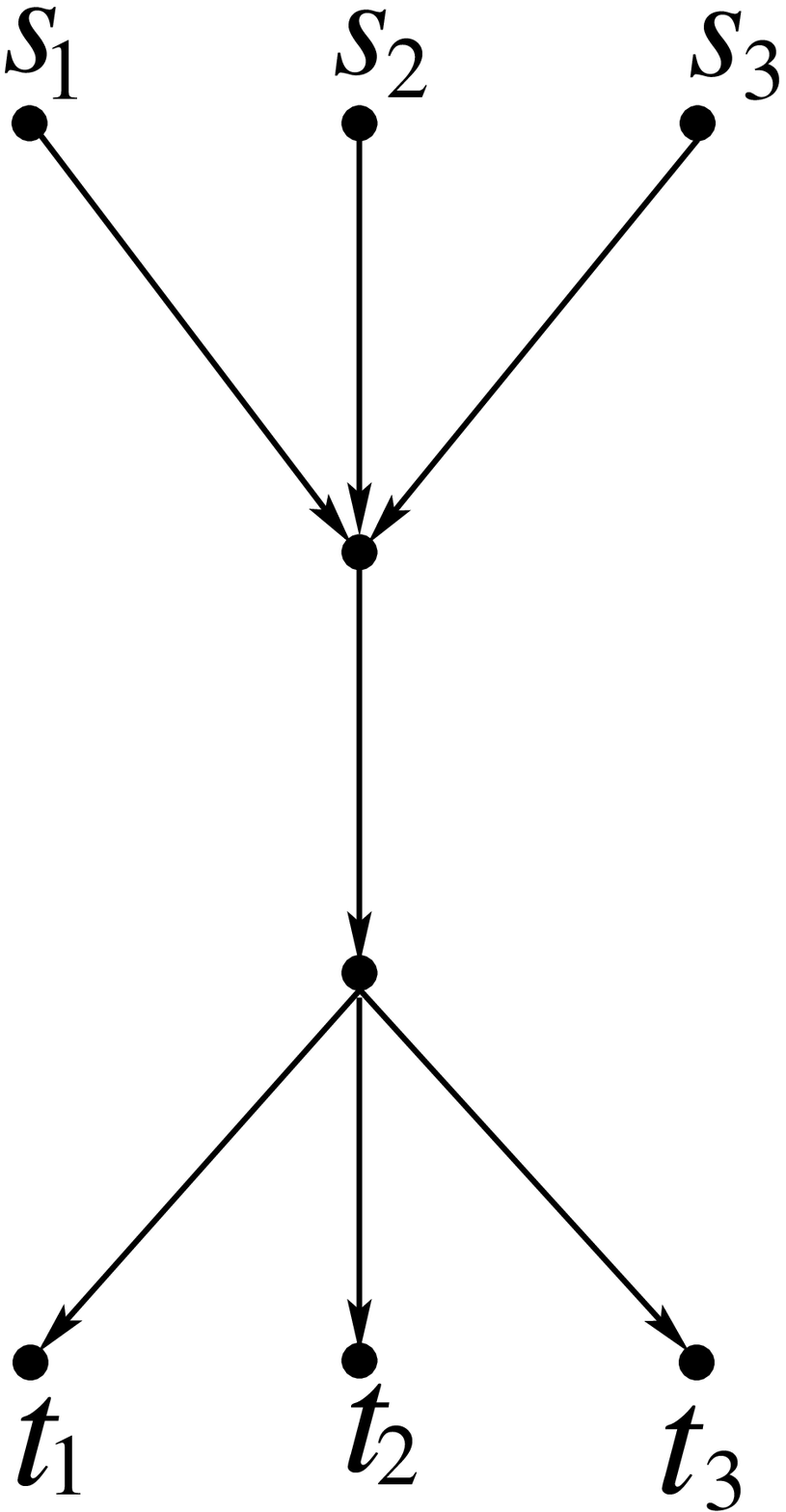}
\label{fig:eg2}
}
\caption[]{Examples of $3s/3t$ sum-networks which are solvable but where
source-terminal pairs are not two-connected}
\label{fig:twoconnected}
\end{figure}

In the following, we list some results known till date which are related
to our present work. 

$\bullet$ Ramamoorthy~\cite{ramamoorthy} showed that when there
are at most two sources or at most two terminals, a sum-network
is solvable over any field if and only if every source-terminal
pair is connected. Their algorithm also used an XOR code as per our
definition.

$\bullet$ The source-terminal connectivity condition is known to be insufficient
when both the number of sources and the number of terminals are
more than two~\cite{RaiD:09a}. In particular, a $3s/3t$ sum-network
(see Fig.~\ref{fig:s3}) was presented in \cite{RaiD:09a, langberg3} which
is not solvable. Further, it was proved in \cite{RaiDS:itw10} that
the capacity of this network is $2/3$. On examination of a variety
of $3s/3t$ sum-networks, it was conjectured in \cite{RaiDS:itw10} that the capacity of
any non-solvable but connected sum-network is $2/3$. This conjecture
is proved in this paper.

$\bullet$ It was proved in \cite{langberg3} that a $3s/3t$ sum-network
which is two-connected, i.e. has two edge-disjoint paths from every
source to every terminal, is solvable over fields of odd characteristic.
This condition is clearly not a necessary condition for solvability.
For instance, the sum-networks shown in Fig.~\ref{fig:twoconnected} do not satisfy the
condition but are clearly solvable.

$\bullet$ It was shown in  \cite{RaiD:09b} that a sum-network
has a $(k,n)$ fractional linear code if and only if the reverse network
has a $(k,n)$ fractional linear code. This implies that the linear coding
capacity of a sum-network is the same as that of its reverse network~\cite{RaiDS:itw10}. 
Since linear codes achieve capacity of a multicast network, this gives
that the capacity of a one-terminal sum-network is the minimum of the
min-cuts between the source-terminal pairs~\cite{RaiDS:itw10}.

$\bullet$ The problem of communicating the sum was shown to be equivalent
to the problem of multiple unicasts and more generally the arbitrary
network communication problem by showing explicit constructions in
\cite{RaiD:09c}. This implied several interesting consequences like
(i) existence of a solvably equivalent sum-network for every system
of integer polynomial equations, (ii) unachievability of capacity
of some sum-networks, and (iii) insufficiency of linear network coding
for sum-networks.

$\bullet$ The communication of more general functions was considered
in \cite{appuswamy1,appuswamy2,appuswamy3} over networks with one terminal.
Specifically, some cut-based bounds on the capacity of such networks
are presented in \cite{appuswamy3}.

\subsection{Our contribution}
We assume that the sources generate symbols from
a field $F$, the edges can carry one symbol from $F$ per use without error and
delay, and the terminals want to recover the sum (defined in $F$)
of the source symbols. The contribution of this paper is the following.

1. We find a set of necessary and sufficient conditions for
solvability of a $3$-source $3$-terminal sum-network over any
field $F$ (Theorems~\ref{theorem:mainres} and ~\ref{theorem:mainres2}).

2. We prove a conjecture made in \cite{RaiDS:itw10} that the capacity
of any non-solvable connected $3s/3t$ sum-network is $2/3$.

3. The proof of the necessary and sufficient conditions also lead us
to some interesting results and insights like sufficiency of linear
codes.
We also identify a significant class of solvable networks ($\kappa \neq 2,3$
in Lemma~\ref{lemma:suffcon}) which are XOR solvable over any field.
In particular, it implies that networks with $\kappa=0$ (equivalently,
where every source-terminal pair is two-connected) are {\it XOR solvable over
any field}, thus significantly strengthening the result of ~\cite{langberg3}.
In contrast, it was shown in~\cite{RaiD:09c} that linear codes are not sufficient
in general for sum-networks with arbitrary number of sources and terminals.

4. As intermediate results, we prove some lemmas which give simpler
sufficient conditions for solvability of a $3s/3t$ sum-network.

The paper is organized as follows.
In Section~\ref{section:notation}, we introduce some notations and
define some new terminology which will be used in this paper.
We present our new results in Section~\ref{section:main} and
prove them in Section~\ref{section:proofs}. The paper is concluded
in Section~\ref{section:conclusion}.

\section{Notations and new definitions}\label{section:notation}
Recall that our network is represented by a directed acyclic multigraph
$\mathscr{N}=(V,E)$ with
source nodes $\{s_1,s_2,s_3,\ldots ,s_l\} \subset V$ and terminal nodes
$\{t_1,t_2,t_3,\ldots ,t_j\}\subset V$. Each source node $s_i$ independently
generates a symbol sequence $x_{it}$ from the alphabet finite field $F$
and each terminal wants to recover $\sum_{i=1}^{l}{x_{it}}$ 
defined over $F$ for every $t$.
Each edge represents an error-free, delay-free link of unit-capacity.
We specifically consider a $3s/3t$ sum-network.
As the sum of sources can not be communicated
to the terminals at any non-zero rate if a network is not connected,
we consider only connected networks in this paper.
 
For any edge $e=(v_i,v_j)\in E$, the node $v_j$ is called its head
and the node $v_i$ its tail and are denoted as $\head{e}$ and $\tail{e}$
respectively. A path $P$ from $v_1$ to $v_l$ - also called a $(v_1,v_l)$
path - is a sequence
of nodes $v_1,v_2,\ldots, v_l$
and edges $e_1,e_2,\ldots, e_{l-1}$ such that $v_i = \tail{e_i}$
and $v_{i+1} = \head{e_i}$ for $1\leq i \leq l-1$.
For any path $P$, $P(v_j:v_k)$ denotes
its section starting from the node $v_j$ and ending at $v_k$.
If $P_1$ is a $(v_i,v_j)$ path and $P_2$ is a $(v_j,v_k)$ path,
then $P_1P_2$ denotes the $(v_i,v_k)$ path
obtained by concatenating $P_1$ and $P_2$.
\begin{definition}
For any $A,B \subset V, A\cap B = \emptyset$,
we write $A \rightarrow B$ if there is a path from every
node in $A$ to every node in $B$, and we write
$A \nrightarrow B$ if there is no path from any node in
$A$ to any node in $B$. Note that $\nrightarrow$ is not the negation
of $\rightarrow$. If $A=\{v_i\}$ and $B=\{v_j\}$
are singletons, we simply write $v_i \rightarrow v_j$ and
$v_i \nrightarrow v_j$. For any edges $e_1,e_2 \in E$,
we write $e_1 \rightarrow e_2$,
$e_1 \rightarrow v_j$ and $v_i \rightarrow e_2$
to mean respectively $\head{e_1} \rightarrow \tail{e_2}$,
$\head{e_1} \rightarrow v_j$
and $v_i \rightarrow \tail{e_2}$.
If for two nodes $m$ and $n$, $m\rightarrow{n}$,
$m$ is called an \textit{ancestor}
of $n$, and $n$ a \textit{descendant} of $m$.
We assume that a node is not its own ancestor or descendant.
For any
$A,B \subset V, A\cap B = \emptyset$, we define
$\Gamma^{A}_{B} = \{v\in V : A \rightarrow v, v \rightarrow B\}$,
$\Gamma^{A} = \{v\in V : A \rightarrow v\}$, $\Gamma_{B} = \{v\in V : v \rightarrow B\}$
and $mincut (A,B)$ to be the least number of edges whose removal
causes $A \nrightarrow B$ in the remaining network.
\end{definition}

We represent the network formed by removing the edges $\{e_1,e_2,...,e_i\}$
from the original network $\mathscr{N}$ by $\{\mathscr{N}-\{e_1,e_2,...,e_i\}\}$.
An edge $e$ is said to \textit{disconnect} an ordered pair of nodes
${(v_i,v_j)}$, if $v_i \rightarrow v_j$ in $\set{N}$ but $v_i \nrightarrow v_j$
in $\{\mathscr{N}-\{e\}\}$.

\begin{definition}\label{definition:kappa}
For a connected sum-network $\mathscr{N}$ the maximum number of source-terminal
pairs that can be disconnected by removing a single edge is called the
\textit{maximum-disconnectivity} of the network and denoted by
$\kappa(\mathscr{N})$. We call any edge whose removal disconnects
$\kappa(\mathscr{N})$ source-terminal pairs as a \textit{maximum-disconnecting} edge.
All edges are maximum-disconnecting
edges if $\kappa(\set{N})=0$.
\end{definition}

For example, if in a $3s/3t$ sum-network every source-terminal
pair is two-connected then removing any single edge can not disconnect
any source-terminal pair; and so the network has $\kappa = 0$.
On the other hand, the network shown in Fig.~\ref{fig:eg2} has a single
bottleneck link whose removal disconnects all the source-terminal
pairs; and so the network has $\kappa = 9$.

We classify the set of all maximum-disconnecting
edges into the following three sets:
(Recall that all edges are maximum-disconnecting
edges if $\kappa(\set{N})=0$.)

$\mathscr{A}:$ the set of all maximum-disconnecting edges
such that there is a path from its head to only one terminal.

$\mathscr{B}:$ the set of all maximum-disconnecting edges
such that there is a path from only one source to its tail.

$\mathscr{C}:$ the set of all maximum-disconnecting edges
such that there is a path from at least two sources to
its tail and to at least two terminals from its head.

Clearly every maximum-disconnecting edge is
in at least one of $\mathscr{A},$ $\mathscr{B}$ and $\mathscr{C}$.
Also, $\mathscr{C}$ is disjoint from $\mathscr{A}$ and $\mathscr{B}$.
If $\kappa(\set{N})=0$ or $1$, a maximum-disconnecting edge
may belong to both $\mathscr{A}$ and $\mathscr{B}$, however
$\mathscr{A}$ and $\mathscr{B}$ are disjoint if $\kappa(\set{N})\geq 2$
because then any maximum-disconnecting edge is connected to either
at least two sources or at least two terminals.

\section{Results}\label{section:main}
In this section, first we present our main results as theorems, and then
we present some lemmas which, on one hand, are used to prove the theorems
and which, on the other hand, also provide simpler sufficient conditions
for solvability.
Recall that a sum-network is nonsolvable if it is not solvable over any field.
\begin{theorem}\label{theorem:mainres}[{\bf{Necessary and Sufficient condition for Solvability}}]
A.
A $3s/3t$ connected sum-network $\mathscr{N}$ is nonsolvable if and
only if  there exist two edges $e_1$ and $e_2$ and some labeling of
the sources and the terminals such that
\begin{enumerate}
\item $mincut({\{s_1\},\{t_3\}})= 0$ in $\{\set{N}-\{e_1\}\}$
\item $mincut({\{s_3\},\{t_1\}}) = 0$ in $\{\set{N}-\{e_1\}\}$
\item $mincut({\{s_2\},\{t_3\}}) = 0$ in $\{\set{N}-\{e_2\}\}$
\item $mincut({\{s_2,s_3\},\{t_2\}}) = 0$ in $\{\set{N}-\{e_2\}\}$
\item $mincut({\{s_3\},\{t_3\}}) = 0$ in $\{\set{N}-\{e_1,e_2\}\}$
\item ${e_1}\nrightarrow{e_2}$ and ${e_2}\nrightarrow{e_1}$
\end{enumerate}

B. Whenever a network is solvable, it is linearly solvable over all fields
except possibly $F_2$.

C. Whenever a network is solvable over $F_2$, it is XOR solvable over
any field.

D. \cite[Conjecture 7]{RaiDS:itw10} The capacity of a connected non-solvable network is $2/3$.
\end{theorem}

Fig.~\ref{fig:non-solvable} shows two networks (\cite{RaiD:09a,langberg3,RaiDS:itw10}) which are nonsolvable.
It can be verified that for the given labeling of sources, terminals and
edges, they satisfy Theorem~\ref{theorem:mainres}.

In~\cite{RaiDS:itw10} it was conjectured that the capacity of a $3s/3t$
sum-network is either $0,2/3$ or $\geq1$. Theorem~\ref{theorem:mainres} part D
states that the capacity of a nonsolvable connected $3s/3t$ sum-network
is $2/3$ and thus proves this conjecture.

A network that does not satisfy the
conditions in Theorem \ref{theorem:mainres} is solvable over all fields
except possibly $F_2$. So the conditions in Theorem \ref{theorem:mainres} are
necessary and sufficient for nonsolvability over any field other
than $F_2$. For $F_2$, the violation of these conditions does not imply
solvability. For example, Fig.~\ref{fig:x3} shows a network which does not satisfy
the hypothesis of Theorem~\ref{theorem:mainres}, but which is not solvable over
$F_2$ as was shown in~\cite{RaiD:09b}.
Theorem \ref{theorem:mainres2} below identifies the conditions under which a
$3s/3t$ network is solvable over any field except $F_2$.

\begin{theorem}\label{theorem:mainres2}[{\bf{Necessary and Sufficient condition for Non-solvability over $F_2$}}]
A connected $3s/3t$ sum-network $\mathscr{N}$ is not solvable
over $F_2$ but linearly solvable over any other field if and
only if  there exist two edges $e_1$ and $e_2$ and some labeling of
the sources and the terminals such that 
\begin{enumerate}
\item $e_1$ disconnects exactly $(s_1,t_3)$ and $(s_3,t_1)$
\item $e_2$ disconnects exactly $(s_2,t_3)$ and $(s_3,t_2)$
\item $mincut({\{s_3\},\{t_3\}}) = 0$ in $\{\set{N}-\{e_1,e_2\}\}$
\item ${e_1}\nrightarrow{e_2}$ and ${e_2}\nrightarrow{e_1}$.
\end{enumerate} 
It can be verified that for the given labeling of the sources, terminals and
edges, the networks in Fig.~\ref{fig:nonbinarysolvable} satisfy Theorem~\ref{theorem:mainres2}.

\end{theorem}

The following lemma is applicable to sum-networks with arbitrary number of
sources and terminals, and may be of independent interest for sum-networks
in general.

\begin{lemma}\label{lemma:reduction}
A connected $l$-source $j$-terminal sum network $\set{N}$ with
$\kappa(\set{N})=k$, $k>0$, and $\set{C}=\phi$ is linearly solvable
(respectively XOR solvable) over a field $F$
if all $l$-source $j$-terminal sum networks with $\kappa<k$ are
linearly solvable (respectively XOR solvable) over $F$.
\end{lemma}

In what follows, we present some lemmas which give simpler sufficient
conditions for a connected $3s/3t$ network to be solvable.
These lemmas will be used to prove the necessity parts of the main theorems.
\begin{lemma}
\label{lemma:con1}
A connected $3s/3t$ sum-network where there is no edge which is connected to
at least two sources and at least two terminals is linearly solvable
by XOR coding over any field.
\end{lemma}
\begin{lemma}\label{lemma:general}
Suppose a connected $3s/3t$ sum-network satisfies the following conditions.
For some labeling of the sources and the terminals,

(a) there is an edge $e$ such that $\{s_1,s_2\} \rightarrow e \rightarrow
\{t_1,t_2\}$.

(b) there is no edge which disconnects $(s_2,t_3)$ and $(s_3,t_1)$; or
$(s_1,t_3)$ and $(s_3,t_2)$.

Then the network is XOR solvable over any field.
\end{lemma}
\begin{lemma}\label{lemma:nopair}
Suppose a $3s/3t$ connected sum-network $\set{N}$ satisfies the following conditions.
(1) There does not exist an edge-pair which satisfies all the four conditions
of Theorem~\ref{theorem:mainres2}.

(2) For some labelling of the sources and the terminals,
there exist two edges $e_1,e_2$ such that

(a) $e_1$ disconnects $(s_1,t_3)$ and $(s_3,t_1)$

(b) $e_2$ disconnects $(s_2,t_3)$ and $(s_3,t_2)$

(c) $e_1\nrightarrow e_2$ and $e_2\nrightarrow e_1$

(d) Removing both $e_1$ and $e_2$ simultaneously does not disconnect
$(s_3,t_3)$.

Then the network is solvable over any field by XOR coding.
\end{lemma}

\begin{lemma}\label{lemma:suffcon}
If for a connected $3s/3t$ sum-network $\mathscr{N}$, $\kappa(\mathscr{N})\neq{3}$
then A. $\mathscr{N}$ is linearly solvable over all fields except possibly $F_2$,
B. whenever $\mathscr{N}$ is solvable over $F_2$, it is XOR solvable over all fields,
and C. if $\kappa(\set{N})\neq 2$; then $\set{N}$ is XOR solvable over all fields.
\end{lemma}
The network shown in Fig. \ref{fig:x3} (originally presented in \cite{RaiD:09b})
and the network shown in Fig. \ref{fig:x3a}
are examples of  networks with $\kappa = 2$ which are not solvable over $F_2$
but are linearly solvable over other fields.

\begin{figure}[htp]
\centering
\subfigure[]{
\includegraphics[height=2in]{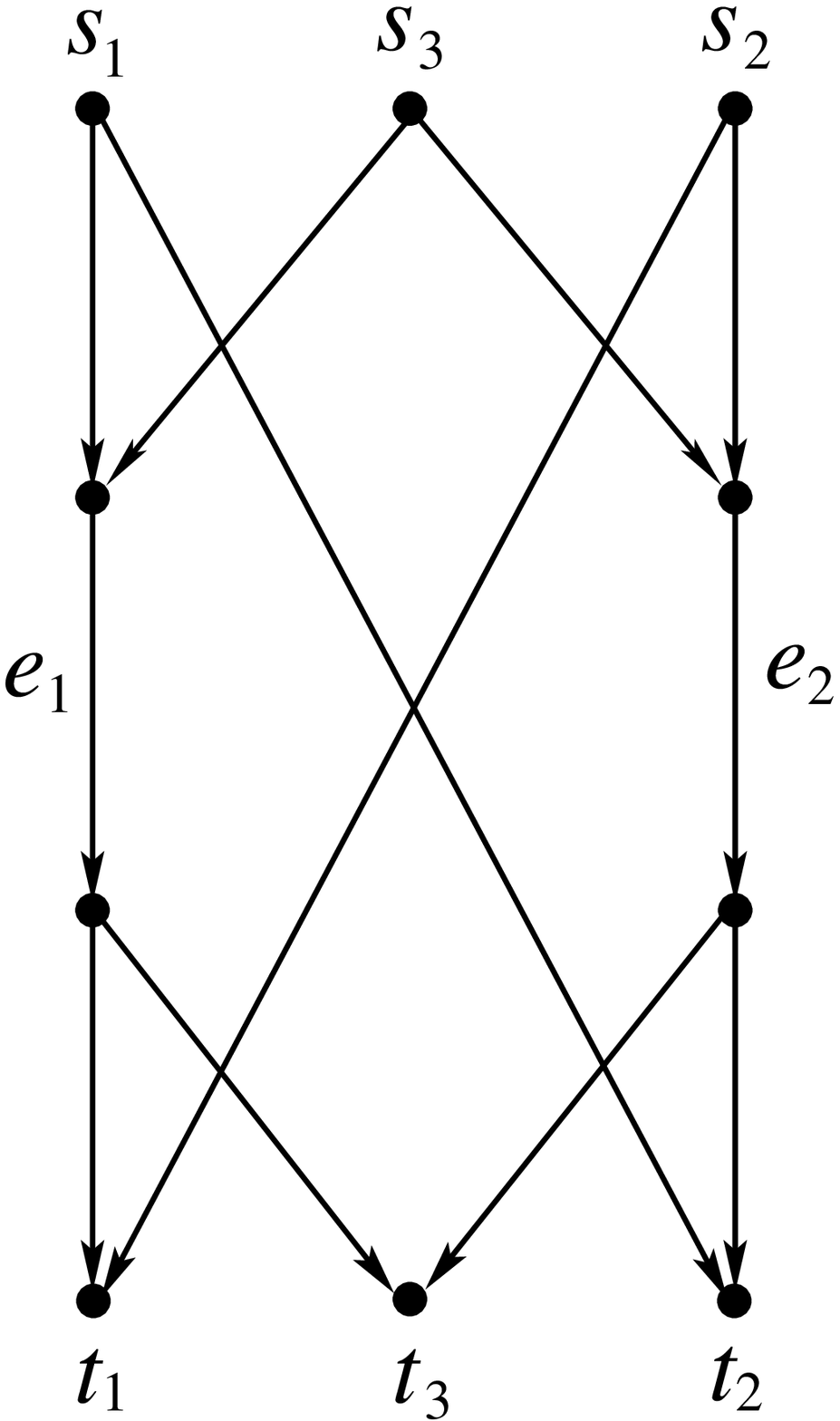}
\label{fig:s3}
}
\subfigure[]{
\includegraphics[height=2in]{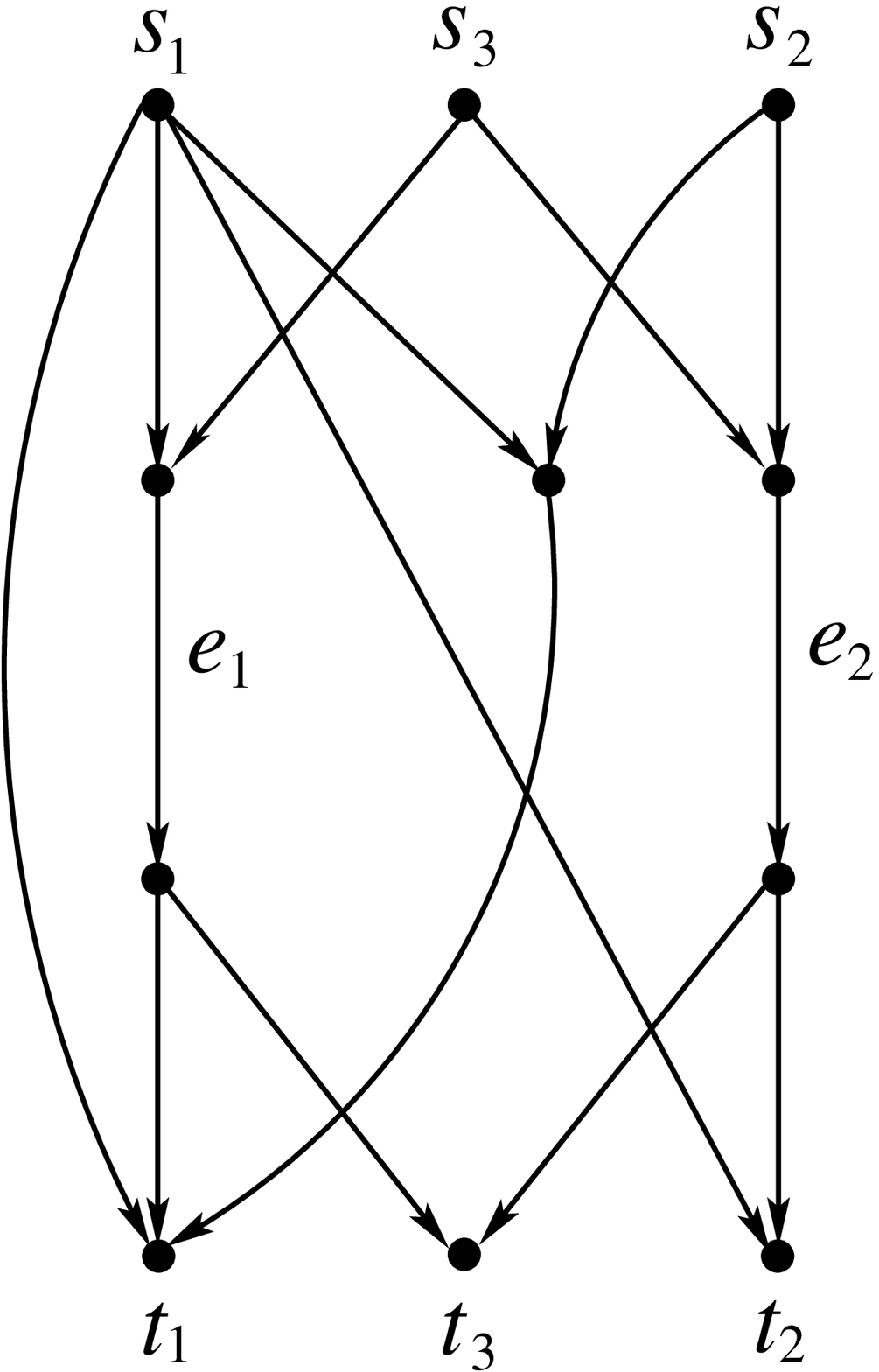}
\label{fig:s3a}
}
\caption[]{Some nonsolvable networks}
\label{fig:non-solvable}
\end{figure}

\begin{figure}[h]
\centering
\subfigure[]{
\includegraphics[height=2in]{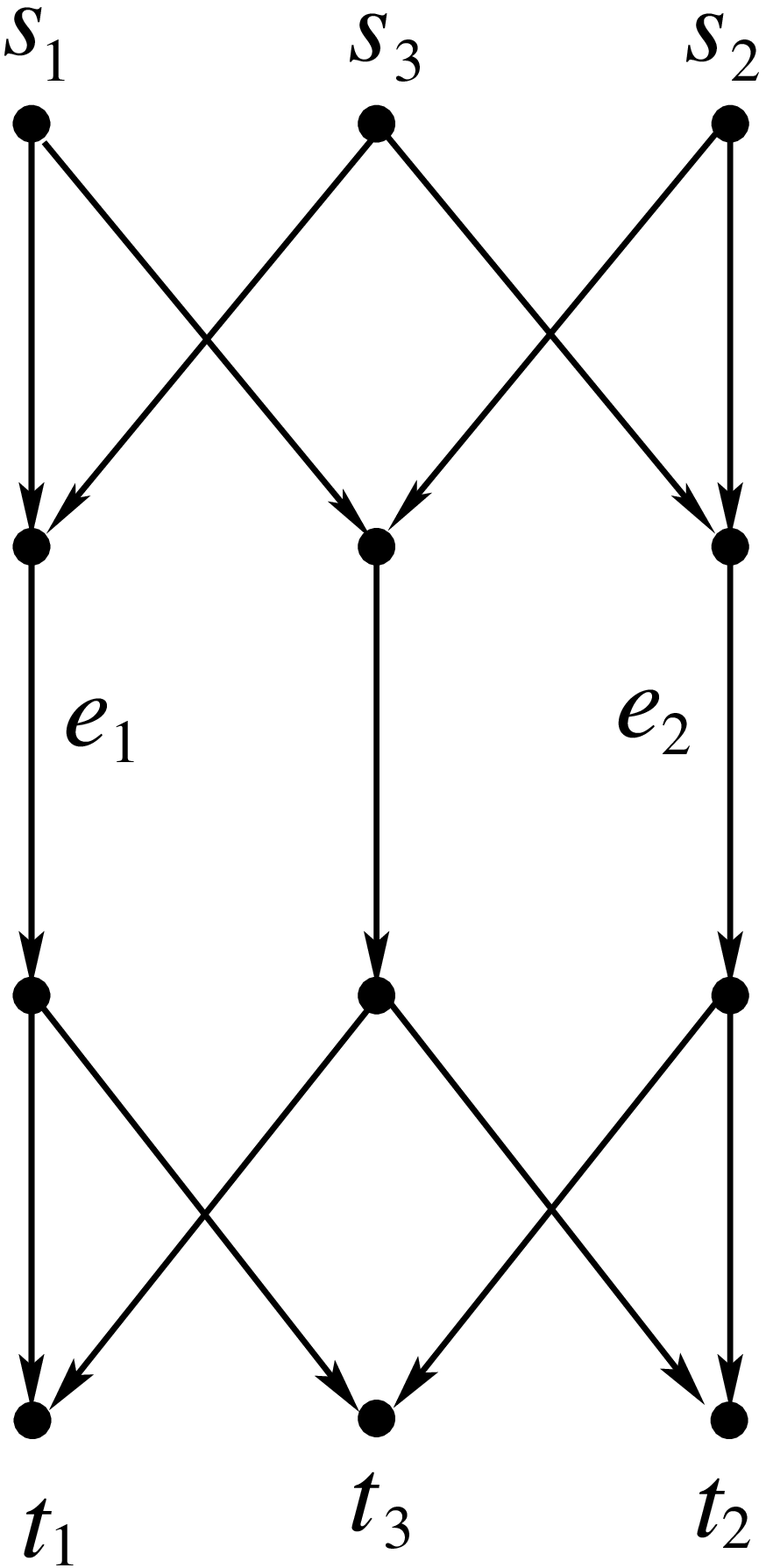}
\label{fig:x3}
}
\subfigure[]{
\includegraphics[height=2in]{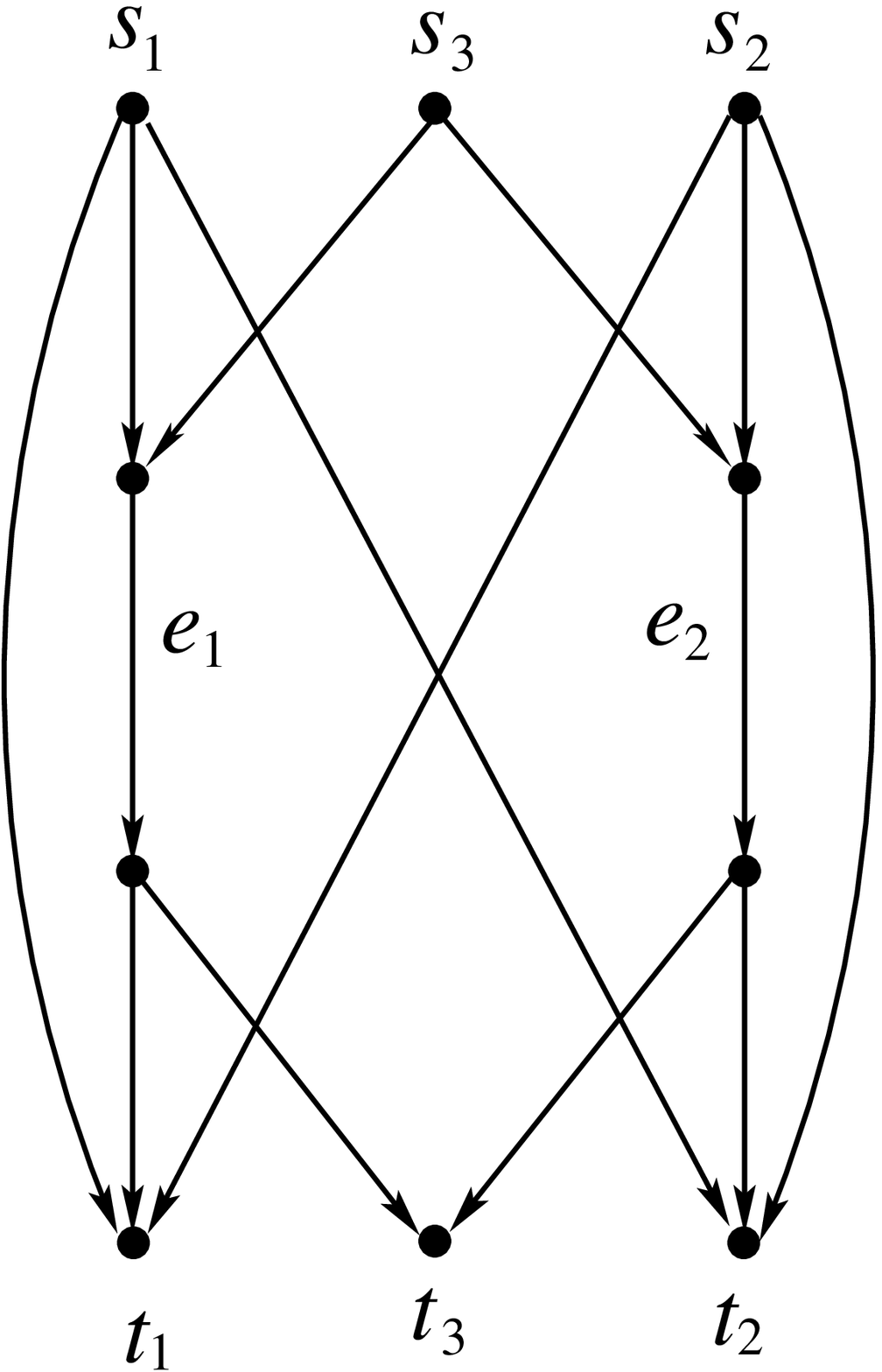}
\label{fig:x3a}
}
\caption[]{Some networks which are not solvable over $F_2$ but are linearly solvable over any other field}
\label{fig:nonbinarysolvable}
\end{figure}

\begin{lemma}\label{lemma:converse1}
Let $\net$ be a connected $3s/3t$ sum-network with $\kappa (\net ) = 3$ and $\set{C}=\phi$.
A. If $\net$ has an edge pair satisfying
the conditions 1-4 of Theorem~\ref{theorem:mainres2} then it is not solvable over $F_2$
but linearly solvable over other fields. B. If $\net$
does not have an edge pair satisfying conditions 1-4 of Theorem~\ref{theorem:mainres2}
then it is XOR solvable over all fields.
\end{lemma}

\begin{lemma}\label{lemma:converse2}
Given a connected $3s/3t$ sum-network $\set{N}$ with $\kappa(\net)=3$, if for some
labeling of its sources and terminals, there exists an edge
$e_2$ satisfying conditions 3 and 4 of Theorem~\ref{theorem:mainres}, then $\set{N}$
is nonsolvable only if another edge $e_1$ exists such
that $e_1$ and $e_2$ satisfy all the six conditions of Theorem~\ref{theorem:mainres},
else $\net$ is XOR solvable over all fields.
\end{lemma}

\section{Proofs}\label{section:proofs}
We start by presenting some known results which will be used in
the proofs of our results. Considering the complexity of the proof of
the main results and their dependence on so many lemmas, a dependency
graph of the results is shown in Fig.~\ref{fig:dependence} for clarity.
\begin{figure*}[!t]
\centering
\includegraphics[width=4in]{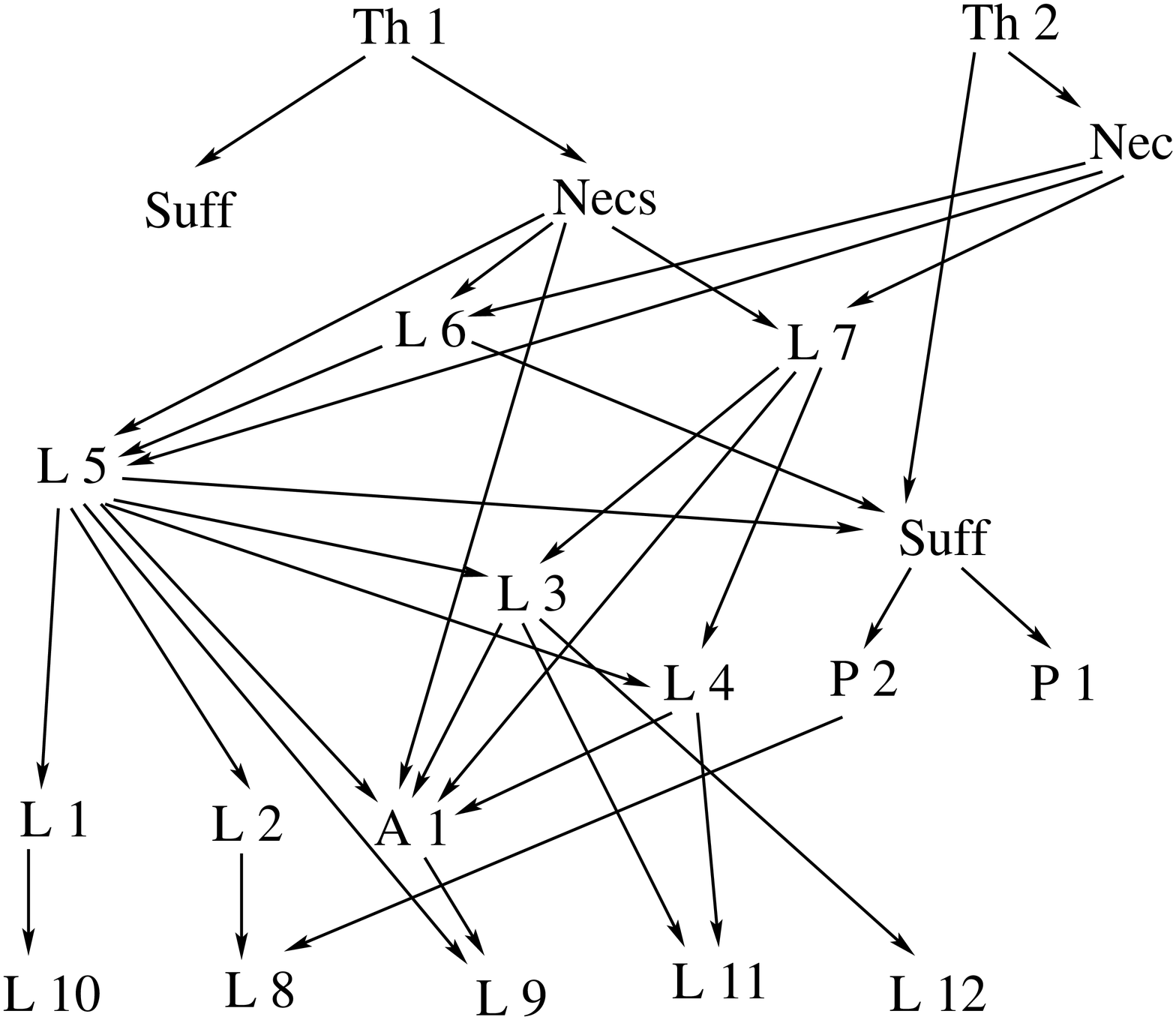}
\caption{Dependency graph of the results. Here L, P, Th, Suff, Necs, and
A stand for respectively Lemma, Part, Theorem, Sufficiency, Necessity, and
Assumption.}
\label{fig:dependence}
\end{figure*}

\begin{lemma}\cite{ramamoorthy}\label{lemma:R}
A sum-network for which either the number of sources or
the number of terminals is at most two is solvable if and only if the network is connected.
Moreover, such a connected network is XOR solvable over any field.
\end{lemma}
In~\cite{langberg3}, the authors proved the following as a side-result:
\begin{lemma}\label{lemma:5con}
\cite{langberg3}
If in a connected $3s/3t$ sum-network there exists a node
$v$ such that there is a path from all the sources (resp. at least two sources)
to $v$ and there is a path from $v$ to at least two terminals
(resp. all the terminals) then the network is XOR solvable over any field.
\end{lemma}
\begin{corollary}\label{corollary:srctrm}
In a $3s/3t$ connected sum-network, if there is a path from
one source (or terminal) to another,
then the network is XOR solvable over any field.
\end{corollary}
\begin{proof}
If $s_i \rightarrow s_j$, then $s_j$ satisfies the hypothesis of
Lemma~\ref{lemma:5con} and thus the corollary follows.
\end{proof}
So w.l.o.g., {\it we assume that the sources have no incoming edges and the
terminals have no outgoing edges.}

\begin{lemma}\cite[Theorem 5]{RaiD:09b}\label{lemma:reverse}
If a sum-network $\mathscr{N}$ is linearly solvable over a field $F$,
then so is its reverse sum-network $\set{\tilde{N}}$. Further, if
$\mathscr{N}$ has a XOR solution over $F$, then so does the reverse network.
\end{lemma}
The second part of the above lemma was not explicitly mentioned
in \cite{RaiD:09b}, but can be easily seen to follow from the reverse code
construction proposed therein.

The next two lemmas are in relation to the double-unicast
problem~\cite{shenvi2}, where there are two source-terminal pairs $(s_1,t_1)$ and
$(s_2,t_2)$, and each terminal wants to recover the symbol generated at the corresponding
source over a directed acyclic network with unit capacity edges.
In~\cite{shenvi2}, a simple necessary and sufficient condition was
given for such a ``double-unicast " network to support two such simultaneous unicasts. The following
lemma is a sufficient condition for supporting two
simultaneous unicasts and was proved in Case IIB of
\cite[Proof of Theorem 1]{shenvi2}.
\begin{figure*}[!t]
\centering
\subfigure[]{
\includegraphics[width=1.2in]{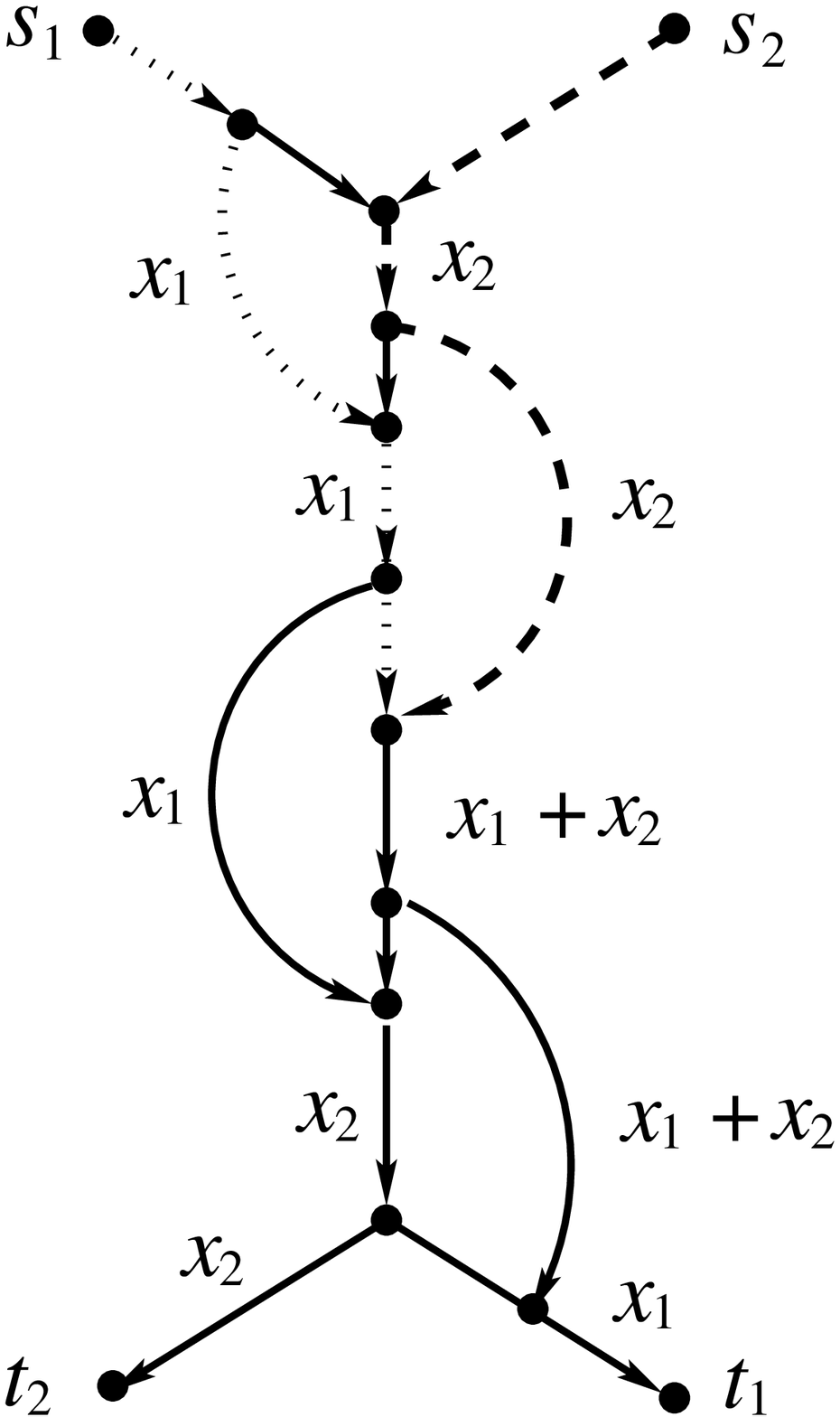}
\label{fig:multigrail1-a}
}
\subfigure[]{
\includegraphics[width=1.2in]{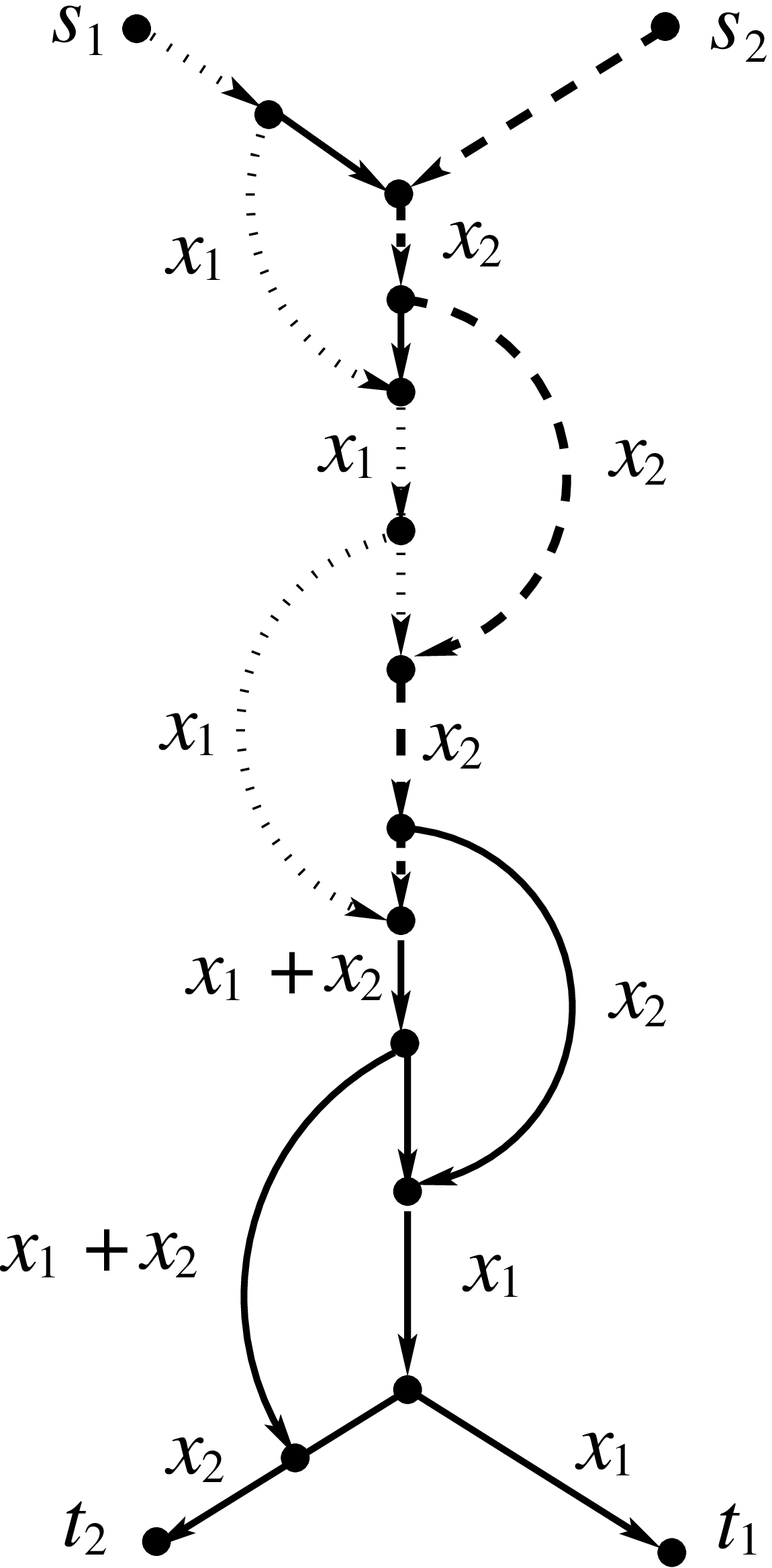}
\label{fig:multigrail1-b}
}
\subfigure[]{
\includegraphics[width=1.2in]{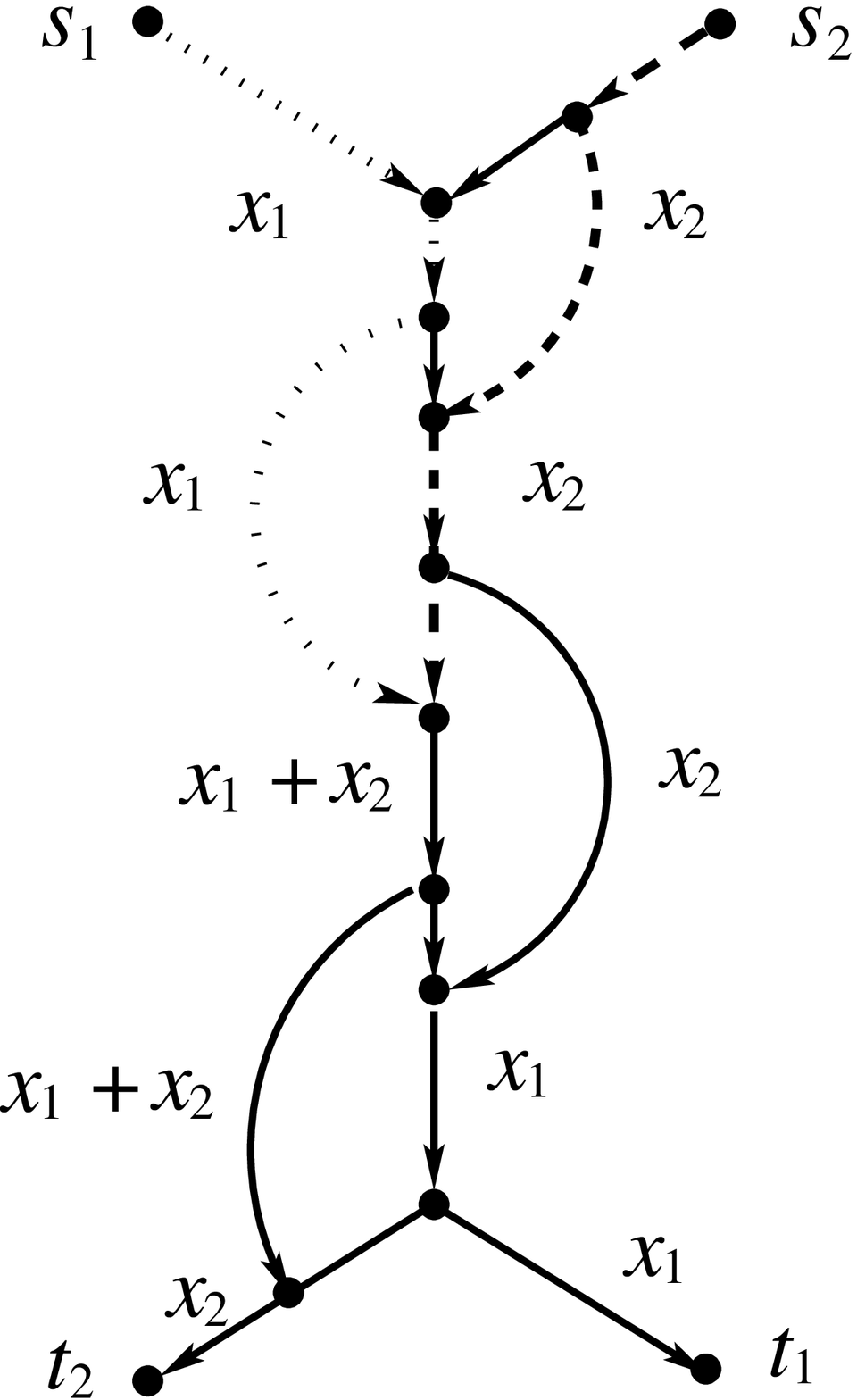}
\label{fig:multigrail1o-a}
}
\subfigure[]{
\includegraphics[width=1.2in]{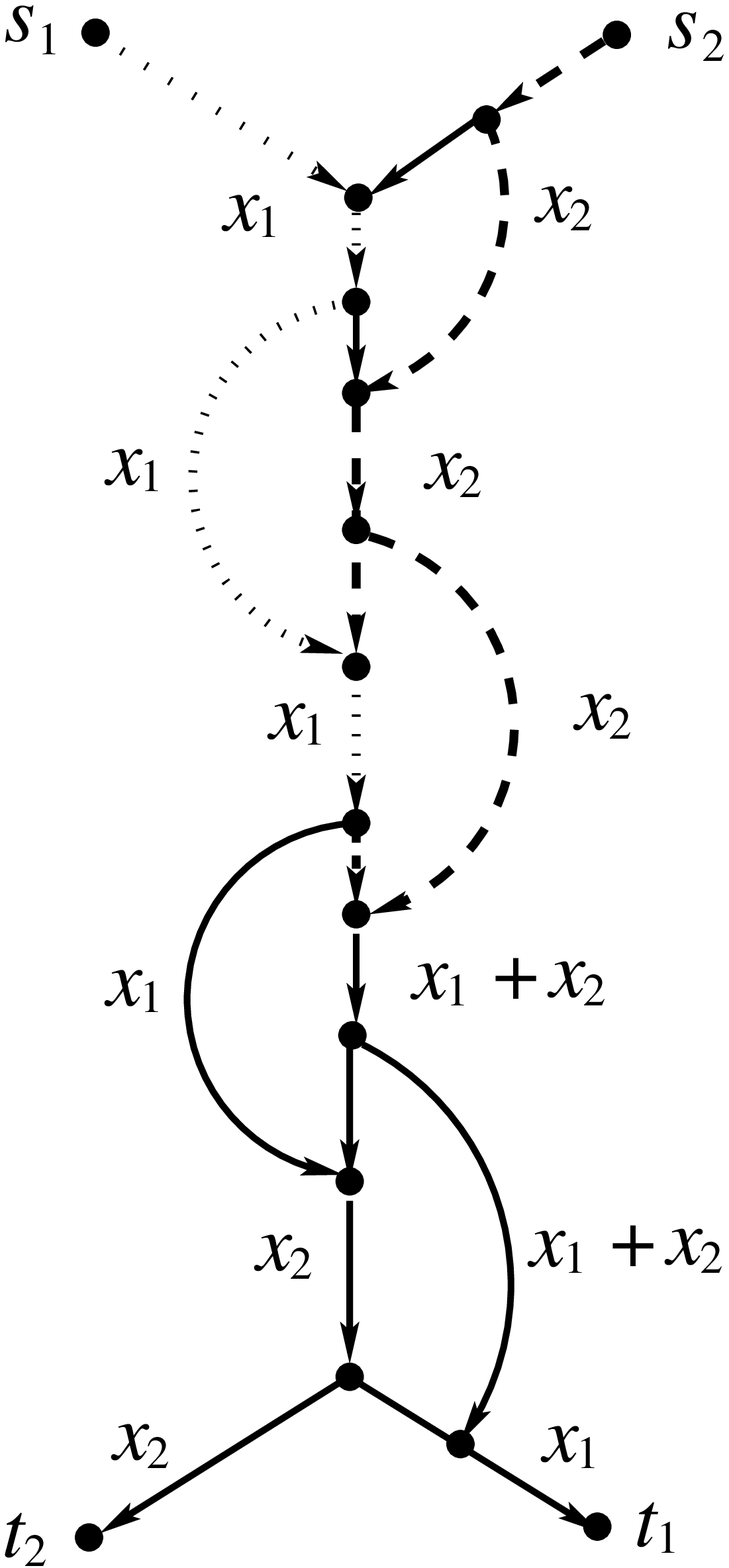}
\label{fig:multigrail1o-b}
}
\caption[]{The coding on grails for Lemma~\ref{lemma:algo0} \cite{{shenvi2}}}
\label{fig:multigrail1}
\end{figure*}
\begin{lemma}\label{lemma:algo0}
\cite{shenvi2}
Suppose in a double-unicast network with connected source terminal pairs
$(s_1,t_1)$ and $(s_2, t_2)$, removing all the edges of any $(s_1,t_1)$
path disconnects $(s_2,t_2)$ and there is no single edge in the network
whose removal disconnects both $(s_1,t_1)$ and $(s_2,t_2)$.
Then there exists a XOR code which allows the communication of 
$x_1$ to $t_1$ and $x_2$ to $t_2$.
\end{lemma}
The proof in \cite{shenvi2} argued that the network is essentially a ``grail''
with possibly multiple (even or odd number of) ``handles'' as shown
in Fig. \ref{fig:multigrail1}. Explicit coding schemes, as shown in the figure,
were given to achieve the double-unicast.

The next lemma follows by simple modifications in the coding
schemes under case IIB of
\cite[Proof of Theorem 1]{shenvi2}.
\begin{figure*}[!t]
\centering
\subfigure[]{
\includegraphics[width=1.2in]{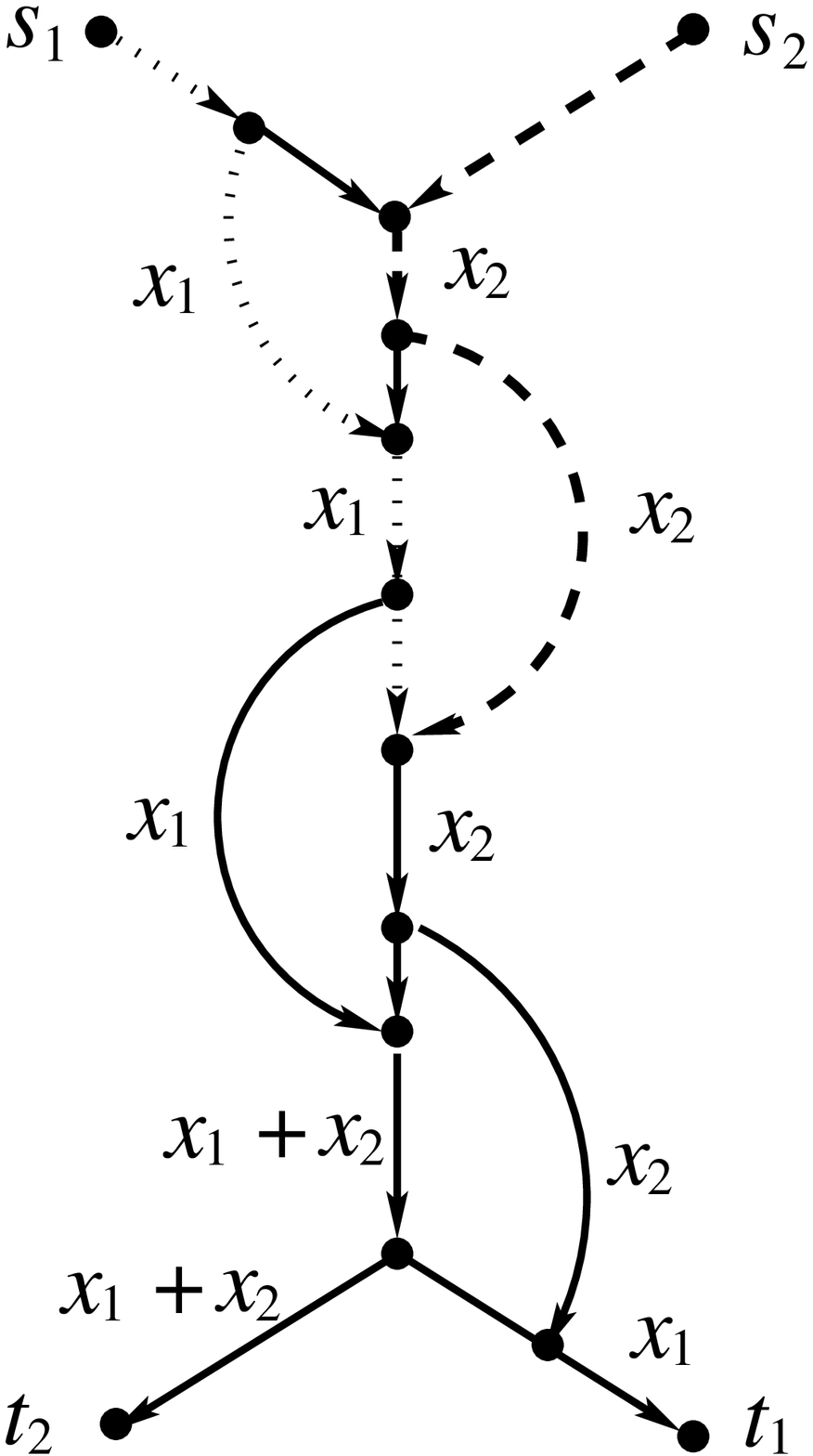}
\label{fig:multigrail2-a}
}
\subfigure[]{
\includegraphics[width=1.2in]{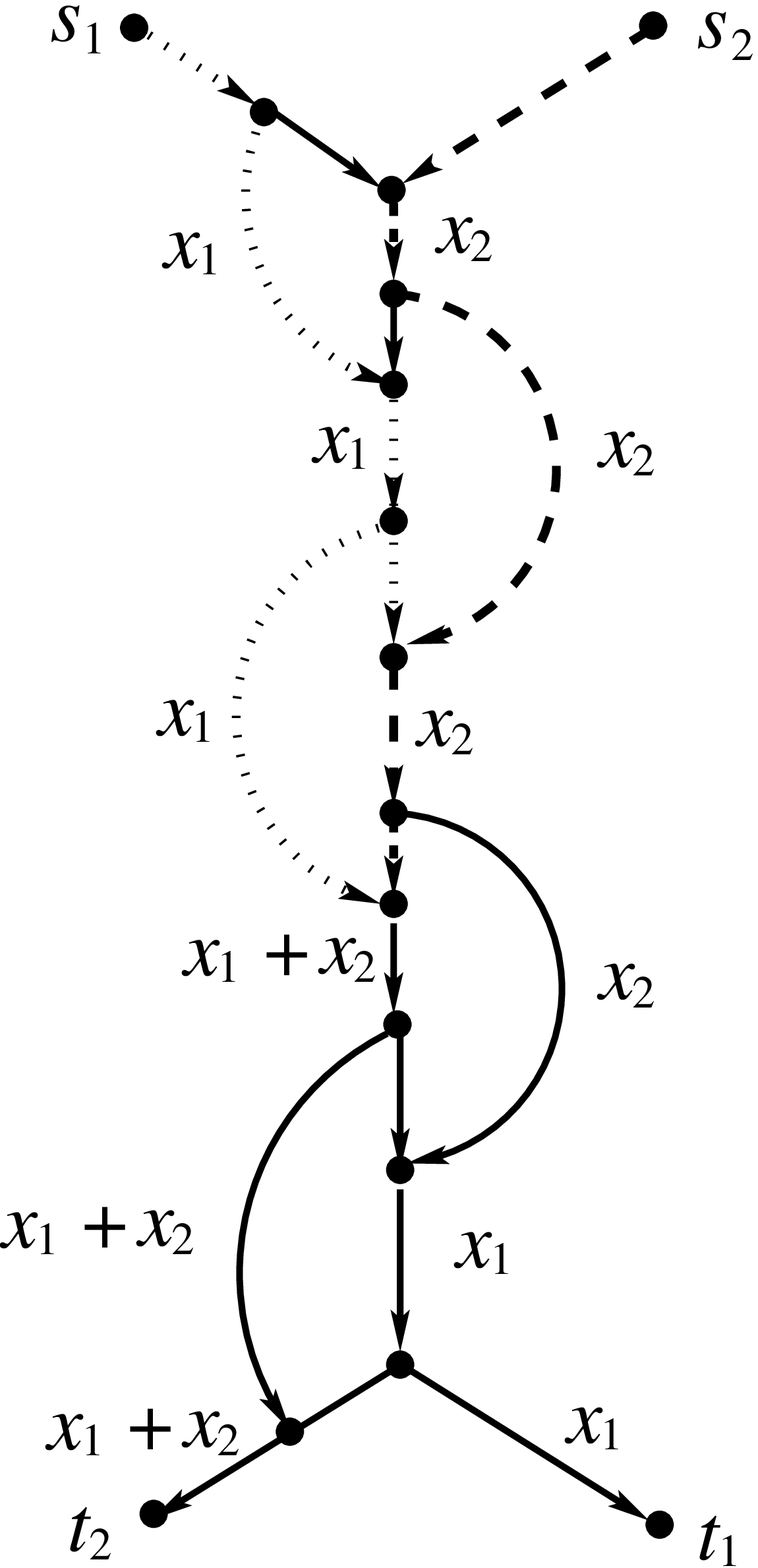}
\label{fig:multigrail2-b}
}
\subfigure[]{
\includegraphics[width=1.2in]{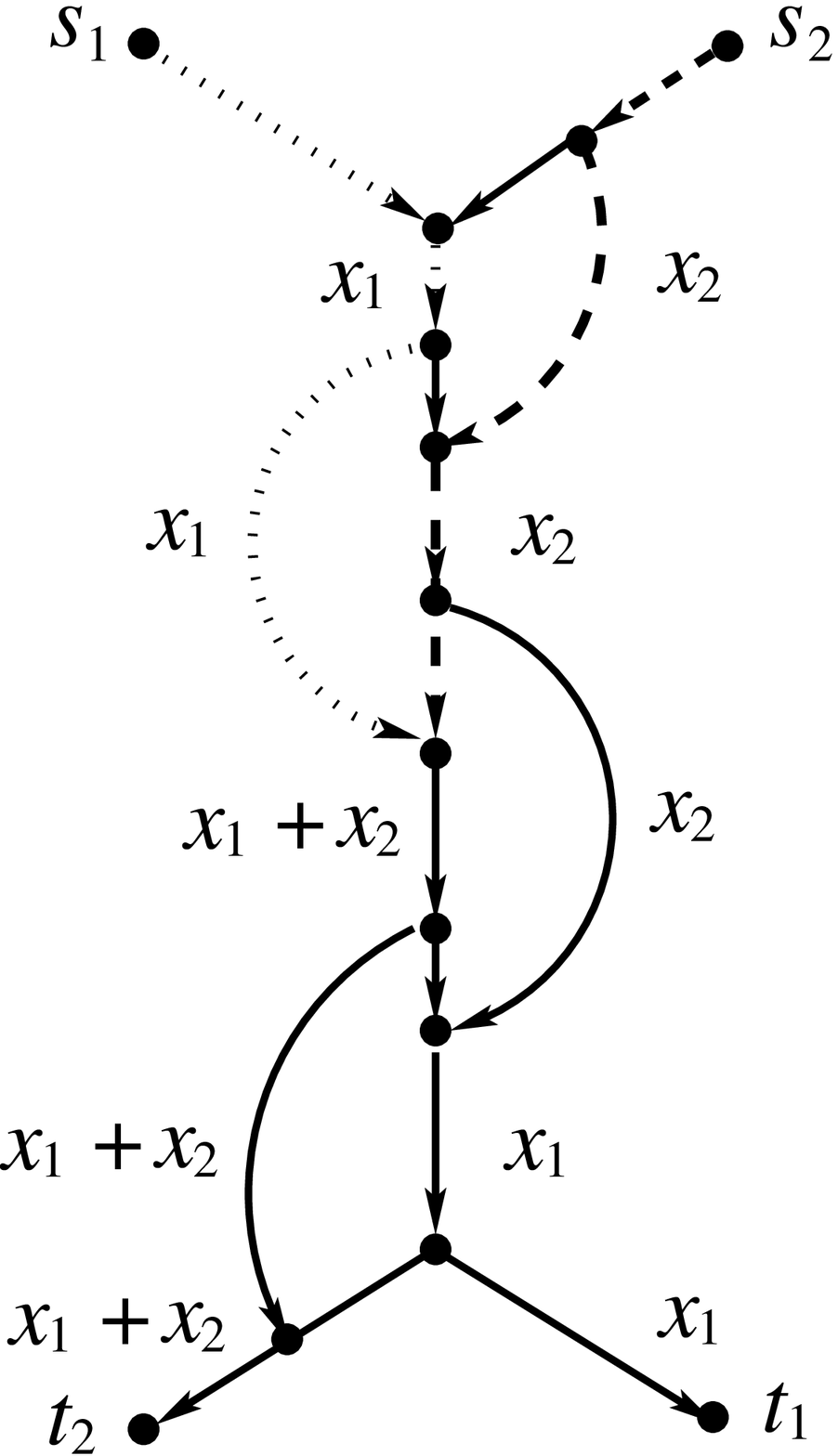}
\label{fig:multigrailo2-a}
}
\subfigure[]{
\includegraphics[width=1.2in]{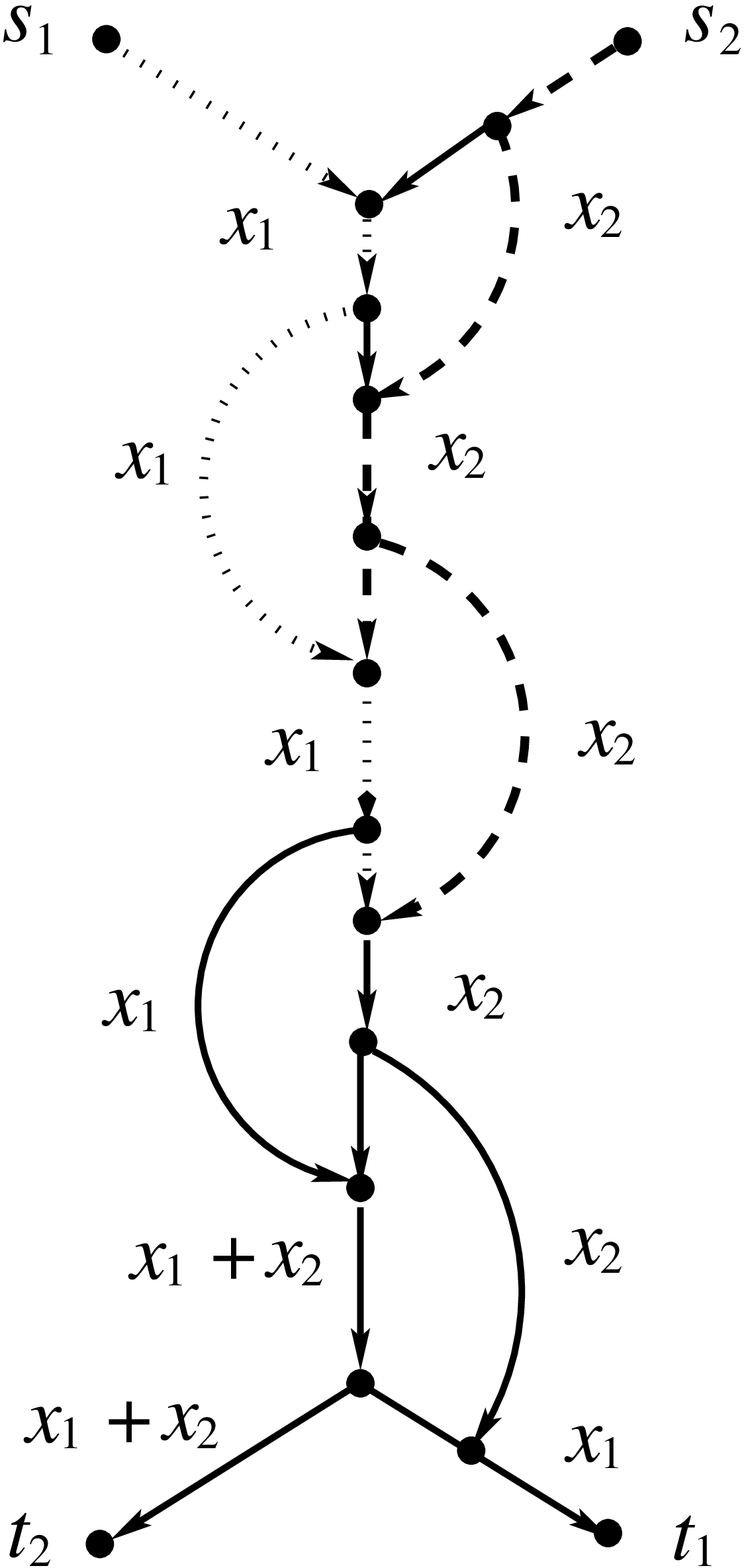}
\label{fig:multigrailo2-b}
}
\caption[]{The coding on grails for Lemma~\ref{lemma:algo}}
\label{fig:multigrail2}
\end{figure*}

\begin{lemma}\label{lemma:algo}
Suppose in a double-unicast network with connected source terminal pairs
$(s_1,t_1)$ and $(s_2, t_2)$, removing all the edges of any $(s_1,t_1)$
path disconnects $(s_2,t_2)$ and there is no single edge in the network
whose removal disconnects both $(s_1,t_1)$ and $(s_2,t_2)$.
Then there exists a XOR code which allows the communication of
$x_1$ to $t_1$ and $x_1+x_2$ to $t_2$.
\end{lemma}

{\it Proof:} The proof is achieved by changing the coding
on the grail networks as shown in Fig. \ref{fig:multigrail2}.

Now we start proving our results. Because of Lemma~\ref{lemma:5con}, 
whenever we need to prove solvability under some conditions,
we make the following assumption without loss of generality.

\newtheorem{assumption}{Assumption}
\begin{assumption}\label{assumption:simplify}  $\set{N}$ does not contain a node that
satisfies the hypothesis of Lemma~\ref{lemma:5con}.
\end{assumption}

{\bf Proof of Lemma~\ref{lemma:reduction}:}

Consider the new network $\set{N}^*$ formed by adding
an edge $e_i^*$ in parallel with the edge $e_i$ for each edge $e_i \in
\mathscr{A}\cup \mathscr{B}$ (Adding $e_i^*$ in parallel with $e_i$ means
that the head and the tail of $e_i^*$ are the same as those of $e_i$).
Clearly $\kappa(\set{N}^*)<k$, and so by the hypothesis of the lemma,
$\set{N}^*$ is linearly solvable over F.
But in any linear code for $\set{N}^*$, for every edge in $\mathscr{B}$, the
edge and its
added parallel edge carry essentially the same data since there is
a path from only one source to the tail of these edges. So we can
remove the edges we added in parallel to the edges of $\mathscr{B}$ and
the new resulting network $\set{N}^{**}$ will still be linearly solvable
over $F$.
Then by Lemma~\ref{lemma:reverse} its reverse network is also linearly
solvable over $F$. But by the same argument, this reverse network remains
linearly solvable over $F$ even after removing the
remaining extra edges in parallel
to the edges in $\mathscr{A}$. So, again by Lemma~\ref{lemma:reverse},
the original network $\mN$ itself is linearly solvable over $F$.
The above arguments also hold word by word if ``linearly solvable'' is
replaced by ``XOR solvable''. This completes the proof.

{\bf Proof of Lemma~\ref{lemma:con1}:}

First communicate $x_1+x_2+x_3$ to $t_1,t_2$ by XOR coding, which is possible by Lemma~\ref{lemma:R}.
Let $\set{N}_1$ be the sub-network used for this code.
Now let $P_1,P_2,P_3$ be some $(s_1,t_3),(s_2,t_3),(s_3,t_3)$ paths respectively and let
$\set{N}_2$ be the sub-network consisting of them. We can simultaneously communicate $x_1+x_2+x_3$ to $t_3$
by XOR coding over $\set{N}_2$ for the following reason. Any edge $e \in \set{N}_1 \cap \set{N}_2$ has paths to at least
two terminals: $t_3$ and at least one of $t_1$, $t_2$. By the hypothesis
of the lemma, there is a path from exactly one source, say $s_1$ (w.l.o.g.),
to $\tail{e}$. Thus $e$ essentially carries only $x_1$ in the coding scheme over $\set{N}_1$, as well as
in the coding scheme over $\set{N}_2$.
Hence there is no conflict between the coding schemes over $\set{N}_1$ and $\set{N}_2$ and both the codes can be simultaneously implemented.
Thus the network is linearly solvable over any field in this case.

{\bf Proof of Lemma~\ref{lemma:general}:}

W.l.o.g. we assume that the network satisfies Assumption~\ref{assumption:simplify}.
\begin{observation}\label{observation:1}
(i) By Assumption~\ref{assumption:simplify},
$s_3 \nrightarrow \Gamma_{\head{e}}\cup \Gamma^{\tail{e}}_{t_1,t_2} $.

(ii) Similarly, $\Gamma^{\tail{e}}\cup \Gamma^{s_1,s_2}_{\head{e}} \nrightarrow t_3$.
\end{observation}
Observation~\ref{observation:1} implies the following.
\begin{observation}\label{observation:1a}
(i) Observation~\ref{observation:1}(ii) implies that no $(s_1,t_3)$ path contains any node from
$ \Gamma^{s_2}_{\tail{e}}\cup \tail{e}\cup \Gamma^{\tail{e}}$, 

(ii) Observation~\ref{observation:1}(ii) implies that no $(s_2,t_3)$ path contains any node from
$ \Gamma^{s_1}_{\tail{e}}\cup \tail{e}\cup \Gamma^{\tail{e}}$, 

(iii) From Observation~\ref{observation:1}(i) and (ii), we have
$s_3 \nrightarrow \Gamma^{s_1}_{\tail{e}}\cup \Gamma^{s_2}_{\tail{e}}$
and $\{\tail{e}\}\cup \Gamma^{\tail{e}} \nrightarrow t_3$. Together,
they imply that no
$(s_3,t_3)$ path contains any node from
$\Gamma^{s_1}_{\tail{e}}\cup \Gamma^{s_2}_{\tail{e}}\cup\{\tail{e}\}\cup \Gamma^{\tail{e}}$.

(iv) Observation~\ref{observation:1}(i) implies that no $(s_3,t_1)$ or
$(s_3,t_2)$ path contains any node from
$\Gamma^{s_1}_{\tail{e}}\cup \Gamma^{s_2}_{\tail{e}}\cup\{\tail{e} 
\}\cup\Gamma^{\tail{e}}_{t_1,t_2}$.
\end{observation}

Let us denote the subnetwork obtained by taking all $(s_1,t_3)$,
$(s_2,t_3)$, $(s_3,t_3)$, $(s_3,t_1)$ and $(s_3,t_2)$ paths by $\onet$.
This subnetwork contains all edges $e'$ such that either
$e' \rightarrow t_3$ or $s_3 \rightarrow e'$.

\begin{observation}\label{observation:1b}
Observation~\ref{observation:1a} above implies that
irrespective of the coding used on $\onet$, we can still
communicate $x_1+x_2$ over $e$ by passing $x_1$ and $x_2$ through
any chosen $(s_1,\tail{e})$ and $(s_2, \tail{e})$ paths respectively.
That is, $x_1+x_2$ can be passed on edge $e$ without putting
any constraint on the coding on the subnetwork $\onet$.
\end{observation}

Let $\mathscr{P}(s_3,t_1) = \{P_1,P_2,\ldots\}$ be the set of
all $(s_3,t_1)$ paths.  For any $P_i \in \mathscr{P}(s_3,t_1)$,
let $z_i$ denote the first descendant of $\head{e}$ on this path.
The existence of $z_i$ is ensured by the fact that $t_1$ is a descendant
of $\head{e}$, and is on $P_i$.
Similarly, let $\mathscr{Q}(s_3,t_2) = \{Q_1,Q_2,\ldots\}$ be the set of
all $(s_3,t_2)$ paths, and for any $Q_j \in \mathscr{Q}(s_3,t_2)$,
let $y_j$ denote the first descendant of $\head{e}$ on this path.

\begin{observation}\label{observation:2}
By Assumption~\ref{assumption:simplify}, $\forall i,j$,
$z_i\nrightarrow \{t_2,t_3\}$ and
$y_j\nrightarrow \{t_1,t_3\}$.
\end{observation}
We consider the following two cases:

{\it Case 1: For any $P_i \in \mathscr{P}(s_3,t_1)$, removing all the
edges on $P_i(s_3:z_i)$ disconnects $(s_1,t_3)$ and/or $(s_2,t_3)$.}

Removing all the edges
of any single $(s_3,t_1)$ path can not disconnect both $(s_1,t_3)$ and $(s_2,t_3)$,
since otherwise this $(s_3,t_1)$ path will contain a node
$v$ such that $\{s_1,s_2,s_3\} \rightarrow v \rightarrow \{t_1,t_3\}$,
thus contradicting
Assumption~\ref{assumption:simplify}. So we consider the following
three cases under Case 1: 
Case 1.1: the removal of any path $P_i$ disconnects
only $(s_2,t_3)$, Case 1.2: the removal of any path $P_i$ disconnects
only $(s_1,t_3)$, and Case 1.3: the removal of some of the paths disconnects
only $(s_1,t_3)$ and the removal of any of the others disconnects only
$(s_2,t_3)$.

{\it Case 1.1: For any $P_i \in \mathscr{P}(s_3,t_1)$, removing all the
edges on $P_i(s_3:z_i)$ disconnects $(s_2,t_3)$ but not $(s_1,t_3)$.}

\begin{observation}\label{observation:3} By Assumption~\ref{assumption:simplify}, we have for this case,

(i) Any $(s_1,t_3)$ path
is node-disjoint from any $P_i\in \mathscr{P}(s_3,t_1)$, since 
otherwise $\Gamma^{s_1,s_2,s_3}_{t_1,t_3} \neq \emptyset$.

(ii) Any $Q_j \in \mathscr{Q}(s_3,t_2)$ shares only those nodes with
$P_i\in \mathscr{P}(s_3,t_1)$ which are not descendants of $s_2$,
since otherwise $\Gamma^{s_2,s_3}_{t_1,t_2,t_3} \neq \emptyset$.

(iii) Any $Q_j \in \mathscr{Q}(s_3,t_2)$ is node-disjoint from any
$(s_2,t_3)$ path, since
otherwise $\Gamma^{s_2,s_3}_{t_1,t_2,t_3} \neq \emptyset$.
\end{observation}

Since for any $i$, removing all the edges on $P_i$ disconnects
$(s_2,t_3)$, and no single
edge in the network disconnects both $(s_2,t_3)$ and
$(s_3,t_1)$ (by hypothesis (b) of the lemma), 
by Lemma~\ref{lemma:algo} we can transmit $x_2+x_3$
to $t_3$ and $x_3$ to $t_1$. The sub-network (say $\set{N}_1$) used
for this purpose is a
grail with either even or odd number of handles like those in
Fig.~\ref{fig:multigrail3}, and w.l.o.g., let us assume that the $(s_3,t_1)$
path taking part in the grail is $P_1$. By this coding (on the grail),
$z_1$ receives $x_3$ and $t_3$ receives $x_2+x_3$.
We now have the following two sub-cases under Case 1.1:

{\it Case1.1.1: For some path in $\mathscr{Q}(s_3,t_2)$, say $Q_1$,
removing all the edges on $Q_1(s_3:y_1)$ does not disconnect $(s_1,t_3)$.}

By Observations~\ref{observation:3}(ii) and \ref{observation:3}(iii), $Q_1$ is node-disjoint from the grail $\set{N}_1$
except at the dark-shaded part shown in Fig.~\ref{fig:multigrail3}.
Since this dark-shaded part also carries $x_3$, we can transmit
$x_3$ on $Q_1(s_3:y_1)$ without affecting the coding on grail $\set{N}_1$.
Further, by Observations~\ref{observation:1a}(ii), (iii), (iv) and Observation~\ref{observation:1b}
we can simultaneously communicate $x_1+x_2$ to $z_1$ and $y_1$ via $e$.
Then $t_1$ and $t_2$ get $x_1+x_2+x_3$ from $z_1$ and $y_1$ respectively.
\begin{figure*}[!t]
\centering
\subfigure[]{
\includegraphics[width=1.2in]{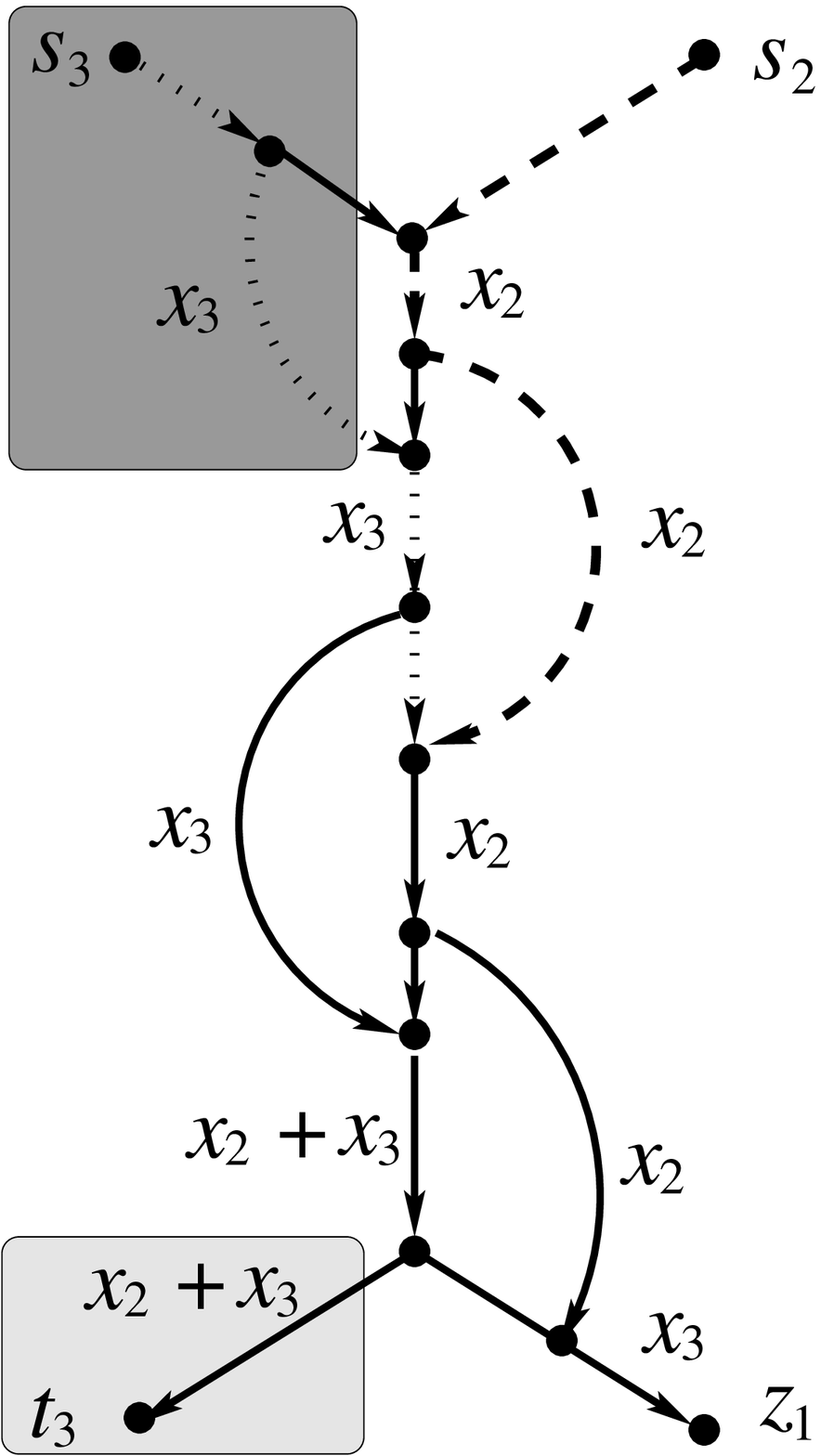}
\label{fig:multigrail3-a}
}
\subfigure[]{
\includegraphics[width=1.2in]{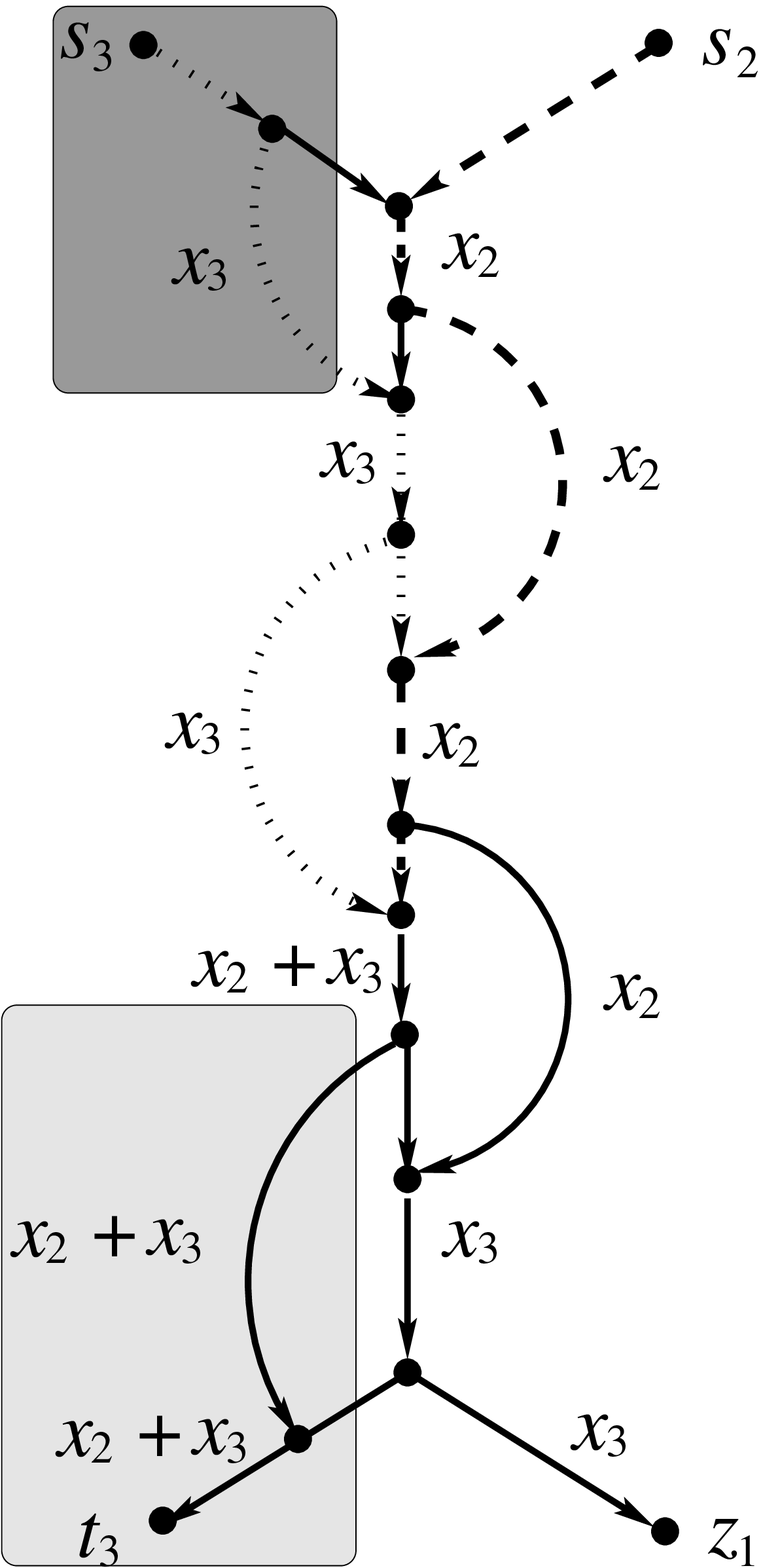}
\label{fig:multigrail3-b}
}
\subfigure[]{
\includegraphics[width=1.2in]{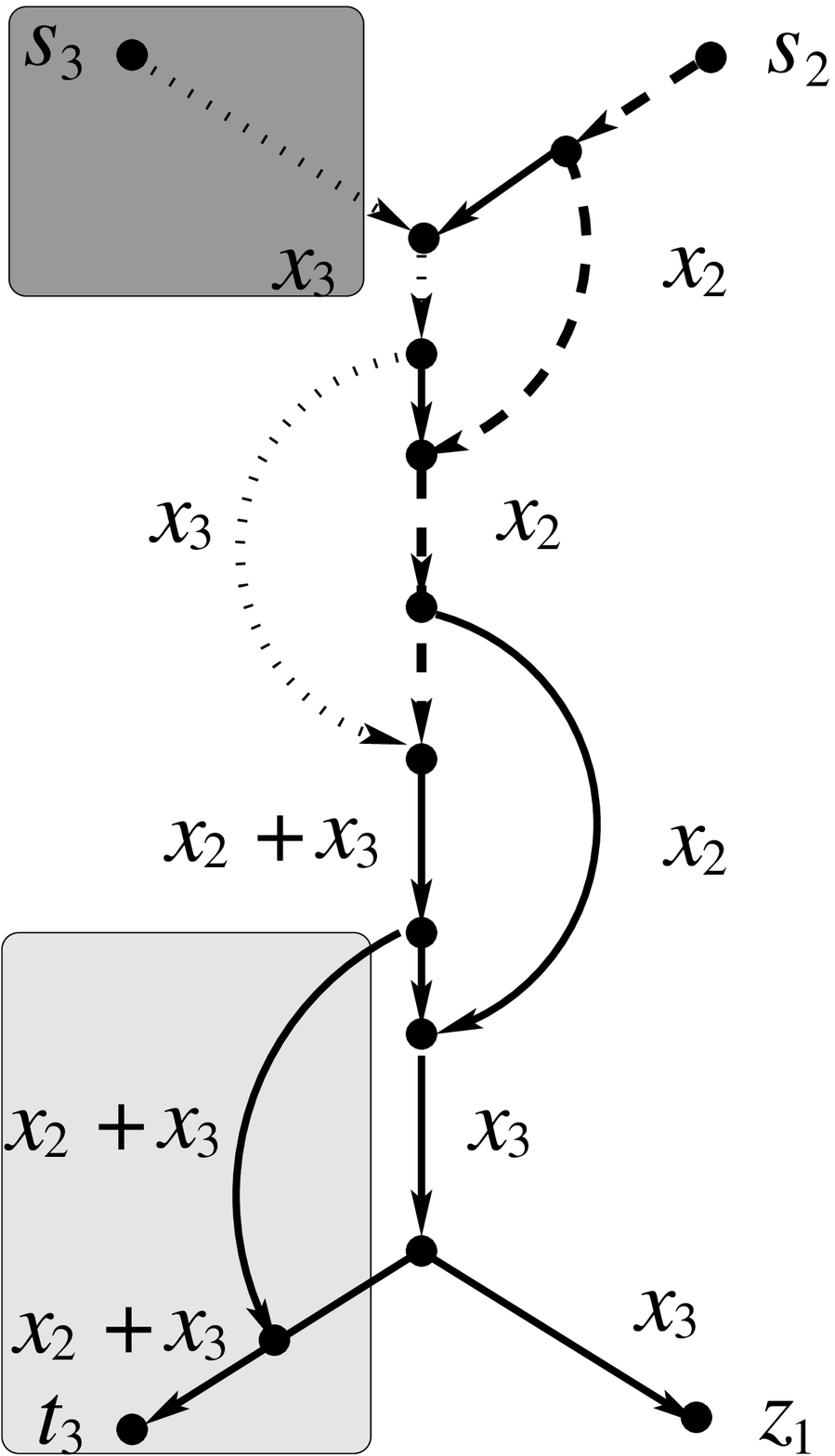}
\label{fig:multigrailo3-a}
}
\subfigure[]{
\includegraphics[width=1.2in]{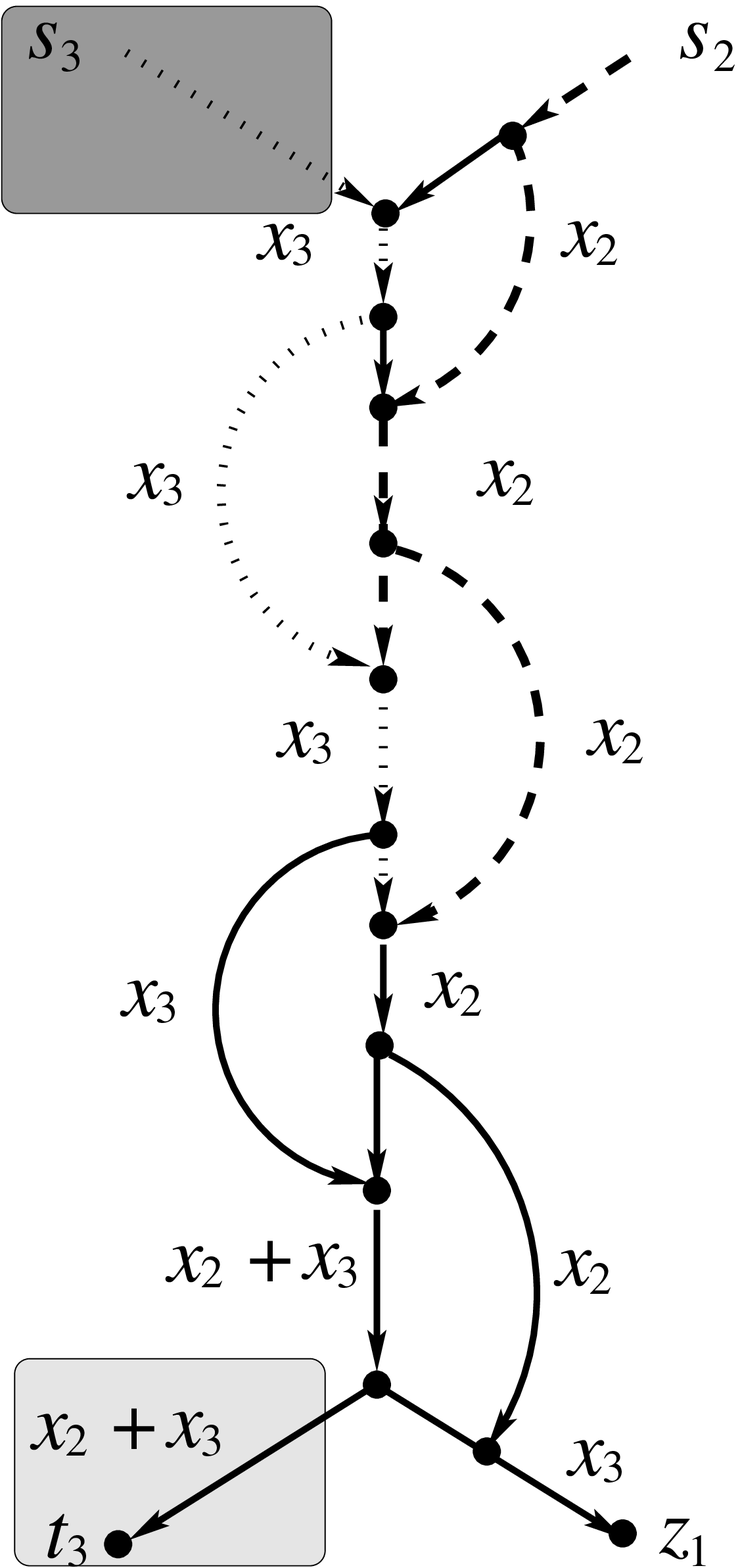}
\label{fig:multigrailo3-b}
}
\caption[]{The coding on grail $\set{N}_1$ for Case 1.1 of Lemma~\ref{lemma:general}}
\label{fig:multigrail3}
\end{figure*}
By the hypothesis of Case 1.1.1, there exists a $(s_1,t_3)$ path, say $P(s_1,t_3)$,
which is edge-disjoint from $Q_1$. By Observation~\ref{observation:3}(i), 
$P(s_1,t_3)$ is also node-disjoint from the grail $\set{N}_1$ except at the light-shaded
part shown in Fig.~\ref{fig:multigrail3} which
carries $x_2+x_3$. This and Observations~\ref{observation:1a}(i), and \ref{observation:1b}
imply that we can now simultaneously transmit $x_3$ along $P(s_1,t_3)$
till it meets the grail $\set{N}_1$ without conflicts in the existing coding. The first node in the light-shaded part
of $\set{N}_1$ which is also on $P(s_1,t_3)$
can clearly compute $x_1+x_2+x_3$ and communicate this to $t_3$.
This completes the proof for Case 1.1.1.
As an illustrative example, in Fig.~\ref{fig:example1},
we show the complete XOR coding solution for the case where the grail $\set{N}_1$
is the one in Fig.~\ref{fig:multigrail3-b}.
\begin{figure}[h]
\centering
\includegraphics[width=2.4in]{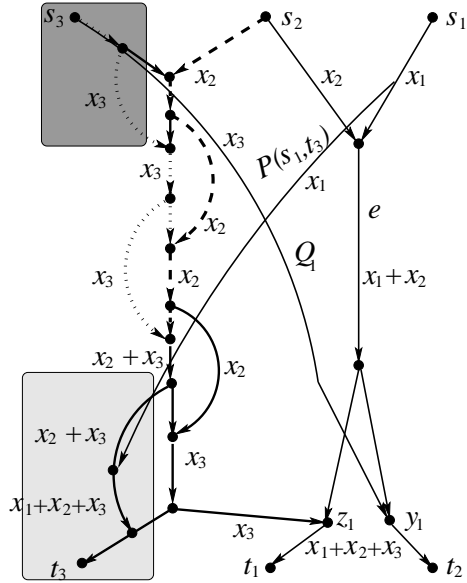}
\caption{An illustration of the coding scheme for Case 1.1.1 of Lemma~\ref{lemma:general}}
\label{fig:example1}
\end{figure}

{\it Case 1.1.2: For any $Q_j \in \mathscr{Q}(s_3,t_2)$, removing all
the edges on $Q_j(s_3:y_j)$ disconnects $(s_1,t_3)$.}

Since for any $j$, removing all the edges on $Q_j$ disconnects
$(s_1,t_3)$, and no single
edge in the network disconnects both $(s_1,t_3)$ and
$(s_3,t_2)$ (by hypothesis (b) of the lemma), 
by Lemma~\ref{lemma:algo0} we can transmit $x_1$
to $t_3$ and $x_3$ to $t_2$. The sub-network (say $\set{N}_2$) used
for this purpose is a
grail with either even or odd number of handles like those in
Fig.~\ref{fig:multigrail4}, and w.l.o.g., let us assume that the $(s_3,t_2)$
path taking part in the grail is $Q_2$. By this coding (on the grail),
$y_2$ receives $x_3$ and $t_3$ receives $x_1$.

\begin{figure*}[!t]
\centering
\subfigure[]{
\includegraphics[width=1.2in]{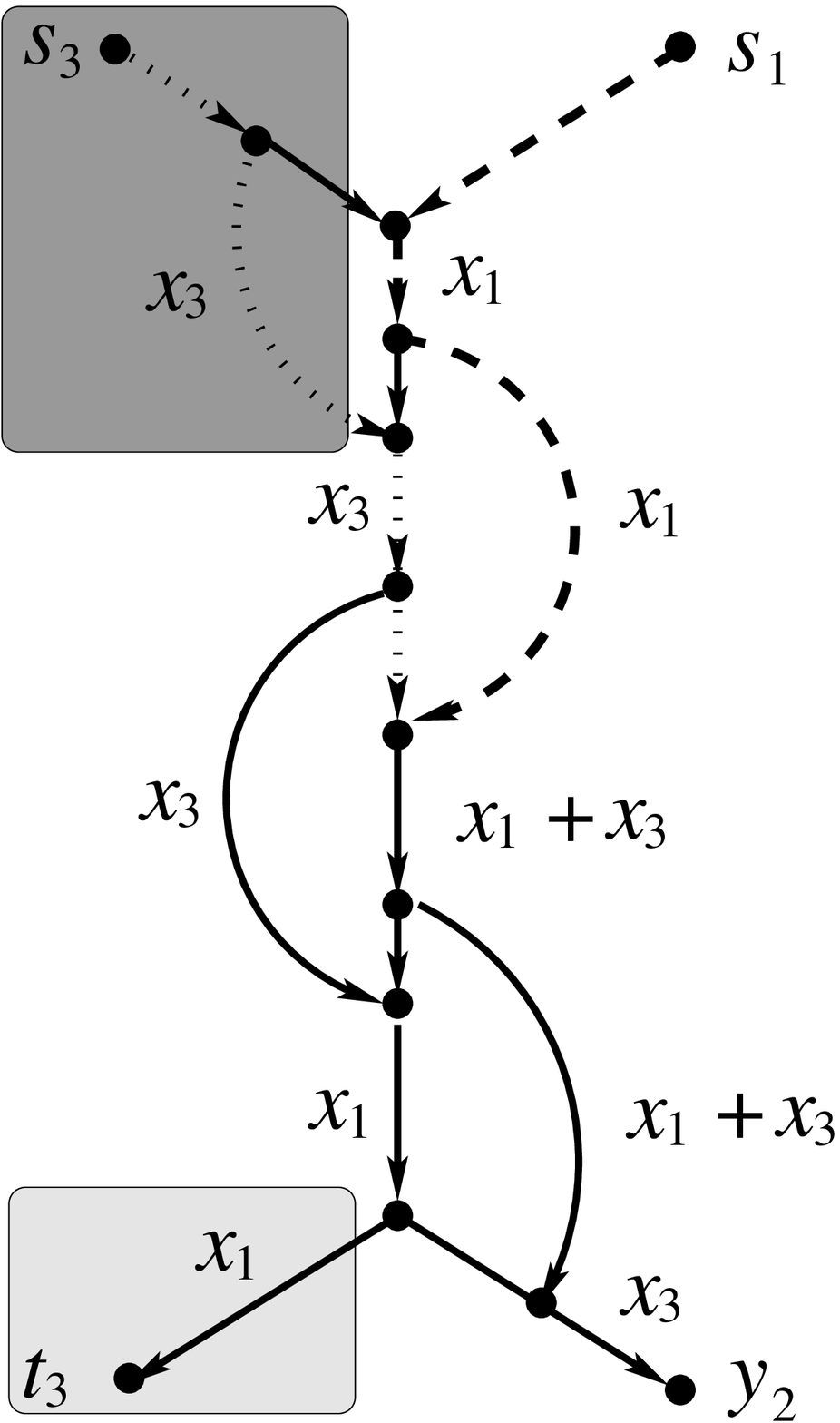}
\label{fig:multigrail4-a}
}
\subfigure[]{
\includegraphics[width=1.2in]{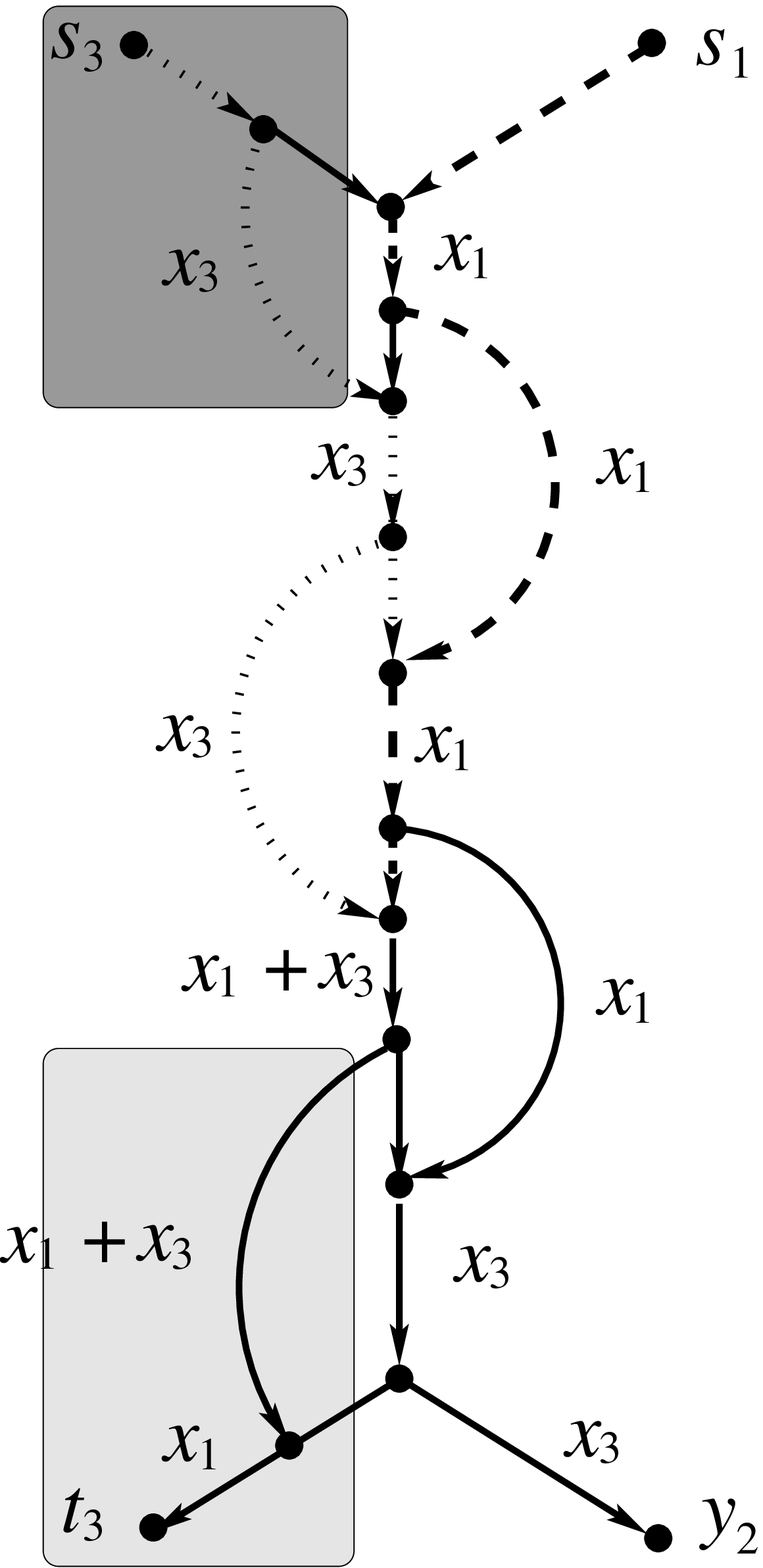}
\label{fig:multigrail4-b}
}
\subfigure[]{
\includegraphics[width=1.2in]{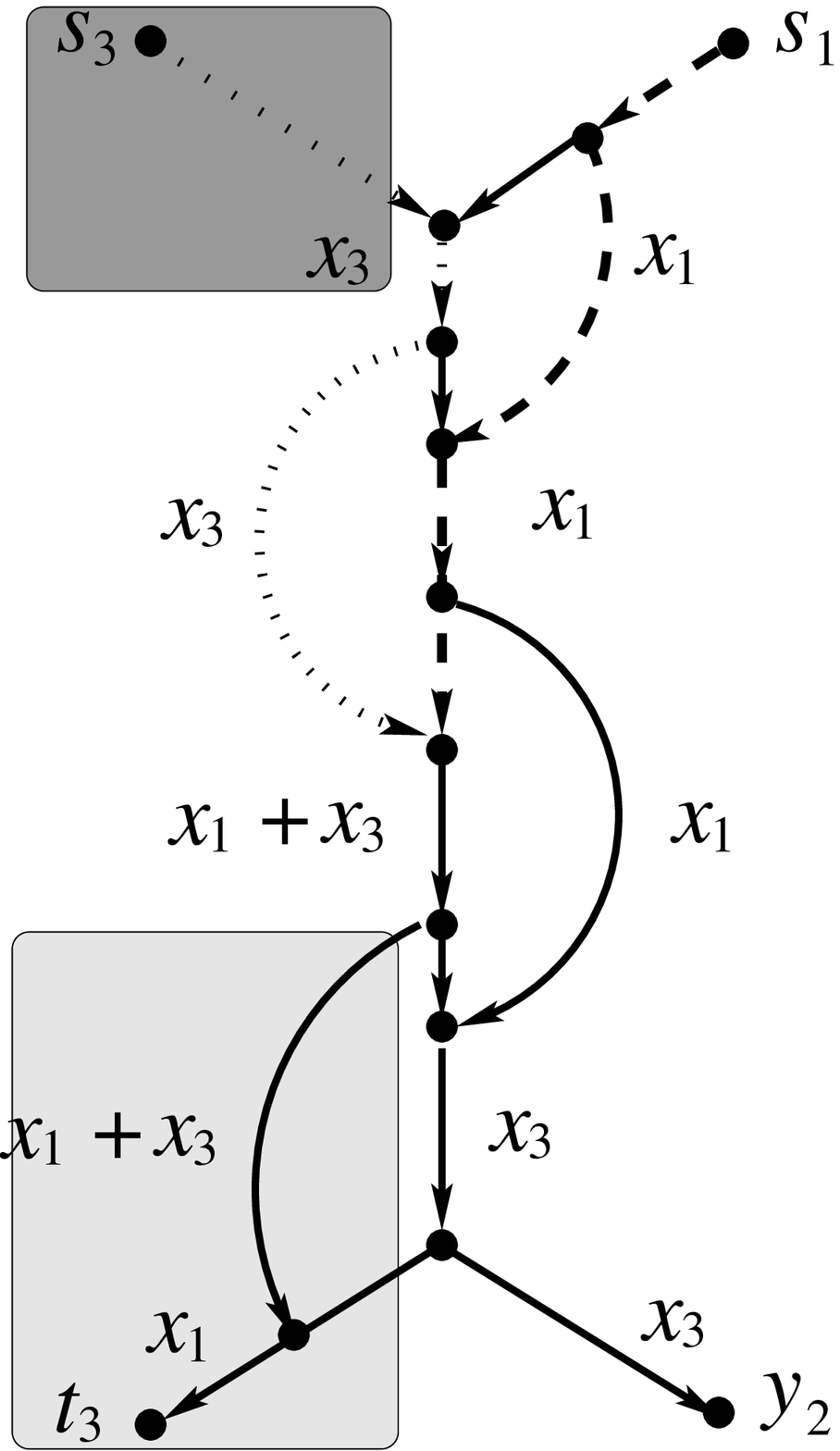}
\label{fig:multigrailo4-a}
}
\subfigure[]{
\includegraphics[width=1.2in]{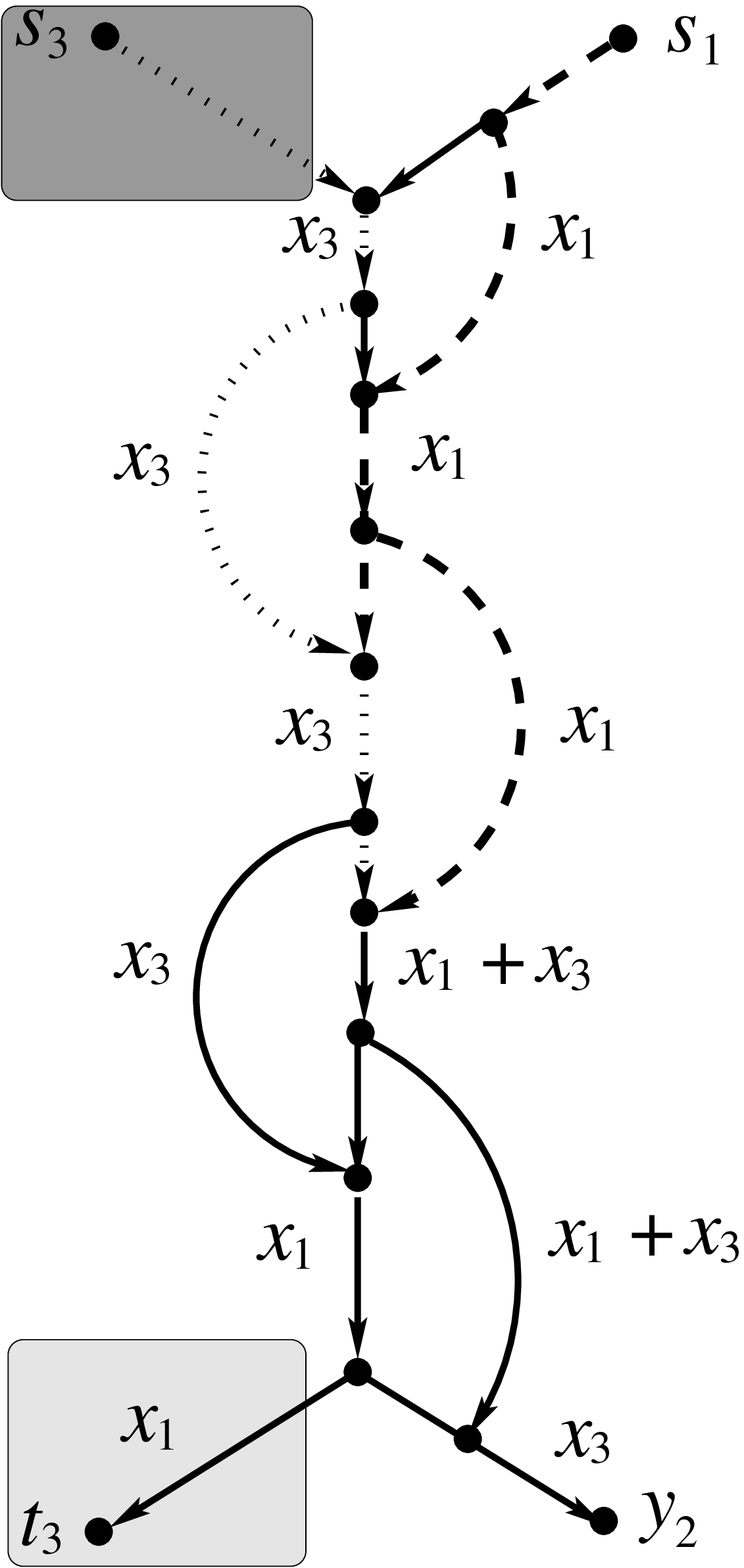}
\label{fig:multigrailo4-b}
}
\caption[]{The coding on grail $\set{N}_2$ for Case 1.1.2 of Lemma~\ref{lemma:general}}
\label{fig:multigrail4}
\end{figure*}

Because of Observation~\ref{observation:3}, the grails $\set{N}_1$ and $\set{N}_2$ can only intersect
in the following possible ways: (i) The dark-shaded part of $\set{N}_1$ intersects with the dark-shaded
part of $\set{N}_2$ and/or
(ii) The light-shaded part of $\set{N}_1$ intersects with the light-shaded part of $\set{N}_2$.
The remaining parts of the grails are node-disjoint. Now the dark-shaded parts of both the grails carry $x_3$,
so such an intersection does not cause any conflict. As for the intersection
between the light-shaded parts of the two grails, it can be easily worked out
that in all possible cases, $t_3$ can easily recover $x_1+x_2+x_3$ by XOR coding.
(One can check that the sub-network formed by the intersection in the
light-shaded parts enables communication of $x_1+x_2+x_3$ to $t_3$ basically
as the sum of some of the inputs to that part - and this is always
feasible in a connected 1-terminal network.)
By the inferences in Observations~\ref{observation:1a},\ref{observation:1b}, we can simultaneously communicate $x_1+x_2$ to $z_1$ and $y_2$ via $e$.
Then $t_1$ and $t_2$ get $x_1+x_2+x_3$ from $z_1$ and $y_2$ respectively.
This completes the proof for Case 1.1.2.
As an illustrative example, in Fig.~\ref{fig:example2},
we show the complete XOR coding solution for the case where the grail $\set{N}_1$
is the one in Fig.~\ref{fig:multigrail3-b} and the grail $\set{N}_2$
is the one in Fig.~\ref{fig:multigrail4-b} and their light-shaded parts and dark-shaded parts intersect
as shown in the figure.
\begin{figure}[h]
\centering
\includegraphics[width=2.7in]{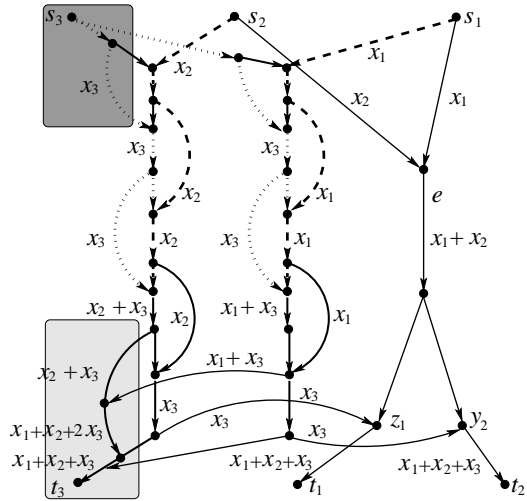}
\caption{An illustration of the coding scheme for Case 1.1.2 of Lemma~\ref{lemma:general}}
\label{fig:example2}
\end{figure}

{\it Case 1.2: For any $P_i \in \mathscr{P}(s_3,t_1)$, removing all
the edges on $P_i(s_3:z_i)$ disconnects $(s_1,t_3)$ but not $(s_2,t_3)$.}

This is the symmetric counterpart of Case 1.1, and the proof is skipped.

{\it Case 1.3: There exist paths $P_1,P_2\in \mathscr{P}(s_3,t_1)$ such
that removing all the edges on $P_1(s_3:z_1)$
disconnects $(s_1,t_3)$ and removing all the edges on $P_2(s_3:z_2)$
disconnects $(s_2,t_3)$.}

By the hypothesis of the case, any $(s_1,t_3)$ (resp. $(s_2,t_3)$) path
shares common edges with $P_1$ (resp. $P_2$) (the reader may like to keep Fig.~\ref{fig:case0} in mind).
If any such $(s_1,t_3)$ (resp. $(s_2,t_3)$) path
shares nodes with $P_2$ (resp. $P_1$), then $P_2$ (resp. $P_1$) has a node $v$
s.t. $\{s_1,s_2,s_3\} \rightarrow v \rightarrow \{t_1,t_3\}$, i.e.,
$\Gamma^{s_1,s_2,s_3}_{t_1,t_3} \neq \emptyset$,
which contradicts Assumption~\ref{assumption:simplify}.
So this is not the case.

Since the network is connected, there
is a $(s_3,t_2)$ path, say $Q_1$. So, under this case, we have a subnetwork
as shown in Fig.~\ref{fig:case0}. Now, again by Assumption~\ref{assumption:simplify}, one can easily
verify that $Q_1$ does not share a node with the rest of the
subnetwork except on the path-segments $P_1(s_3:v_1)$ above $v_1$,
$P_2(s_2:v_2)$ above $v_2$ and the $(\head{e},t_2)$ path-segment below $\head{e}$.
So the coding
scheme shown in Fig.~\ref{fig:case0} completes the proof of this case.
In particular, $t_1$ and $t_3$ use $x_1$
obtained from $P_1$ and
$x_2+x_3$ obtained from $P_2$ to get $x_1+x_2+x_3$, while $t_2$ uses $x_3$ obtained from $Q_1$ and $x_1+x_2$ obtained
from $e$ to get $x_1+x_2+x_3$.

\begin{figure}[h]
\centering
\includegraphics[height=2.2in]{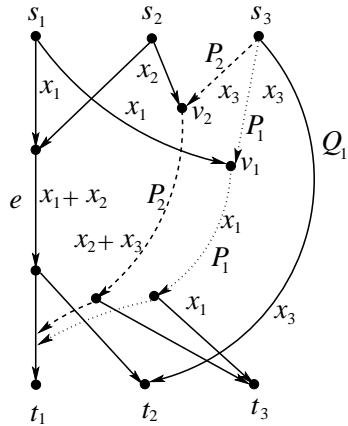}
\caption{The sub-network and code for Case 1.3 of Lemma~\ref{lemma:general}}
\label{fig:case0}
\end{figure}

{\it Case 2: There exists some path in $\mathscr{P}(s_3,t_1)$, say $P_1$, such that
removing all the edges on $P_1(s_3:z_1)$ does not disconnect either $(s_1,t_3)$ or $(s_2,t_3)$.}

We can have the following sub-cases under Case 2:

{\it Case 2.1: For any $Q_j \in \mathscr{Q}(s_3,t_2)$, removing all
the edges on $Q_j(s_3:y_j)$ disconnects $(s_1,t_3)$ and/or $(s_2,t_3)$.}

This subcase statement is the symmetric counterpart of the Case 1 statement.
Case 2.1 is thus a special case (because of the additional constraint
in the Case 2 statement) of that symmetric counterpart and so
the proof follows in a similar way.

{\it Case 2.2:
For some path in $\mathscr{Q}(s_3,t_2)$, say $Q_1$, removing all the
edges on $Q_1(s_3:y_1)$
does not disconnect $(s_1,t_3)$ or $(s_2,t_3)$.}

This is considered in two further sub-cases:

{\it Case 2.2.1:
Removing all the edges on the pair of paths $P_1(s_3:z_1)$ and $Q_1(s_3:y_1)$
simultaneously does not disconnect $(s_1,t_3)$ or $(s_2,t_3)$.}

By the hypothesis
of Case 2.2.1, there exist $(s_1,t_3)$ and $(s_2,t_3)$ paths, called respectively
$P(s_1,t_3)$ and $P(s_2,t_3)$, which are edge-disjoint from both $P_1$ and $Q_1$.
Consider a lowest node $v$ in ancestral order in the set 
$\Gamma_{t_3}\cap (P_1 \cup Q_1)$.
Such a node is above $z_1$ or $y_1$ by Observation~\ref{observation:2}.
W.l.o.g., let us assume that $v$ is on $P_1$, and $P'$ is a $(v,t_3)$ path.
Now, $P(s_3,t_3) = P_1(s_3:v)P'$ is a $(s_3,t_3)$ path.
By the hypothesis
of Case 2.2.1, one can communicate
$x_1+x_2+x_3$ to $t_3$ via XOR coding on $P(s_1,t_3),P(s_2,t_3)$ and $P(s_3,t_3)$, while simultaneously
communicating $x_3$ on $P_1(s_3:z_1)$ and $Q_1(s_3:y_1)$.
Further, by
Observations~\ref{observation:1a} and \ref{observation:1b}, one can also communicate
$x_1+x_2$ to $z_1$ and $y_1$ via $e$.
Nodes $z_1$ and $y_1$
can then recover $x_1+x_2+x_3$
and transmit this to $t_1$ and $t_2$ respectively.

{\it Case 2.2.2:
Removing all the edges on the pair of paths $P_1(s_3:z_1)$ and $Q_1(s_3:y_1)$
simultaneously disconnects $(s_1,t_3)$ or $(s_2,t_3)$ or both.}
\begin{figure}[h]
\centering
\includegraphics[height=2.2in]{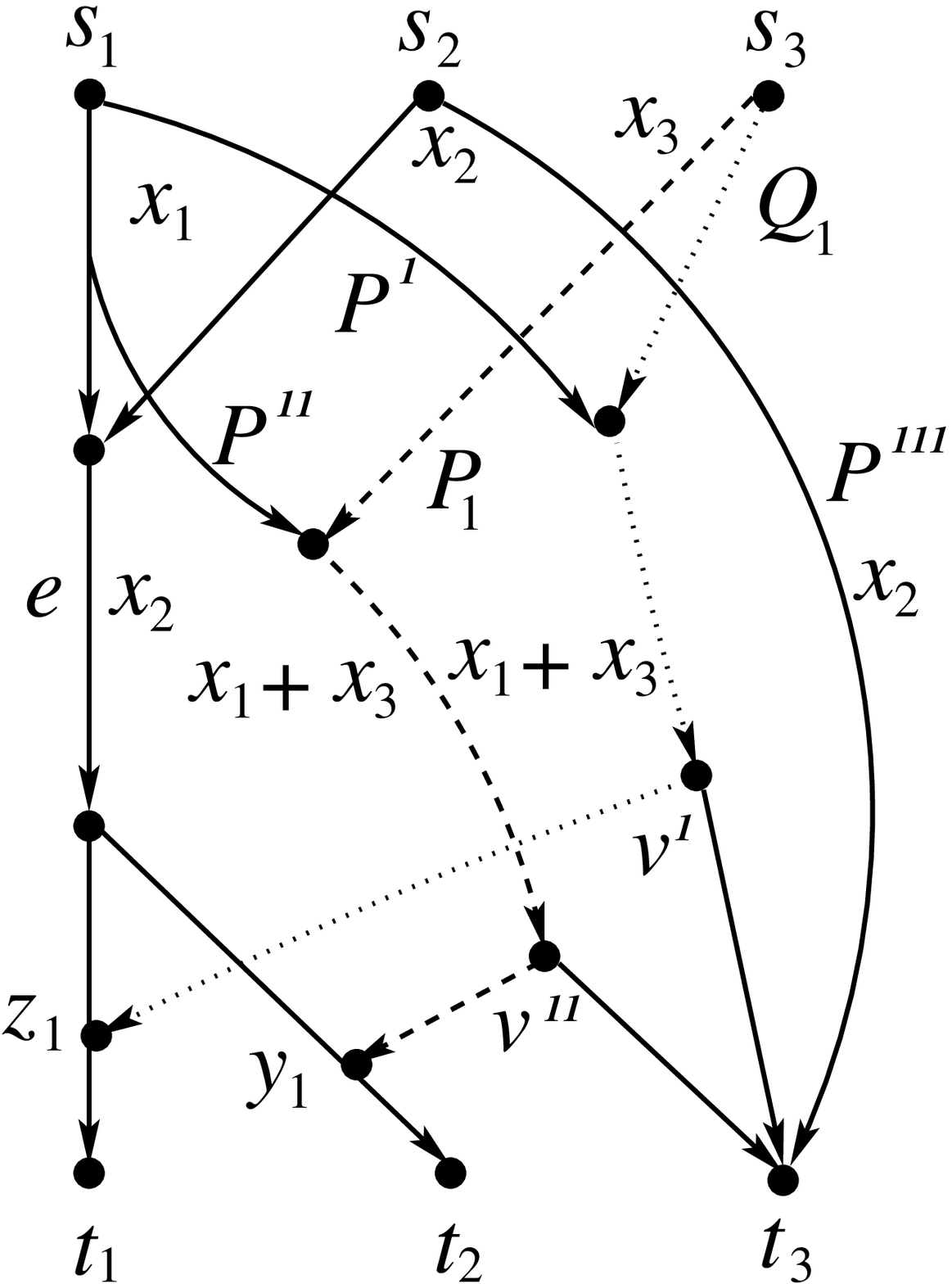}
\caption{The sub-network and code for Case 2.2.2 Lemma~\ref{lemma:general}}
\label{fig:case0'}
\end{figure}

W.l.o.g., let us assume that removing all the edges on $P_1(s_3:z_1)$ and $Q_1(s_3:y_1)$ simultaneously
disconnects $(s_1,t_3)$. By the hypothesis of the case, in the network formed
by removing all the edges on $P_1(s_3:z_1)$, removing all the edges on $Q_1(s_3:y_1)$
disconnects $(s_1,t_3)$. Hence there exists a $(s_1,t_3)$ path $P'$ which shares edges with
$Q_1(s_3:y_1)$. By Assumption~\ref{assumption:simplify}, $P'$ is node-disjoint from $P_1$ and $Q_1(y_1:t_2)$
since otherwise, $\Gamma^{s_1,s_3}_{t_1,t_2,t_3}\neq \emptyset$.
By similar reasoning, there exists a $(s_1,t_3)$ path $P''$ which shares edges with
$P_1(s_3:z_1)$ but is node-disjoint from $Q_1$ and $P_1(z_1:t_1)$.
This and Observation~\ref{observation:1a}(ii),(iii),(iv) imply
that there exists a subnetwork as shown in Fig.~\ref{fig:case0'}. Here $P'''$
is any $(s_2,t_3)$ path.
By Assumption~\ref{assumption:simplify} and Observation~\ref{observation:1a}(i), it can be verified that
$P'''$ does not share any node with the rest of the sub-network except on the path segments
$P'(v':t_3)$ below $v'$, $P''(v'':t_3)$ below $v''$, and the $(s_2,\tail{e})$ path segment above $\tail{e}$.
So the coding scheme shown in Fig.~\ref{fig:case0'} completes the proof of this case. (The reader may note that
the subnetwork in Fig.~\ref{fig:case0'} is actually the reverse network of the one in Fig.~\ref{fig:case0}.)

{\bf Proof of Lemma~\ref{lemma:nopair}:}

Let $\set{N}$ satisfy the hypotheses of Lemma~\ref{lemma:nopair} and consider the edges
$e_1,e_2$ and the labeling of the sources and terminals with which hypothesis (2)
of Lemma~\ref{lemma:nopair} is satisfied.
By hypothesis 2(d), the set $\set{R}(s_3,t_3)=\{R_1,R_2,...\}$ of all $(s_3,t_3)$ paths that do not contain
either $e_1$ or $e_2$ is not empty.
W.lo.g., let $\set{N}$ satisfy Assumption~\ref{assumption:simplify}.
Because $e_1,e_2$ satisfy hypothesis (2), and because no path in $\set{R}(s_3,t_3)$
contains either of them, we have
\begin{observation}\label{observation:1.1}
No $(s_3,t_3)$ path in $\set{R}(s_3,t_3)$ contains a node $v$ such that $s_i\rightarrow v$
or $v\rightarrow t_j, i,j\in\{1,2\}$ i.e. $s_1,s_2,t_1,t_2$ are not connected to any path
in $\set{R}(s_3,t_3)$.
\end{observation}
This means that as shown in Fig.~\ref{fig:nopair1}, any $R_1\in\set{R}(s_3,t_3)$
does not share any node with the rest of the sub-network except on the $(s_3,\tail{e_1})$ path-segment above $\tail{e_1}$,
the $(s_3,\tail{e_2})$ path segment above $\tail{e_2}$, the $(\head{e_1},t_3)$ path segment below $\head{e_1}$
and the $(\head{e_2},t_3)$ path segment below $\head{e_2}$.
\begin{figure}[h]
\centering
\includegraphics[width=1.25in]{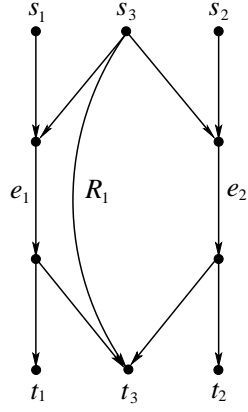}
\caption{The sub-network and code for Case 1 of Lemma~\ref{lemma:nopair}}\label{fig:nopair1}
\end{figure}
Let $\set{P}(s_1,t_2)=\{P_1,P_2,...\}$ be the set of all $(s_1,t_2)$ paths
and let $\set{Q}(s_2,t_1)=\{Q_1,Q_2,...\}$ be the set of all $(s_2,t_1)$ paths.
\begin{observation}\label{observation:1'}
By Assumption~\ref{assumption:simplify},

(i) No $(s_1,t_2)$ path contains any node from $\Gamma^{s_3}_{\tail{e_1}}\cup
\Gamma^{s_3}_{\tail{e_2}}\cup \Gamma^{s_2}_{\tail{e_2}}\cup\{\tail{e_1}\}\cup\{\tail{e_2}\}
\cup  \Gamma^{\tail{e_1}}\cup  \Gamma^{\tail{e_2}}_{t_2,t_3}$.

(ii) No $(s_2,t_1)$ path contains any node from $\Gamma^{s_3}_{\tail{e_1}}\cup
\Gamma^{s_3}_{\tail{e_2}}\cup \Gamma^{s_1}_{\tail{e_2}}\cup\{\tail{e_1}\}\cup\{\tail{e_2}\}
\cup  \Gamma^{\tail{e_1}}_{t_1,t_3}\cup  \Gamma^{\tail{e_2}}$.

Observation~\ref{observation:1.1} gives,

(iii) any path in $\set{R}(s_3,t_3)$ is node-disjoint from any path in $\set{P}(s_1,t_2)$ or $\set{P}(s_2,t_1)$.
\end{observation}
\begin{figure*}[!t]
\centering
\subfigure[]{
\includegraphics[height=2.1in]{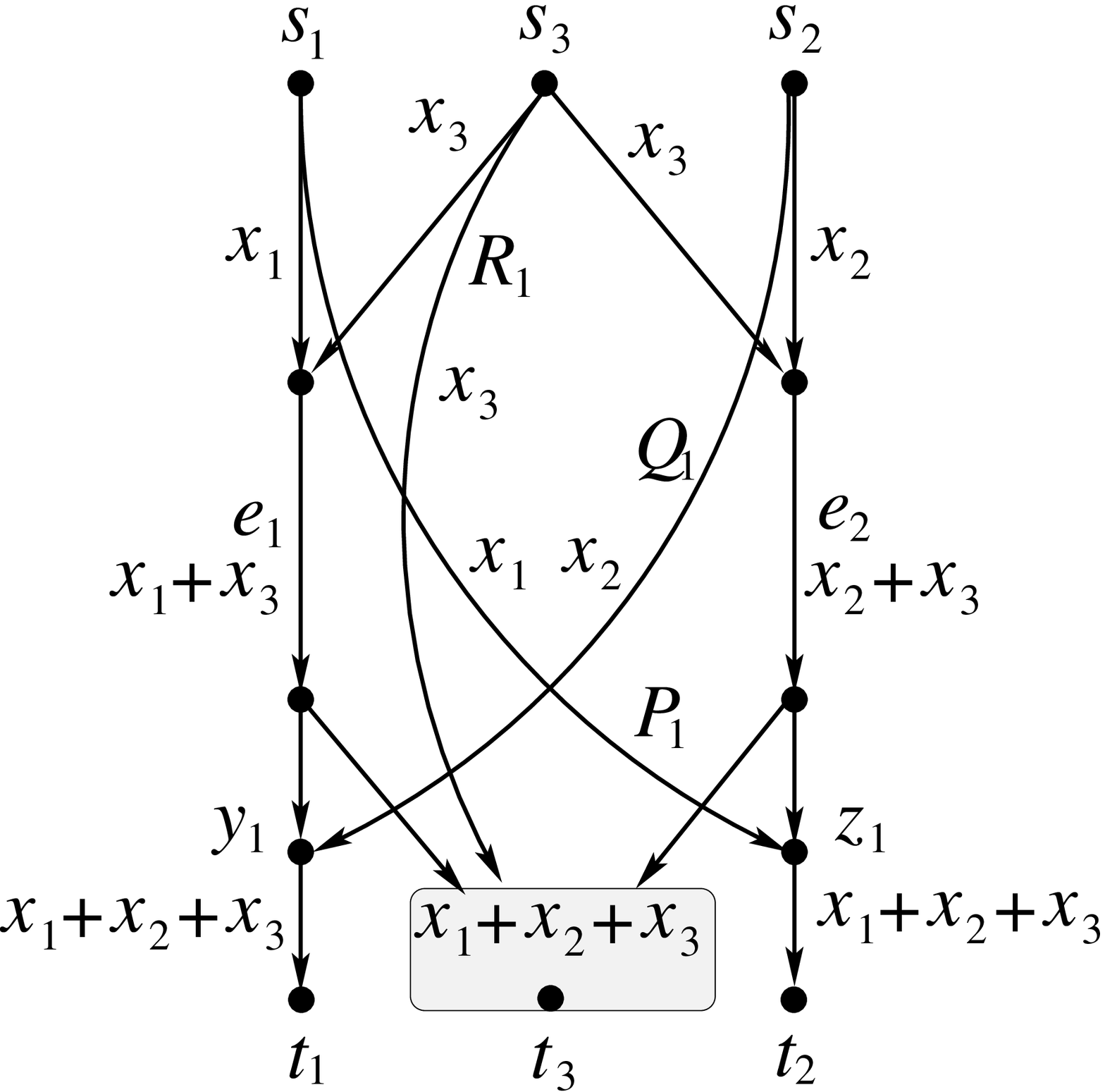}
\label{fig:not5}
}
\subfigure[]{
\includegraphics[height=2.1in]{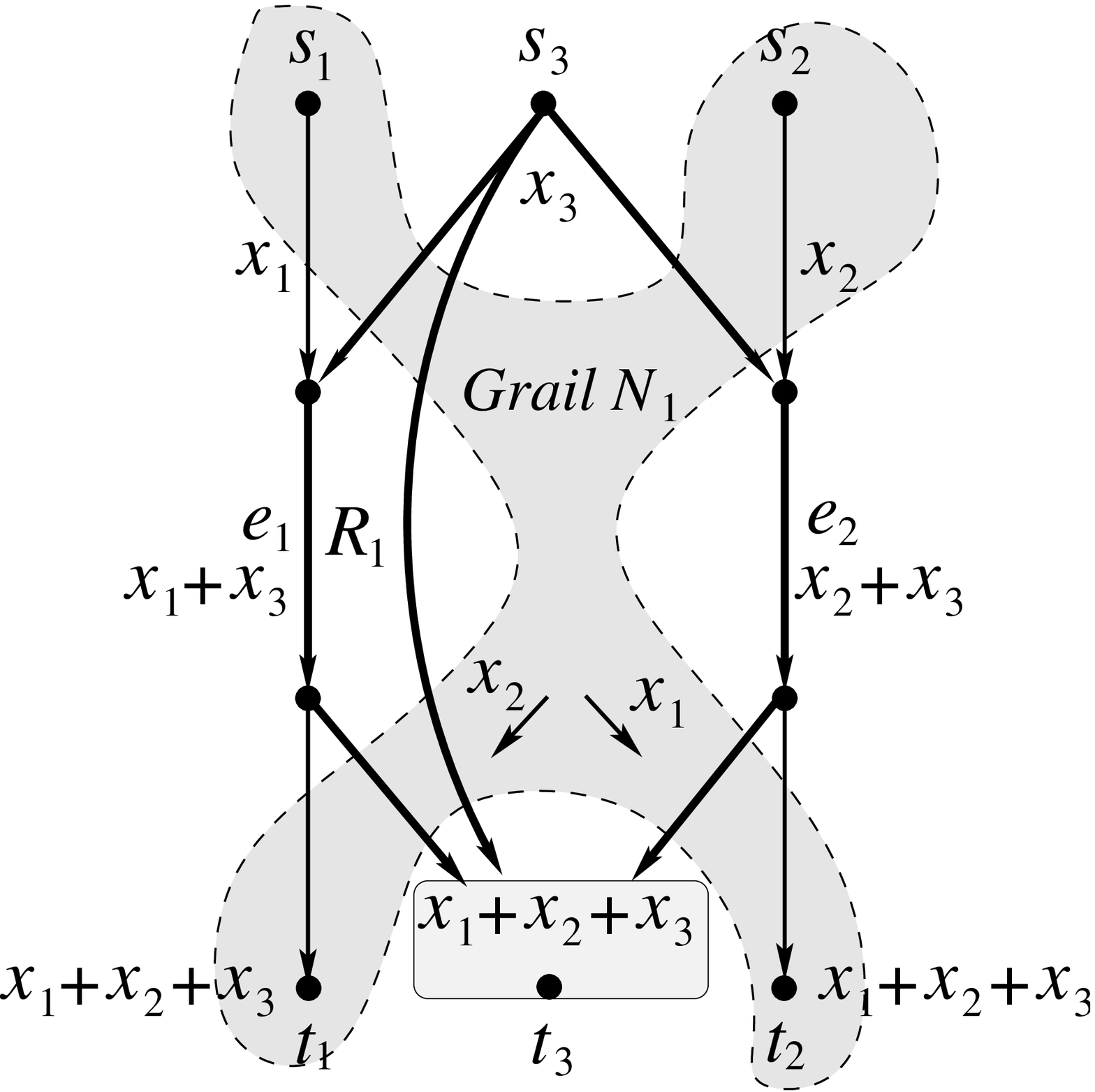}
\label{fig:nopair2}
}
\subfigure[]{
\includegraphics[height=2.1in]{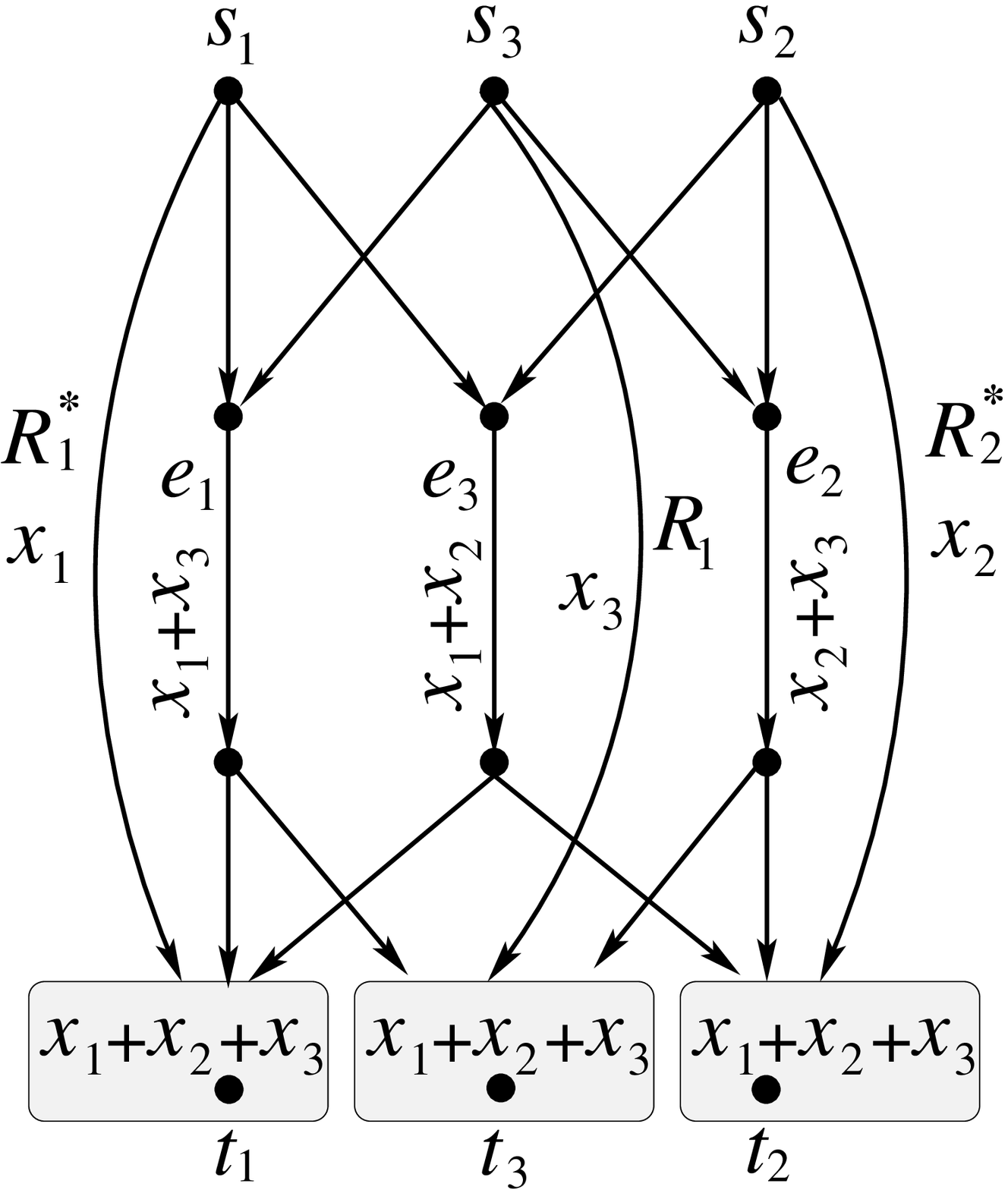}
\label{fig:not3}
}
\caption[]{The coding for the different cases under Lemma~\ref{lemma:nopair}.
The shaded rectangle containing a terminal
is drawn to mean that irrespective of the order in which the three incoming paths
meet, this terminal can always recover
$x_1+x_2+x_3$ by XOR coding.
}
\label{fig:nopair}
\end{figure*}
For any $P_i\in \set{P}(s_1,t_2)$ let $z_i$ be the first descendant of $e_2$ on $P_i$.
Similarly, for any $Q_j\in \set{Q}(s_2,t_1)$ let $y_j$ be the first descendant of $e_1$ on $Q_j$.
We consider two cases:

{\it Case 1: There exist $P_1\in \set{P}(s_1,t_2)$ and $Q_1\in\set{Q}(s_2,t_1)$
such that $P_1$ and $Q_1$ are edge-disjoint.}

In this case there exists a subnetwork shown in
Fig.~\ref{fig:not5} because, by Observation~\ref{observation:1'}, we have that (i)
$P_1$ does not share any node with the rest of the subnetwork
except on the $(s_1,\tail{e_1})$ path segment above $\tail{e_1}$ and
the  $(\head{e_2},t_2)$ path segment below $\head{e_2}$, and (ii) 
$Q_1$ does not share any node with the rest of the subnetwork
except on the $(s_2,\tail{e_2})$ path segment above $\tail{e_2}$ and
the  $(\head{e_1},t_1)$ path segment below $\head{e_1}$. The XOR coding scheme shown in Fig.~\ref{fig:not5}
completes the proof.
The shaded rectangle containing $t_3$
is drawn to mean that irrespective of the order in which the three incoming paths (carrying
$x_1+x_3$, $x_3$ and $x_2+x_3$) meet, $t_3$ can always recover
$x_1+x_2+x_3$ by XOR coding.

{\it Case 2: For any $P_i\in \set{P}(s_1,t_2)$ removing all the edges of $P_i(s_1,z_i)$
disconnects $(s_2,t_1)$.}

{\it Case 2.1: There does not exist a single edge which disconnects both $(s_1,t_2)$ and $(s_2,t_1)$.}

By Lemma~\ref{lemma:algo0}, we can use
a grail subnetwork $\set{N}_1$ to transmit $x_1$ to $t_2$ and $x_2$ to $t_1$.
This grail network $\set{N}_1$  and its coding is as shown in Fig.~\ref{fig:multigrail1} but with
$t_1$ and $t_2$ interchanged.
The situation then is as shown in Fig~\ref{fig:nopair2}. Here the details
of the grail $\net_1$ are suppressed for clarity and it is represented by
a shaded region. 
By Observation~\ref{observation:1.1},
$\set{N}_1$ is node-disjoint from $e_1, e_2,R_1$ and the $(s_3,\tail{e_1}),(s_3,\tail{e_2}),(\head{e_1},t_3)$
and $(\head{e_2},t_3)$ path segments (which are shown with thick edges in the figure).
The coding scheme (shown in the figure) where $e_1$ is used to communicate $x_1+x_3$ to $t_1,t_3$,
$e_2$ is used to communicate $x_2+x_3$ to $t_2,t_3$, $R_1$ is used to communicate $x_3$ to $t_3$
and the grail $\set{N}_1$ used to communicate $x_1$ to $t_2$ and $x_2$ to $t_1$ completes the proof for
this case.

{\it Case 2.2: There exists an edge $e_3$ which disconnects both $(s_1,t_2)$ and $(s_2,t_1)$.}

In this case we will show the existence of a subnetwork as shown in Fig.~\ref{fig:not3}.
It is easy to see by Assumption~\ref{assumption:simplify} that,
$e_1\nrightarrow e_3,e_2\nrightarrow e_3,e_3\nrightarrow e_1,e_3\nrightarrow e_2.$
Now, removing the pair $e_1,e_3$ does not disconnect $(s_1,t_1)$
since otherwise $\set{N}$ would satisfy the hypothesis of Theorem~\ref{theorem:mainres2}
for some labelling of the sources and terminals.
For the same reason, removing the pair $e_2,e_3$ does not disconnect $(s_2,t_2)$.
Hence there exists a $(s_1,t_1)$ path $R_1^*$ not containing $e_1$
or $e_3$ and a $(s_2,t_2)$ path $R_2^*$ not containing $e_2$ or $e_3$.

Because of the conditions that $e_1,e_2,e_3$ satisfy, we have
(i) there is no node $v$ on $R_1^*$ satisfying either $s_2\rightarrow v$, $s_3\rightarrow v$,
$v\rightarrow t_2$ or $v\rightarrow t_3$ and  (ii) there is no node $v$ on $R_2^*$ satisfying either $s_1\rightarrow v$, $s_3\rightarrow v$,
$v\rightarrow t_1$ or $v\rightarrow t_3$.

This and Observation~\ref{observation:1'} imply the existence of the subnetwork 
subnetwork shown in Fig.~\ref{fig:not3} such that

(i) $R_1^*$ does not share any node with the rest of the subnetwork
except on the $(s_1,\tail{e_1})$ path segment above $\tail{e_1}$,  the $(s_1,\tail{e_3})$ path segment above $\tail{e_3}$,
the $(\head{e_1},t_1)$ path segment below $\head{e_1}$, 
and the $(\head{e_3},t_1)$ path segment below $\head{e_3}$; 
and

(ii) $R_2^*$ does not share any node with the rest of the subnetwork
except on the $(s_2,\tail{e_2})$ path segment above $\tail{e_2}$,  the $(s_2,\tail{e_3})$ path segment above $\tail{e_3}$,
the  $(\head{e_2},t_2)$ path segment below $\head{e_1}$
and the  $(\head{e_3},t_2)$ path segment below $\head{e_3}$.

The XOR code shown in Fig.~\ref{fig:not3} completes the proof.

{\bf Proof of Lemma~\ref{lemma:suffcon}:}

We will only prove part A  of Lemma~\ref{lemma:suffcon}.
Since our proof is constructive, the proof of the other parts will follow
from the coding solutions offered in the proof of part A.

If $\kappa(\set{N})\geq 5$, $\set{N}$ is XOR solvable over any field
by Lemma~\ref{lemma:5con}. Hence in the remaining part of this proof,
we only consider networks with $\kappa(\set{N})=0,1,2$ and $4$
and prove Lemma~\ref{lemma:suffcon} for each value. In the light of Lemma~\ref{lemma:5con},
it is enough to prove Lemma~\ref{lemma:suffcon} for networks
satisfying Assumption~\ref{assumption:simplify}. Then if $\set{C}\neq \emptyset$, it only contains
maximum-disconnecting edges such that there is a path from exactly two sources to its tail
and there is a path from its head to exactly two terminals.

$ \bullet{}\boldsymbol{\kappa(\set{N})=0:}$ In this case,
there exist two edge-disjoint paths between
each source-terminal pair. The main result of~\cite{langberg3}
is that such a sum-network is solvable over fields of odd characteristic.
In the following, we present a significantly different proof which
also gives a stronger result that such a network is solvable over {\it any}
field by a {\it XOR code}.

We consider two cases depending on whether or not
$\mathscr{C} =\emptyset$:

\begin{claim}{Sum-networks with $\kappa=0$ and $\mathscr{C} =\emptyset$ are
XOR solvable over any field.}\label{claim:0class1}
\end{claim}
{\it Proof:}
The proof follows by Lemma~\ref{lemma:con1}.

\begin{claim}{Sum-networks with $\kappa=0$ and $\mathscr{C} \neq \emptyset$ are
XOR solvable over any field.}\label{claim:0class2}
\end{claim}
{\it Proof:}
The proof follows by Lemma~\ref{lemma:general} since hypothesis (a) of the lemma follows from $\set{C}\neq \emptyset$
with suitable labeling of the sources and the terminals and hypothesis (b)
follows from $\kappa =0$.

$ \bullet{}\boldsymbol{\kappa(\set{N})=1:}$ 
We consider two cases depending on whether or not
$\mathscr{C} =\emptyset$:

\begin{claim}{Sum-networks with $\kappa=1$ and $\mathscr{C} =\emptyset$ are
XOR solvable over any field.}\label{claim:1class1}
\end{claim}
{\it Proof:}
The proof follows from Lemma~\ref{lemma:reduction} since we have proved that networks with $\kappa=0$
are XOR solvable over any field.

\begin{claim}\label{claim:k1}{Sum-networks with $\kappa=1$ and $\mathscr{C} \neq \emptyset$
are XOR solvable over any field.}\label{claim:1class2}
\end{claim}
{\it Proof:}
The proof follows by Lemma~\ref{lemma:general} since hypothesis (a) of the lemma follows from $\set{C}\neq \emptyset$
with suitable labeling of the sources and the terminals and hypothesis (b)
follows from $\kappa =1$.

$ \bullet{}\boldsymbol{\kappa(\set{N})=2:}$ 
We will prove that networks with $\kappa=2$ are linearly solvable.
But solvability in this case is not necessarily over $F_2$, and even over other
fields, the solvability may not be by XOR coding.
Specifically, this happens only in Case 1.2 under $\set{C} \neq \emptyset$
(Claim~\ref{claim:noXOR}).

We consider two cases depending on whether or not $\mathscr{C} =\emptyset$:

\begin{claim}{Sum-networks with $\kappa=2$ and $\mathscr{C} =\emptyset$ are
XOR solvable over any field.}\label{claim:2class1}
\end{claim}
{\it Proof:} The proof follows from Lemma~\ref{lemma:reduction}
since we have proved that networks with $\kappa=0,1$
are XOR solvable over any field.

\begin{claim}{Sum-networks with $\kappa=2$ and $\mathscr{C} \neq \emptyset$
are linearly solvable over all fields except possibly $F_2$.}\label{claim:2class2}
\label{claim:noXOR}
\end{claim}
{\it Proof:}
 Let us assume that $e \in \mathscr{C}$ and that
$\{s_1,s_2\} \rightarrow e \rightarrow \{t_1,t_2\}$. 

Now $e$ can disconnect two source-terminal pairs in essentially three different
ways.
It can disconnect either $(s_1,t_1)$ and $(s_2,t_2)$, $(s_1,t_1)$ and $(s_2,t_1)$, or $(s_1,t_1)$ and $(s_1,t_2)$.
We consider each case in turn.

\textit{Case 1: Edge $e$ disconnects $(s_1,t_1)$ and $(s_2,t_2)$.}

\textit{Case 1.1: There does not exist an edge disconnecting
either $(s_2,t_3)$ and $(s_3,t_1)$; or $(s_1,t_3)$ and $(s_3,t_2)$.}

The network is easily seen to satisfy the hypotheses of Lemma~\ref{lemma:general} and
is thus XOR solvable over any field.

\textit{Case 1.2: There exists an edge $e'$ disconnecting
either $(s_2,t_3)$ and $(s_3,t_1)$; or $(s_1,t_3)$ and $(s_3,t_2)$.}

{\it Note:} The networks shown in Fig.~\ref{fig:nonbinarysolvable} fall under this case, and
they are not solvable over $F_2$ but linearly solvable over any other
field though not by XOR coding~\cite{RaiD:09b}.

In this case, $e,e'$ satisfy conditions $1,2$ in Theorem~\ref{theorem:mainres2}
for a suitable relabeling of the sources since $\kappa=2$,
and condition $4$ by Assumption~\ref{assumption:simplify}.
Thus, if there does not exist an edge pair satisfying
all the four conditions in Theorem~\ref{theorem:mainres2}, then by taking $e,e'$ as
$e_1,e_2$ under a suitable relabeling of the sources and terminals, the hypotheses of
Lemma~\ref{lemma:nopair} are satisfied, and the network is XOR solvable over any field.
If on the other hand, there exists an edge pair satisfying
all the four conditions in Theorem~\ref{theorem:mainres2}, then by the sufficiency part of Theorem~\ref{theorem:mainres2}
(proved independently later, see the dependency graph in Fig.~\ref{fig:dependence}), the network is not solvable over $F_2$ but
linearly solvable over all other fields.

\textit{Case 2: Edge $e$ disconnects $(s_1,t_1)$ and $(s_2,t_1)$}

We assume that
there is no maximum-disconnecting edge disconnecting $(s_i,t_j)$
and $(s_{i'},t_{j'})$, $i\neq i', j\neq j'$ for any labeling of the sources and terminals, since otherwise such an edge satisfies {\it Case 1} and the
proof follows from the proof of that case.
So there does not exist another edge $e'$ which disconnects $(s_2,t_3)$ and $(s_3,t_1)$;
or $(s_1,t_3)$ and $(s_3,t_2)$. Thus the hypotheses of Lemma~\ref{lemma:general} are satisfied.
and the network is linearly solvable over any field using XOR coding
in this case.

\textit{Case 3: Edge $e$ disconnects $(s_1,t_1)$ and $(s_1,t_2)$}

In this case, the reverse network falls under {\it Case 2}, and is
thus XOR solvable over any field.
The proof then follows by Lemma~\ref{lemma:reverse}.

$ \bullet{}\boldsymbol{\kappa(\set{N})=4:}$ We consider three cases:

\textit{Case 1: A maximum-disconnecting edge $e$ disconnects
one terminal from all the sources, and another
terminal from a single source.}

This case can not occur under Assumption~\ref{assumption:simplify},
since then that edge would be connected to all the sources and two
terminals.

\textit{Case 2: A maximum-disconnecting edge $e$ disconnects
one source from all the terminals, and another
source from a single terminal.}

This case can not occur under Assumption~\ref{assumption:simplify},
since then that edge would be connected to two sources and all the three
terminals.

\textit{Case 3: A maximum-disconnecting edge $e$ disconnects
both $s_1$ and $s_2$ from both $t_1$ and $t_2$.}

Under the given labeling of the sources and terminals, 
hypothesis (a) of Lemma~\ref{lemma:general} is satisfied by $e$. We now show that
hypothesis (b) is also satisfied. If this was not so, let there exist an edge $e'$
which disconnects (w.l.o.g.) $(s_2,t_3)$ and $(s_3,t_1)$. This implies that $s_2\rightarrow e'\rightarrow t_1$
i.e. $s_2\rightarrow t_1$. But since $e$ disconnects $(s_2,t_1)$, we must have $e\rightarrow e'$ (or $e'\rightarrow e$).
Then it is easy to check that 
$s_1,s_2 \rightarrow \head{e} \rightarrow t_1,t_2,t_3$
(or resp. $s_2,s_3 \rightarrow \head{e'} \rightarrow t_1,t_2,t_3$),
which violates Assumption~\ref{assumption:simplify}.
Hence hypothesis (b) of Lemma~\ref{lemma:general} is also satisfied.
Then the network is XOR solvable over any field by the lemma.

{\bf Proof of Lemma~\ref{lemma:converse1}}

Part A of the lemma follows from the sufficiency part of Theorem~\ref{theorem:mainres2}
(proved later independently, see Fig.~\ref{fig:dependence}).

Now we prove part B of the lemma.
The following observation sums up some of the things we have already proved, and which we
will draw upon.
\begin{observation}\label{observation:f}
(i)
If $\kappa(\net ) = 0,1,$ or $\geq 4$, then by Lemma~\ref{lemma:suffcon}, the network
is XOR solvable over any field. 

(ii) If $\kappa(\net ) = 2$, only under Case 1.2 of $\kappa = 2$
in the proof of Lemma~\ref{lemma:suffcon}, the network may not be solvable over $F_2$.
In all the other cases, the network is XOR solvable over any field.
Further, networks under Case 1.2 were shown to be of two types, namely, they either satisfied
the hypothesis of Theorem~\ref{theorem:mainres2} and were not solvable over $F_2$
but linearly solvable over other fields,
or, they did not satisfy the hypothesis of Theorem~\ref{theorem:mainres2}
and were XOR solvable over any field. 
\end{observation}

For networks in part B of the lemma,
consider the network $\net^*$ obtained by
adding parallel edges to the edges in $\set{C}$ as in the proof of Lemma~\ref{lemma:reduction}.
Now, $\kappa (\net^* ) \leq 2$, so, as inferred in Observation~\ref{observation:f},
 $\net^*$ is either XOR
solvable over all fields or $\net^*$ is a $\kappa=2$ network
having an edge pair satisfying conditions 1-4 of Theorem~\ref{theorem:mainres2}
and thus is not solvable over $F_2$ but linearly solvable over other fields.
We will show that $\set{N}^*$ does not contain such an edge pair
and, thus, is XOR solvable over all fields.
Then, as was shown in the proof of Lemma~\ref{lemma:reduction}, $\net$ too will be XOR
solvable over all fields, thereby proving the lemma.

Suppose this was not true, and $\set{N}^*$ has $e_1,e_2$ satisfying conditions 1-4
of Theorem~\ref{theorem:mainres2}.
Since we are only adding parallel edges in the process of constructing $\set{N}^*$ from $\net$,
$e_1,e_2$ satisfy conditions 3 and 4 of Theorem~\ref{theorem:mainres2} in $\net$ itself.
Further, the only way $e_1$ (resp. $e_2$) could dissatisfy condition 1 (resp. 2) of Theorem~\ref{theorem:mainres2} in $\net$ itself
would be if it also disconnected an additional source-terminal pair in $\net$, which will mean that $e_1$ (resp. $e_2$) $\in \set{C}$, thus contradicting
$\set{C} = \emptyset$. Thus $e_1,e_2$ satisfy conditions 1-4
of Theorem~\ref{theorem:mainres2} in $\net$ itself. This gives a contradiction.

{\bf Proof of Lemma~\ref{lemma:converse2}:}

Let $\set{N}$ be a nonsolvable network with $\kappa(\net)=3$ containing an edge
$e_2$ satisfying conditions 3 and 4 of Theorem~\ref{theorem:mainres}.
We will show that the statement of Lemma~\ref{lemma:converse2} holds for $\net$.
We assume that $\set{N}$ satisfies Assumption~\ref{assumption:simplify}
since otherwise the network is XOR solvable over all fields
by Lemma~\ref{lemma:5con}.

We note that there can not exist an edge $e'$ which disconnects $(s_1,t_2)$ and $(s_2,t_1)$ (or $(s_1, t_3)$ and $(s_2, t_1)$, or $(s_1, t_2)$ and
$(s_3, t_1)$ - though we are not using these)
since otherwise $s_2\rightarrow e'\rightarrow t_2$ implies $e'\rightarrow e_2$
or $e_2\rightarrow e'$ (since $e_2$ disconnects $(s_2,t_2)$), any of which contradicts Assumption~\ref{assumption:simplify}.

Then, if there does not exist an edge $e_1$ that satisfies
conditions $1$ and $2$ of Theorem~\ref{theorem:mainres},
$\set{N}$ satisfies the hypotheses in Lemma~\ref{lemma:general} for a suitable relabeling of the
sources and terminals
and so is XOR solvable over any field. So an edge $e_1$
satisfying conditions $1$ and $2$ of Theorem~\ref{theorem:mainres} exists.

Further, $s_1\nrightarrow{\head{e_2}}$ $\&$ $s_1\nrightarrow{\tail{e_2}}$
(by Assumption~\ref{assumption:simplify}) $\&$ $s_1\rightarrow{\tail{e_1}}$ $\Rightarrow$ $e_1\nrightarrow{e_2}$.
Similarly, $s_2\nrightarrow{\head{e_1}}$ $\&$ $s_2\nrightarrow{\tail{e_1}}$ $\&$ $s_2\rightarrow{\tail{e_2}}$ $\Rightarrow$ $e_2\nrightarrow{e_1}$.  
 
Thus, we have so far proved that there exists $e_1,e_2$ satisfying
conditions 1,2,3,4, and 6 of Theorem~\ref{theorem:mainres}.

Now, we argue that
the only way an edge $e'$ disconnecting exactly
$(s_i,t_j)$ and $(s_{i'},t_{j'})$, $i\neq i',j\neq j'$
can exist in this network is if it disconnects
exactly $(s_1,t_3)$ and $(s_3,t_1)$. We prove this using a sequence of
four steps in the following.

(i) If $\{i,i'\} = \{j,j'\} = \{1,3\}$ is not true, then
$e_2 \rightarrow e'$ or $e' \rightarrow e_2$. This is because,
if, w.l.o.g., $i= 2$, since one of $j,j'$ has to be $2$ or $3$,
there is a path from $s_2$ to $t_2$ or $t_3$ via $e'$.
Since $e_2$ disconnects $(s_2,t_2)$ and $(s_2,t_3)$, $e'$
must be an ancestor or descendant of $e_2$.

(ii) If $e'$ is a descendant or ancestor of $e_2$,
then $\{i,i'\} = \{j,j'\} = \{2,3\}$ by Assumption~\ref{assumption:simplify}.

(iii) If $e'$ is a descendant or ancestor of $e_2$,
then $e'$ can not disconnect exactly the source-terminal pairs
$(s_2,t_3)$ and $(s_3,t_2)$ (or
$(s_2,t_2)$ and $(s_3,t_3)$ -this case follows similarly, and will
not be elaborated). Otherwise, after removing $e'$, there exists a $(s_2,t_2)$ path $P$ containing $e_2$. If $e'$ is an ancestor
of $e_2$, then after removing $e'$, $P(s_2:\tail{e_2})$ concatenated
with any $(\tail{e_2}, t_3)$ path gives a $(s_2,t_3)$ path not
containing $e'$ and thus gives a contradiction. Similarly we can reach
a contradiction if $e'$ is a descendant of $e_2$.

(iv) If $\{i,i'\} = \{j,j'\} = \{1,3\}$, then $e'$ can not disconnect
exactly $(s_1,t_1)$ and $(s_3,t_3)$. This follows by similar arguments
as in (iii) above by considering the edge $e_1$.

This proves that the only way an edge $e'$ disconnecting exactly
$(s_i,t_j)$ and $(s_{i'},t_{j'})$, $i\neq i',j\neq j'$
can exist in this network is if it disconnects
exactly $(s_1,t_3)$ and $(s_3,t_1)$. 
Thus an edge pair satisfying conditions 1-4 of Theorem~\ref{theorem:mainres2}
does not exist in this network.
So, the network satisfies
condition 1 of Lemma~\ref{lemma:nopair}.
Now, if $(e_1,e_2)$ do not satisfy condition 5 of Theorem~\ref{theorem:mainres},
then since they satisfy conditions 1,2,3,4, and 6 of Theorem~\ref{theorem:mainres}, they also
satisfy condition 2 in Lemma~\ref{lemma:nopair}.
Thus both the conditions in Lemma~\ref{lemma:nopair} are satisfied and thus the network
is XOR solvable over all fields. Hence for nonsolvability, 
$(e_1,e_2)$ satisfy condition 5 of Theorem~\ref{theorem:mainres} as well.

{\bf Proof of Theorem~\ref{theorem:mainres}:}

The major part of the proof is in two parts. In the Sufficiency part,
we show that once conditions 1)-6) in part A are satisfied by
two edges in
a connected network, the network has the capacity $2/3$ and is thus not
solvable. This will prove part D as well as the sufficiency of part A
of the theorem. In the necessity part of the proof, we will show
that if a pair of edges satisfying conditions 1)-6) in part A does not exist,
then the network is linearly solvable over any field.
Parts B and C will be proved in parallel.
The reader may find it useful to keep Fig.~\ref{fig:non-solvable} in mind while going through the proof.

{{\bf Sufficiency:}}

We will show that the capacity of a connected $3s/3t$ sum-network
satisfying the hypothesis of Theorem ~\ref{theorem:mainres} is $2/3$ and
thus is not solvable.
It was proved in~\cite[Theorem 4]{RaiDS:itw10} using time-sharing arguments that the
coding capacity of any connected $3s/3t$ network
is at least $2/3$. Hence all we need to prove is that the capacity of a network
satisfying conditions $1-6$ of Theorem~\ref{theorem:mainres} is $\leq{2/3}$.
The idea of this proof is similar to that of \cite[Theorem 6]{RaiDS:itw10}.
Suppose there is a $(k,n)$ fractional coding solution for the network.
That is, the messages at the sources are $x_1,x_2,x_3 \in F^k$,
the terminals recover the sum $x_1+x_2+x_3 \in F^k$, and each edge in the
network carries an element from $F^n$. We allow non-linear coding.
Let the symbols transmitted over
$e_1$ and $e_2$ be denoted by $Y_{e_1}$ and $Y_{e_2}$ respectively.
Let us add an edge $e_2^*$ from $\head{e_2}$ to $t_3$ and an
edge $e_1^*$ from $\head{e_1}$ to $t_3$. Clearly this new network $\set{N}^*$
also satisfies the six conditions of Theorem~\ref{theorem:mainres}
and is stronger than $\set{N}$. We show that the capacity of $\set{N}^*$
itself is bounded by $2/3$.

Since $\set{N}^*$ is a connected sum-network and
satisfies the hypothesis of Theorem~\ref{theorem:mainres},  

1. Conditions $1,2 \Rightarrow \{s_1,s_3\} \rightarrow \tail{e_1}$ and
$\head{e_1} \rightarrow \{t_1,t_3\}$.

2. Condition $3,4 \Rightarrow \{s_2,s_3\} \rightarrow \tail{e_2}$ and
$\head{e_2} \rightarrow \{t_2,t_3\}$.

3. Conditions $1,3,6 \Rightarrow s_1 \nrightarrow \tail{e_2},s_2 \nrightarrow \tail{e_1}$.

By statement 3 above, $Y_{e_1}$
is only a function of $x_1$ and $x_3$, but not of $x_2$; and $Y_{e_2}$ is only a function of
$x_2$ and $x_3$, but not of $x_1$. Let us denote
them as $Y_{e_1} = \phi(x_1, x_3)$ and $Y_{e_2} = \psi(x_2, x_3)$.

\begin{claim}\label{claim:fone} (i) $\phi(x_1,x_3)$ is a 1-1 function
of $x_3$ for a fixed value of $x_1$ and a 1-1 function of
$x_1$ for a fixed value of $x_3$.
(ii) $\psi(x_2,x_3)$ is a 1-1 function
of $x_2$ for a fixed value of $x_3$ and a 1-1 function of $x_3$ for
a fixed value of $x_2$. 
\end{claim}
\begin{proof} We prove (i). The proof of (ii) is similar.

Since $t_1$ can recover $x_1+x_2+x_3$, for any fixed
values of $x_1$ and $x_2$, the set of messages received by the
terminal $t_1$ is a 1-1 function of $x_3$ as $x_1+x_2+x_3$
is a 1-1 function of $x_3$ for fixed $x_1$ and $x_2$.
But by condition $2$ of Theorem \ref{theorem:mainres},
all $(s_3,t_1)$ paths pass through $e_1$. Hence
$\phi(x_1,x_3)$ is a 1-1 function of $x_3$ for a fixed value of $x_1$.

Similarly, since $t_3$ can recover $x_1+x_2+x_3$, for any
fixed values of $x_2$ and $x_3$, the set of messages received by
the terminal $t_3$ is a 1-1 function of $x_1$ as $x_1+x_2+x_3$
is a 1-1 function of $x_1$ for fixed $x_2$ and $x_3$. But by condition
$1$ of Theorem \ref{theorem:mainres}, all $(s_1,t_3)$
paths pass through $e_1$. Hence $\phi(x_1,x_3)$
is a 1-1 function of $x_1$ for a fixed value of $x_3$.
\end{proof}

\begin{claim}\label{claim:recall} In $\set{N}^*$ the node $t_3$ can
 recover $x_1,x_2$ and $x_3$.
\end{claim}
\begin{proof} For a fixed $x_1$, $x_1+x_2+x_3$ is a $1-1$ function
of $x_2+x_3$. Since $t_2$ recovers $x_1+x_2+x_3$, by condition 4,
it implies that $\psi(x_2,x_3)$ is a $1-1$ function of $x_2+x_3$.
But since $t_3$ also gets $\psi(x_2,x_3)$ via
$e_2^*$, it can also recover $x_2+x_3$. Then by subtracting this from
$x_1+x_2+x_3$, $t_3$ can get $x_1$. Then using $x_1$ and $\phi(x_1,x_3)$, which it gets
via $e_1^*$ and which is a 1-1 function of $x_3$ for fixed $x_1$,
$t_3$ can recover $x_3$.
As $\psi(x_2,x_3)$ is a 1-1 function of $x_2$ for a fixed $x_3$,
$t_3$ can recover $x_2$. Hence $t_3$ can recover $x_1,x_2$ and $x_3$.
\end{proof}
 Now $(x_1, x_2, x_3)$ takes $|F|^{3k}$ possible values.
 On the other hand, by conditions $1,3$ and $5$
 of Theorem \ref{theorem:mainres}, $\{(e_1), (e_2)\}$ is a cut
 between $\{s_1,s_2,s_3\}$ and $t_3$ (even in $\set{N}^*$),
 and this cut can carry at most $|F|^{2n}$
 possible different message-pairs.
So $|F|^{2n} \geq |F|^{3k} \Rightarrow k/n \leq 2/3$. Thus
the capacity of $\set{N}^*$  and hence of $\set{N}$ is bounded by
$2/3$. As this rate is achievable in $\set{N}$, the capacity of
$\set{N}$ is exactly $2/3$.

{\bf Necessity:}

In this part, we will show that if a network does not satisfy the
conditions 1)-6) in part A of the theorem, then the network is
solvable. Parts B and C of the theorem will also be proved in parallel.
In light of Lemma~\ref{lemma:suffcon} and Lemma~\ref{lemma:converse1},
we only need to concern ourselves with networks $\net$ having $\kappa(\net)=3$
and $\set{C}\neq\emptyset$. In light of Lemma~\ref{lemma:5con}, we can also additionally
assume that $\net$ satisfies Assumption~\ref{assumption:simplify}.
Then $\net$ contains an edge $e_2$ satisfying conditions $3,4$ of Theorem~\ref{theorem:mainres}
for suitable labeling of the sources and the terminals. The desired result then follows
from Lemma~\ref{lemma:converse2}.

{\bf Proof of Theorem~\ref{theorem:mainres2}:}

The reader may find it useful to keep Fig.~\ref{fig:nonbinarysolvable}
in mind while going through the proof.

{\bf Sufficiency:}

Let $\set{N}$ satisfy the hypothesis of Theorem~\ref{theorem:mainres2}.

{\it Part 1: Non-solvability of $\set{N}$ over $F_2$.}

Let the symbols transmitted over
$e_1$ and $e_2$ be denoted by $Y_{e_1}$ and $Y_{e_2}$ respectively.
By arguments similar to those given in the proof of
Sufficiency of Theorem~\ref{theorem:mainres}, one can show that
$Y_{e_1}$ is a function of only $x_1$ and $x_3$ and
$Y_{e_2}$ is a function of only $x_2$ and $x_3$.
Let us call them $f(x_1,x_3)$ and $g(x_2,x_3)$ respectively.
\begin{claim}\label{claim:fone2} (i) $f(x_1,x_3)$ is a 1-1 function of $x_3$
for a fixed value of $x_1$ and a 1-1 function of
$x_1$ for a fixed value of $x_3$.
(ii) $g(x_2,x_3)$ is a 1-1 function
of $x_2$ for a fixed value of $x_3$ and a 1-1 function of $x_3$ for
a fixed value of $x_2$. 
\end{claim}
\textit{Proof:} 
The proof is the same as the one given for Claim~\ref{claim:fone}.

If $\set{N}$ is solvable over $F_2$,
then $f(x_1,x_3)$ is a function of $F_2\times F_2$ into $F_2$.
It is easy to verify that all such functions can be represented by
polynomials of the form
$\alpha x_1 + \beta x_3 + \gamma x_1x_3 + \delta$
for $\alpha, \beta, \gamma, \delta \in F_2$.
It is also easy to verify that the only such functions
that satisfy Claim~\ref{claim:fone2}(i) are of the form 
$x_1+x_3+\delta$ for $\delta \in F_2$. Hence w.l.o.g., we assume that 
$f=x_1+x_3$. By similar arguments, we assume
$g=x_2+x_3$.

But by conditions (1-3) of Theorem~\ref{theorem:mainres2}, $\{e_1,e_2\}$ is a cut between $\{s_1,s_2,s_3\}$
and $t_3$. So $t_3$ can obtain $x_1+x_2+x_3$ only if for
some $\alpha, \beta, \gamma, \delta \in F_2$,
$\alpha f(x_1,x_3) + \beta g(x_2,x_3) + \gamma f(x_1,x_3)g(x_2,x_3) + \delta= x_1+x_2+x_3$
$\Rightarrow \alpha(x_1+x_3)+\beta(x_2+x_3) + \gamma(x_1+x_3)(x_2+x_3)+\delta = x_1+x_2+x_3$.
Now, substituting $x_1=x_2=x_3=0$ in this equation gives $\delta=0$
while substituting $x_1=x_2=x_3=1$ gives $\delta=1$
 --- a contradiction since $1 \neq 0$ in $F_2$.
Hence $\set{N}$ is not solvable over $F_2$.

{\it Part 2: Solvability of $\set{N}$ over all other fields.}

For this part let $F$ be any field except $F_2$.

Since $e_1$ does not disconnect $(s_1,t_1)$ and $e_2$ does not disconnect $(s_2,t_2)$, let

a) $Q_1$ be a $(s_1,t_1)$ path not containing $e_1$,

b) $Q_2$ be a $(s_2,t_2)$ path not containing $e_2$,

c) $R_1$ be a $(s_1,t_2)$ path and 

d) $R_2$ be a $(s_2,t_1)$ path. 

In this case there exists a subnetwork shown in Fig.~\ref{fig:thm2a}, where the shaded circular region
means that $Q_1,Q_2,R_1,R_2$ may share edges. They do not share
edges with other parts of the network shown in the figure.
This is because,

(i) By condition $1,4$ of Theorem~\ref{theorem:mainres2}, $Q_1$
does not contain $e_2$ or any node from the $(s_3,\tail{e_1})$,
$(s_3,\tail{e_2})$, $(s_2,\tail{e_2})$, $(\head{e_1},t_3)$ or $(\head{e_2},t_3)$ path segments. It
also does not contain $e_1$ by definition. 

(ii)
Similarly  by condition $2,4$
of Theorem~\ref{theorem:mainres2}, $Q_2$
does not contain $e_1$ or any node from the $(s_3,\tail{e_1})$,
$(s_3,\tail{e_2})$, $(s_1,\tail{e_1})$, $(\head{e_1},t_3)$ or $(\head{e_2},t_3)$ path segments.
It also does not contain $e_2$ by definition. 

(iii)
By condition $4$ of Theorem~\ref{theorem:mainres2},
$R_1$ or $R_2$ does not contain both $e_1$ and $e_2$. Then by condition $1$,
$R_1$ can not contain $e_2$ or a node from the $(s_3,\tail{e_2})$
or $(s_2,\tail{e_2})$ or $(\head{e_2},t_3)$,
and by condition $2$ it can not contain $e_1$ or a node from $(s_3,\tail{e_1})$
or $(\head{e_1},t_3)$.
Similarly by condition $2$,
$R_2$ can not contain $e_1$ or a node from $(s_3,\tail{e_1})$
or $(s_1,\tail{e_1})$ or $(\head{e_1},t_3)$,
and by condition $1$ it can not contain $e_2$ or a
node from $(s_3,\tail{e_2})$ or $(\head{e_2},t_3)$.

\begin{figure}[h]
\centering
\subfigure[]{
\includegraphics[height=2.5in]{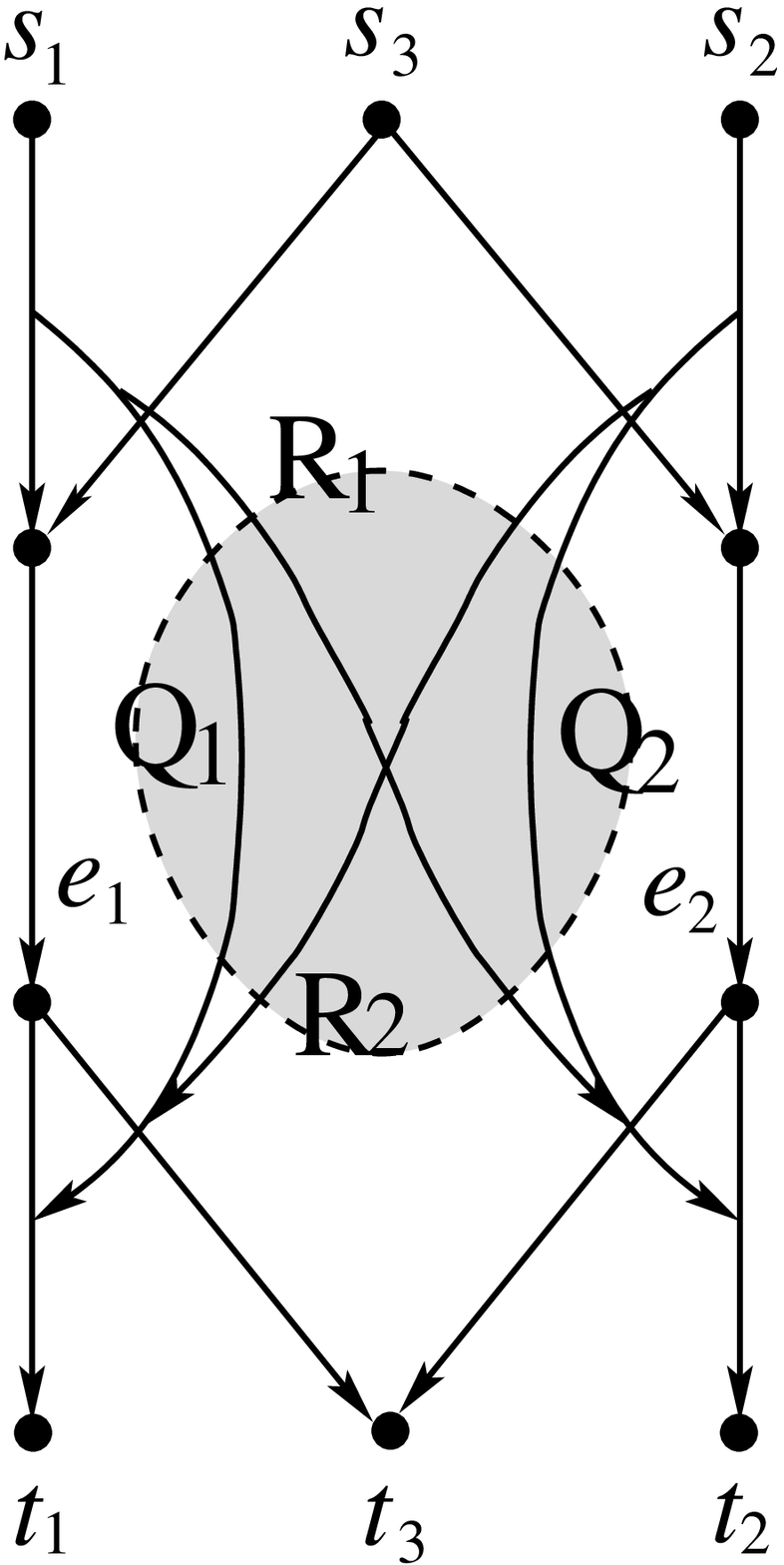}
\label{fig:thm2a}
}
\subfigure[]{
\includegraphics[height=2.5in]{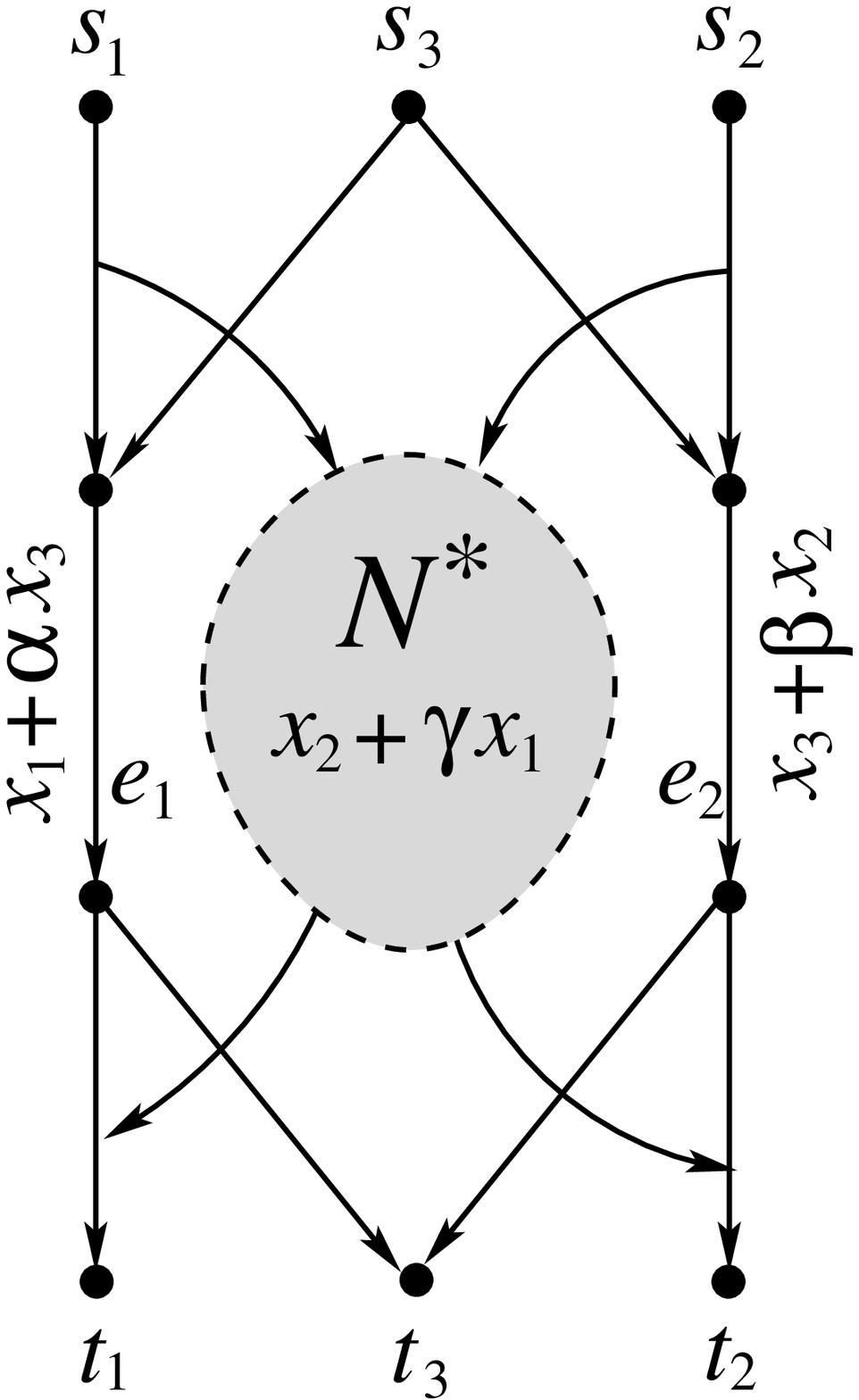}
\label{fig:thm2b}
}
\caption{The sub-network and the code over other fields}
\end{figure}

Now we give the coding scheme over any field $F\neq F_2$.
Let $\alpha \in F\backslash \{0,1\}$,
$\beta = (1-\alpha)^{-1}$ and $\gamma =1 - \alpha^{-1}$.
Consider the sub-network $\set{N}^*$ formed by considering all the nodes
of $\set{N}$,
but only those edges from $\set{N}$ belonging to the paths $Q_1,Q_2,R_1$ or $R_2$.
Due to the statements above, $\set{N}^*$ does not contain $e_1,e_2$ or edges from the $(s_3,\tail{e_1})$,
$(s_3,\tail{e_2})$, $(\head{e_1},t_3)$ or $(\head{e_2},t_3)$ path segments;
and further, in $\set{N}^*$,
$\{s_1,s_2\} \rightarrow \{t_1,t_2\}$. So using the edges in $\set{N}^*$,
and by pre-multiplying $x_1$ by $\gamma$ at $s_1$,
we can communicate $x_2+\gamma x_1$ to $t_1$ and $t_2$ by Lemma~\ref{lemma:R}.
Then in $\set{N}$ we can simultaneously transmit
$x_1+\alpha x_3$ on $e_1$ and $x_3+\beta x_2$ on $e_2$. This is shown in Fig.~\ref{fig:thm2b}.
By obtaining $x_1+\alpha x_3$
through $P(\head{e_1},t_3)$ and $x_3+\beta x_2$
through $P(\head{e_2},t_3)$, $t_3$ can
get $x_1+x_2+x_3 = (x_1+\alpha x_3) + \beta^{-1}(x_3+\beta x_2)$.
Using $x_2+\gamma x_1$ (received from $\set{N}^*$)
and $x_1+\alpha x_3$ (received from $e_1$),
$t_1$ can get $x_1+x_2+x_3 = (x_2+\gamma x_1)+\alpha^{-1}(x_1+\alpha x_3)$.
Similarly, $t_2$ can combine $x_2+\gamma x_1$ (received on
$\set{N}^*$) and $x_3+\beta x_2$ (received from $e_2$) to get
$x_1+x_2+x_3 = \gamma^{-1}(x_2+\gamma x_1) + (x_3+\beta x_2)$.

{\bf Necessity:}

We wish to show that networks which are not solvable over $F_2$ but solvable over all other fields
have an edge pair satisfying conditions 1-4 of Theorem~\ref{theorem:mainres2}.

From Lemma~\ref{lemma:suffcon} we see
that networks with
$\kappa =0,1,\geq 4$ are XOR solvable over all fields.
Lemma~\ref{lemma:converse2} shows that networks with $\kappa(\set{N}) =3$ and $\set{C} \neq \emptyset$
are either XOR
solvable over all fields or are nonsolvable.
Lemma~\ref{lemma:converse1} shows that a network with $\kappa(\set{N}) =3$ and $\set{C} = \emptyset$
either satisfies
the four conditions in Theorem~\ref{theorem:mainres2} and is nonsolvable over $F_2$ but solvable over other fields; or does not satisfy the conditions 
in Theorem~\ref{theorem:mainres2} and is XOR solvable over all fields.
The proof of Lemma~\ref{lemma:suffcon} for networks  with $\kappa =2$
shows that networks with $\kappa=2$ which are not solvable over $F_2$ but solvable over all other fields (some networks in {\it Case 1.2})
have an edge pair satisfying conditions 1-4 of Theorem~\ref{theorem:mainres2}.
Thus the necessity of Theorem~\ref{theorem:mainres2} holds for all $3s/3t$ networks. 

\section{Conclusion}\label{section:conclusion}
We presented a set of necessary and sufficient conditions for solvability
of a $3$-source $3$-terminal sum-network over any field $F$.
The conditions are the
same for all fields except $F_2$. This explains the existence of
the networks in Fig.~\ref{fig:nonbinarysolvable} which are not solvable
over $F_2$ though they are solvable over any other field. The conditions present
full insight into the case of $3$-sources and $3$-terminals - the smallest
sum-networks with non-trivial characterization.

The complexity of the proofs for this very specific case makes it
clear that stronger tools are needed to characterize the problem
for higher number of sources and terminals. However, this is not
surprising, considering that sum-networks have been proved~ \cite{RaiD:09c}
to be
equivalent to the multiple-unicast networks as a class of problems.
Even for multiple-unicast networks, explicit
characterization of solvable networks is not available. Except for
the double-unicast problem~\cite{wang2,shenvi2}, only cut based necessary conditions~\cite{yan1}
are known to the best of our knowledge.
It is fair to expect that tools developed to analyze/characterize
will have strong relation with each other for these two classes of
problems.
\section{Acknowledgement}\label{section:acknowledgement}
This work was supported in part by Bharti Centre for Communication at
IIT Bombay and a project from the Department of Science and Technology
(DST), India.


\begin{thebibliography}{10}

\bibitem{ahlswede1}
R.~Ahlswede, N.~Cai, S.-Y.~R. Li, and R.~W. Yeung, ``Network information
  flow,'' {\em IEEE Trans. Inform. Theory}, vol.~46, no.~4, pp.~1204--1216,
  2000.

\bibitem{li1}
S.-Y.~R. Li, R.~W. Yeung, and N.~Cai, ``Linear network coding,'' {\em IEEE
  Trans. Inform. Theory}, vol.~49, no.~2, pp.~371--381, 2003.

\bibitem{koetter2}
R.~Koetter and M.~M\'{e}dard, ``An algebraic approach to network coding,'' {\em
  IEEE/ACM Transactions on Networking}, vol.~11, no.~5, pp.~782--795, 2003.

\bibitem{jaggi1}
S.~Jaggi, P.~Sanders, P.~A. Chou, M.~Effros, S.~Egner, K.~Jain, and
  L.~Tolhuizen, ``Polynomial time algorithms for multicast network code
  construction,'' {\em IEEE Trans. Inform. Theory}, vol.~51, no.~6,
  pp.~1973--1982, 2005.

\bibitem{ho1}
T.~Ho, R.~Koetter, M.~M\'{e}dard, M.~Effros, J.~Shi, and D.~Karger, ``A random
  linear network coding approach to multicast,'' {\em IEEE Trans. Inform.
  Theory}, vol.~52, no.~10, pp.~4413--4430, 2006.

\bibitem{yeung1}
R.~Yeung, {\em Information Theory and Network Coding}.
\newblock Springer, 2008.

\bibitem{korner1}
J.~Korner and K.~Marton, ``How to encode the modulo-two sum of binary
  sources,'' {\em IEEE Trans. Inform. Theory}, vol.~25, no.~2, pp.~219--221,
  1979.

\bibitem{han1}
T.~S. Han and K.~Kobayashi, ``A dichotomy of functions $f(x,y)$ of correlated
  sources $(x,y)$,'' {\em IEEE Trans. Inform. Theory}, vol.~33, no.~1,
  pp.~69--86, 1987.

\bibitem{krithivasan1}
D.~Krithivasan and S.~S. Pradhan, ``An achievable rate region for distributed
  source coding with reconstruction of an arbitrary function of the sources,''
  in {\em Proceedings of IEEE International Symposium on Information Theory},
  (Toronto, Canada), pp.~56--60, 2008.

\bibitem{doshi1}
V.~Doshi, D.~Shah, M.~M\'{e}dard, and S.~Jaggi, ``Distributed functional
  compression through graph coloring,'' in {\em Proceedings of Data compression
  Conference}, pp.~93--102, 2007.

\bibitem{orlitsky1}
A.~Orlitsky and J.~R. Roche, ``Coding for computing,'' {\em IEEE Trans. Inform.
  Theory}, vol.~47, no.~3, pp.~903--917, 2001.

\bibitem{feng1}
H.~Feng, M.~Effros, and S.~A. Savari, ``Functional source coding for networks
  with receiver side information,'' in {\em Proceedings of the Allerton
  Conference on Communication, Control, and Computing}, September 2004.

\bibitem{gallager2}
R.~G. Gallager, ``Finding parity in a simple broadcast network,'' {\em IEEE
  Trans. Inform. Theory}, vol.~34, pp.~176--180, 1988.

\bibitem{giridhar}
A.~Giridhar and P.~R. Kumar, ``Computing and communicating functions over
  sensor networks,'' {\em IEEE J. Select. Areas Commun.}, vol.~23, no.~4,
  pp.~755--764, 2005.

\bibitem{kanoria}
Y.~Kanoria and D.~Manjunath, ``On distributed computation in noisy random
  planar networks,'' in {\em Proceedings of ISIT, Nice, France}, 2008.

\bibitem{boyd}
S.~Boyd, A.~Ghosh, B.~Prabhaar, and D.~Shah, ``Gossip algorithms: design,
  analysis and applications,'' in {\em Proceedings of IEEE INFOCOM},
  pp.~1653--1664, 2005.

\bibitem{varshney1}
P.~K. Varshney, {\em Distributed detection and data fusion}.
\newblock Springer-Verlag New York, Inc., 1996.

\bibitem{chair1}
Z.~Chair and P.~K. Varshney, ``Optimal data fusion in multiple sensor detection
  systems,'' {\em IEEE Trans. Aerosp. Electron. Syst.}, vol.~22, no.~1,
  pp.~98--101, 1986.

\bibitem{hu1}
J.~Hu and R.~S. Blum, ``On the optimality of finite-level quantizations for
  distributed signal detection,'' {\em IEEE Trans. Inform. Theory}, vol.~47,
  no.~4, pp.~1665--1671, 2001.

\bibitem{ramamoorthy}
A.~Ramamoorthy, ``Communicating the sum of sources over a network,'' in {\em
  Proceedings of ISIT, Toronto, Canada, July 06-11}, pp.~1646--1650, 2008.

\bibitem{RaiD:09a}
B.~K. Rai, B.~K. Dey, and A.~Karandikar, ``Some results on communicating the
  sum of sources over a network,'' in {\em Proceedings of NetCod 2009}, 2009.

\bibitem{RaiD:09b}
B.~K. Rai and B.~K. Dey, ``Feasible alphabets for communicating the sum of
  sources over a network,'' in {\em Proceedings of IEEE International Symposium
  on Information Theory}, (Seoul, Korea), 2009.

\bibitem{langberg3}
M.~Langberg and A.~Ramamoorthy, ``Communicating the sum of sources in a
  3-sources/3-terminals network,'' in {\em Proceedings of IEEE International
  Symposium on Information Theory}, (Seoul, Korea), 2009.

\bibitem{RaiDS:itw10}
B.~K. Rai, B.~K. Dey, and S.~Shenvi, ``Some bounds on the capacity of
  communicating the sum of sources,'' in {\em Proceedings of IEEE Information
  Theory Workshop}, 2010.

\bibitem{appuswamy1}
R.~Appuswamy, M.~Franceschetti, N.~Karamchandani, and K.~Zeger, ``Network
  computing capacity for the reverse butterfly network,'' in {\em Proceedings
  of IEEE International Symposium on Information Theory}, (Seoul, Korea), 2009.

\bibitem{appuswamy2}
R.~Appuswamy, M.~Franceschetti, N.~Karamchandani, and K.~Zeger, ``Network
  coding for computing,'' in {\em Proceedings of Annual Allerton Conference},
  (UIUC, IIlinois, USA), 2008.

\bibitem{appuswamy3}
R.~Appuswamy, M.~Franceschetti, N.~Karamchandani, and K.~Zeger, ``Network
  coding for computing part i : Cut-set bounds,'' {\em available at
  http://arxiv.org/abs/0912.2820}.

\bibitem{RaiD:09c}
B.~K. Rai and B.~K. Dey, ``Sum-networks: system of polynomial equations,
  reversibility, insufficiency of linear network coding, unachievability of
  coding capacity,'' {\em available at http://arxiv.org/abs/0906.0695}.

\bibitem{koetter3}
R.~Koetter, M.~Effros, T.~Ho, and M.~M\'{e}dard, ``Network codes as codes on
  graphs,'' in {\em Proceedings of the 38th annual conference on information
  sciences and systems (CISS)}, 2004.

\bibitem{dougherty3}
R.~Dougherty and K.~Zeger, ``Nonreversibility and equivalent constructions of
  multiple-unicast networks,'' {\em IEEE Trans. Inform. Theory}, vol.~52,
  no.~11, pp.~5067--5077, 2006.

\bibitem{riis1}
S.~Riis, ``Reversible and irreversible information networks,'' in {\em IEEE
  Trans. Inform. Theory}, no.~11, pp.~4339--4349, 2007.

\bibitem{ho2}
T.~Ho and D.~Lun, {\em Network Coding: An Introduction}.
\newblock Cambridge, U.K.: Cambridge University Press, 2008.

\bibitem{shenvi2}
S.~Shenvi and B.~K. Dey, ``A simple necessary and sufficient condition for the
  double unicast problem,'' in {\em IEEE International Conference on
  Communications (ICC)}, (Cape Town, South Africa), May 2010.

\bibitem{wang2}
C.~C. Wang and N.~B. Shroff, ``Beyond the butterfly -- a graph-theoretic
  characterization of the feasibility of network coding with two simple unicast
  sessions,'' in {\em Proceedings of IEEE International Symposium on
  Information Theory}, 2007.

\bibitem{yan1}
X.~Yan, J.~Yang, and Z.~Zhang, ``An outer bound for multisource multisink
  network coding with minimum cost consideration,'' {\em IEEE Trans. Inform.
  Theory}, vol.~52, no.~6, pp.~2373--2385, 2006.

\end{thebibliography}
\end{document}